\documentclass[a4paper,allowtoday,onecolumn,accepted=2024-06-22,allowfontchangeintitle]{quantumarticle} 
\pdfoutput=1 
\usepackage{dsfont} 	
\usepackage{microtype} 	
\usepackage[numbers,sort&compress]{natbib} 	
\usepackage[utf8]{inputenc} 	
\usepackage[T1]{fontenc} 	
\usepackage[british]{babel} 	
\usepackage[dvipsnames,x11names]{xcolor} 	
\usepackage{amsmath,amssymb,amsthm,bm,amsfonts,bbm} 	
\usepackage{subfigure}
\usepackage{mathtools} 	
\usepackage[colorlinks=true,citecolor=green,linkcolor=Purple,urlcolor=Magenta]{hyperref} 
\usepackage{verbatim} 
\usepackage{empheq}
\usepackage{hhline}
\usepackage{colortbl}
\usepackage[theorems, breakable]{tcolorbox}
\usepackage{placeins}
\usepackage{cancel}

\newcommand*\widefbox[1]{\fbox{\hspace{2em}#1\hspace{2em}}}

\newcommand{\ket}[1]{\vert#1\rangle}

\newcommand{\braket}[2]{\langle#1\vert#2\rangle}
\newcommand{\ketbra}[2]{\vert #1 \rangle \hspace{-.4mm} \langle #2 \vert}

\newcommand{\sbt}{\,\begin{picture}(-1,1)(-1,-3)\circle*{3}\end{picture}\ } 

\newcommand{\id}{\mathds{1}}

\DeclareMathOperator{\tr}{tr}

\newcommand{\ie}{\textit{i.e.}}

\renewcommand{\H}{\mathcal{H}}

\renewcommand{\L}{\mathcal{L}}

\newcommand{\Scal}{\mathcal{S}}

\newcommand{\inp}{\textup{\texttt{i}}}
\newcommand{\out}{\textup{\texttt{o}}}
\newcommand{\att}{\textup{\texttt{a}}}

\newcommand{\xtt}{\textup{\texttt{x}}}
\newcommand{\xt}{\texttt{x}}
\newcommand{\yt}{\texttt{y}}
\newcommand{\zt}{\texttt{z}}

\makeatletter 
\newcommand{\doublewidetilde}[1]{{%
  \mathpalette\double@widetilde{#1}%
}}
\newcommand{\double@widetilde}[2]{%
  \sbox\z@{$\m@th#1\widetilde{#2}$}%
  \ht\z@=.9\ht\z@
  \widetilde{\box\z@}%
}
\newcommand{\map}[1]{\widetilde{#1}}  

\makeatother

\newtheorem{definition}{Definition}
\newtheorem{theorem}{Theorem}
\newtheorem{lemma}{Lemma}

\theoremstyle{definition}
\newtheorem{example}{Example}
\newtheorem{observation}{Observation}

\definecolor{brow}{rgb}{0.8515625,0.67578125,0.5234375} 
\definecolor{blu}{rgb}{0.39453125,0.6171875,0.734375} 

\begin{document}
\setlength\fboxrule{1pt}
\title{Characterising transformations between quantum objects, `completeness' of quantum properties, and transformations without a fixed causal order}

\author{Simon Milz}
\orcid{0000-0002-6987-5513}
\email{MILZS@tcd.ie}
\affiliation{School of Physics, Trinity College Dublin, Dublin 2, Ireland}
\affiliation{Trinity Quantum Alliance, Unit 16, Trinity Technology and Enterprise Centre, Pearse Street, Dublin 2, D02YN67, Ireland}
\affiliation{Institute for Quantum Optics and Quantum Information, Austrian Academy of Sciences, Boltzmanngasse 3, 1090 Vienna, Austria}
\affiliation{Faculty of Physics, University of Vienna, Boltzmanngasse 5, 1090 Vienna, Austria}

\author{Marco T\'ulio Quintino}
\orcid{0000-0003-1332-3477}
\email{Marco.Quintino@lip6.fr}
\affiliation{{Sorbonne Universit\' {e}, CNRS, LIP6, F-75005 Paris, France}}

\date{15th July 2024}

\begin{abstract}
Many fundamental and key objects in quantum mechanics are linear mappings between particular affine/linear spaces. This structure includes basic quantum elements such as states, measurements, channels, instruments, non-signalling channels and channels with memory, and also higher-order operations such as superchannels, quantum combs, n-time processes, testers, and process matrices which may not respect a definite causal order. Deducing and characterising their structural properties in terms of linear and semidefinite constraints is not only of foundational relevance, but plays an important role in enabling the numerical optimisation over sets of quantum objects and allowing simpler connections between different concepts and objects. Here, we provide a general framework to deduce these properties in a direct and easy to use way. While primarily guided by practical quantum mechanical considerations, we also extend our analysis to mappings between \textit{general} linear/affine spaces and derive their properties, opening the possibility for analysing sets which are not explicitly forbidden by quantum theory, but are still not much explored. Together, these results yield versatile and readily applicable tools for all tasks that require the characterisation of linear transformations, in quantum mechanics and beyond. As an application of our methods, we discuss how the existence of indefinite causality naturally emerges in higher-order quantum transformations and provide a simple strategy for the characterisation of mappings that have to preserve properties in a `complete' sense, i.e., when acting non-trivially only on parts of an input space.
\end{abstract}
\maketitle

\clearpage 
\tableofcontents
\clearpage 

\section{Introduction}
Many fundamental objects in quantum mechanics can, at their most basic level, be understood as (linear) transformations of other basic objects. For example, measurements are transformations of states to probabilities, while quantum channels are transformations of quantum states to quantum states. This simple understanding of quantum objects as transformations can straightforwardly be extended, leading to a whole host of \textit{higher-order} transformations (see Fig.~\ref{fig::chain} for an example). To name but a few, transformations of channels to channels (yielding so-called superchannels~\cite{chiribella08quantum_supermaps}), sequences of quantum channels with memory to quantum channels (yielding so-called quantum combs~\cite{chiribella07circuit_architecture} and quantum strategies~\cite{gutoski07quantum_games}), a sequence of channels to states (yielding so-called multi-time processes~\cite{pollock15markovian}) and collections of channels to probabilities (yielding testers~\cite{chiribella09networks,ziman08_process_POVM,bavaresco21} and process matrices~\cite{oreshkov11}) have been devised in recent years, each with their own respective physical motivation. On the other hand, such higher-order transformations can equivalently be motivated as the correct descriptor in many physical situations, where states, measurements and channels alone would prove to be insufficient practical tools. Consequently, they have, amongst others, found active use in the fields of open quantum system dynamics~\cite{pollock15markovian}, quantum circuit architecture~\cite{chiribella07circuit_architecture}, the investigation of channels with memory~\cite{kretschmann05channels_memory}, as well as the study of causal indefiniteness~\cite{oreshkov11} and the dynamics of causal order~\cite{castro-ruiz_dynamics_2018}. 
\begin{figure}[ht!]
    \centering
    \includegraphics[width = 0.96\linewidth]{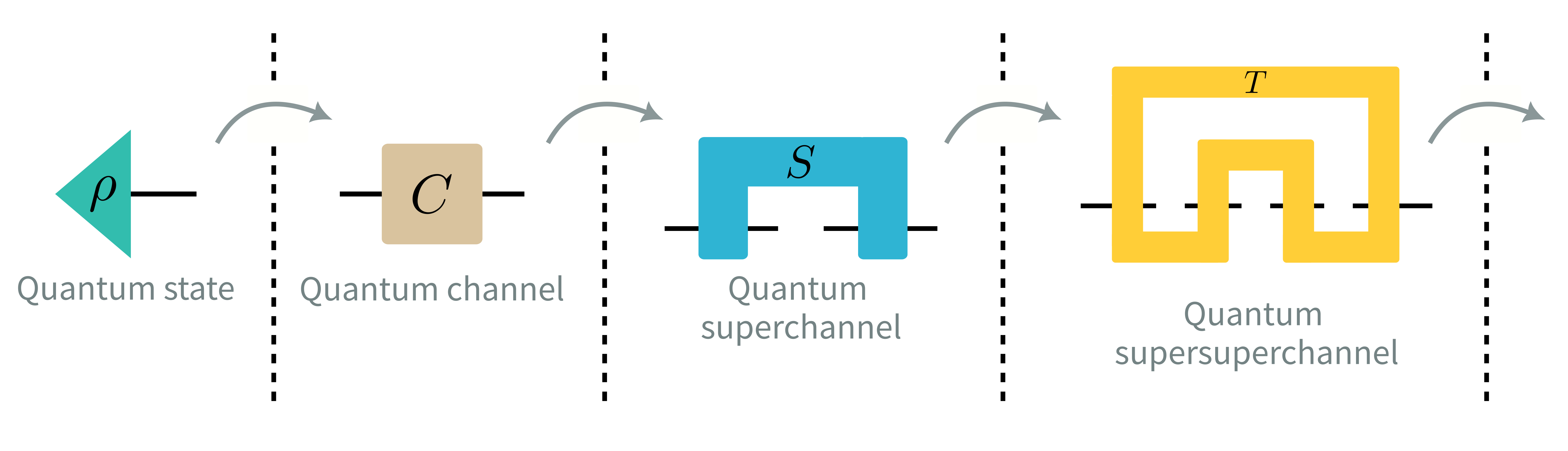}
\caption{\textbf{Chain of higher-order operations.} A pictorial illustration of an exemplary chain of higher-order operations. Channels describe transformations between states, superchannels describe transformation between channels, and supersuperchannels describe transformations between superchannels. The methods we provide allow for the characterization of all mappings in this chain, but more generally all mappings between \textit{arbitrary} quantum objects. The precise definition and a characterisation of such objects and the mappings between them can be found in Sec.~\ref{sec::Emergence}.}
\label{fig::chain}
\end{figure}

Independent of the respective concrete motivation, in any of these investigations it is, as a first step, frequently necessary to deduce the structural properties of the considered transformations, i.e., a characterisation of the transformation in a chosen representation\footnote{Here, and throughout, unless stated otherwise, we mean a characterisation of the \textit{Choi matrix} of the transformation when we consider its `properties'.} that goes beyond its original definition. For example, for the case of process matrices, one is interested in the structure of mappings that map pairs of independent quantum channels (or, equivalently, any two-party non-signalling channel) to unit probability, in order to analyse the set of processes that abide by causality locally, but not necessarily globally~\cite{oreshkov11}. Having their at hand then does not only allow one to deduce that this latter set fundamentally differs from the set of causally ordered processes, but also enables numerical optimisation over causally indefinite processes. On the more axiomatic side, recent works have discussed the properties of the hierarchy of transformations that emerges from starting at quantum states and 'going up the ladder' of transformations, i.e., transformations of states, transformations of transformations of states, etc.~\cite{perinotti16higher, bisio_theoretical_2019, simmons_higher-order_2022, hoffreumon_projective_2022}. 

Here, we stay agnostic with respect to the origin of the respective transformations and provide a general framework to answer the question: What are the properties of linear transformations between affine/linear spaces? For example, this question directly captures the case of quantum channels -- completely positive mappings of the affine space of matrices with unit trace onto itself -- but also \textit{all} conceivable transformations between quantum objects alluded to before. Concretely, we phrase the properties of the respective spaces as well as the transformations between them in terms of linear projectors and employ their representation as matrices via the Choi isomorphism for their explicit characterisation. The 'projector approach' and the methods presented here may be viewed as deeper analysis and a generalisation of the ideas first presented in Ref.~\cite{araujo15witnessing}, which were developed to present a basis-independent characterisation for process matrices and to study quantum processes which do not respect a definite causal order.
A similar approach to the characterisation of higher-order maps has, for example, been taken in~\cite{castro-ruiz_dynamics_2018,milz_resource_2022,hoffreumon_projective_2022,simmons_higher-order_2022}. Moreover, in-depth investigations into the structure and type theory of the hierarchy of conceivable higher-order quantum maps can be found in Refs.~\cite{perinotti16higher,bisio_theoretical_2019,hoffreumon_projective_2022,simmons_higher-order_2022,apadula22nosignalling}, where, in particular Ref.~\cite{hoffreumon_projective_2022} not only employs a similar approach to ours, but also provides a detailed analysis of the logical and set-theoretic structure of the projectors that define the properties of the transformations we analyse. 

This current work is more modest in scope; leveraging the link product~\cite{chiribella09networks} and the linearity of the involved transformations, we provide a straightforward, systematic way to derive the properties of arbitrary transformations between sets that naturally emerge in quantum mechanics. In turn, this allows us to re-derive the properties of a vast array of relevant quantum transformations appearing in different parts of the literature in a unified and direct way. We further demonstrate the effectiveness and ease of this framework by explicitly discussing affine dual sets as well as probabilistic quantum operations and the signalling structures of quantum maps. As an additional application of our methods, we analyse the emergence of indefinite causal order in higher-order quantum operations.

In many cases of interest, a transformation does not only have to preserve certain properties (like, for example, the trace of a state) when acting entirely on the input space, but also when acting non-trivially on only a part of it. This is the reasoning behind the introduction of \emph{complete} positivity (in contrast to mere positivity), which requires not only positivity of a map $\map{C}$, but also positivity of all of its trivial extensions $\map{C} \otimes \map{\id}$, where $\map{\id}$ is the identity map on an arbitrary anciliary space. While most commonly encountered in the context of positivity and complete positivity, the analogous question can equivalently be raised for other properties of quantum transformations. For example, it is obvious that $\map{C} \otimes \map{\id}$ is trace preserving if $\map{C}$ is, but it is a priori unclear if `completeness' of properties is automatically implied for more general properties. Even more so, for many relevant cases, it not even immediate or unique, how a specific property should be extended to a larger input space. Parts of this question have -- in different contexts -- been considered in the literature, either for concrete questions at hand~\cite{araujo16purification,burniston_necessary_2020, milz_resource_2022}, or for specific types of extensions~\cite{simmons_higher-order_2022, hoffreumon_projective_2022,Wilson22Locality,Wilson23linearity}. Here, using the characterisation of quantum transformations, we provide a simple strategy to derive the structure of transformations that have to preserve a structural property in a `complete' sense. In particular, this approach is agnostic to the chosen extension, and allows for the identification of sufficient conditions under which a transformation is automatically completely admissible. The versatility and applicability of this approach is then demonstrated by means of concrete examples, that not only recover known cases from the literature, but also demonstrate how the requirements imposed by `completeness' fundamentally depend on the chosen extension of a property.

Finally, owing to the simplicity of our approach, we are also able to drop the assumptions generally fulfilled in quantum mechanics -- like, for example, the self-adjointness of the involved projectors, or the membership of the identity matrix to the respective input and output spaces -- and derive the properties of transformation between \textit{general} spaces, a result that might be of interest in its own right. Together, our results provide a unified framework to discuss and characterise \textit{all} (quantum) objects and linear transformations thereof as well as the restrictions imposed by `completeness', thus offering a versatile tool for a wide array of problems that naturally occur in the field of quantum mechanics in beyond.  

\section{Warming up: quantum states and quantum channels}
\label{sec::QstatesQchannels}
The fundamental set of objects that quantum mechanics is concerned with are quantum states $\rho \in \L(\H)$, unit trace (i.e., $\tr[\rho]=1$), positive semidefinite (i.e., $\rho \geq 0$) linear operators acting on a Hilbert space $\H$. Here, and throughout, we consider $\H$ to be finite dimensional, such that $\H \cong \mathbb{C}^d$ for some $d \in \mathbb{N}$. 

Transformations of quantum states are then described by \textit{linear maps} $\map{T}:\L(\H_\inp)\to\L(\H_\out)$, where we adopt the convention of referring to the input space as $\H_\inp$, the output space as $\H_\out$ and maps between operators with a tilde. For better bookkeeping, we always explicitly distinguish between input and output spaces, even if $\H_\inp \cong \H_\out$.\footnote{The only exception to this rule will be the projectors $\widetilde P$ that we introduce below.} A transformation $\map{T}$ is a valid \textit{quantum channel}, i.e., it  represents a deterministic transformation between quantum states that can be physically implemented, if it is a completely positive (CP)\footnote{A linear map $\map{T}:\L(\H_\inp)\to\L(\H_\out)$ is positive when $\map{T}[\rho]\geq0$ for every positive semidefinite linear operator $ \rho\geq0, \rho\in \L(\H_\inp)$. A map is CP if it is positive for all trivial extensions, that is, $\map{T}\otimes\map{\id}(\sigma)\geq0$ for every positive semidefinite $ {\sigma\in\L(\H_\inp\otimes\H_\text{aux})}$ where $\map{\id}:\L(\H_\text{aux})\to\L(\H_\text{aux})$ is the identity map on an arbitrary finite space $\H_\text{aux}$, i.e., $\map{\id}[\rho]=\rho, \forall \rho\in\L(\H_\text{aux})$.} and trace-preserving\footnote{A linear map $\map{T}:\L(\H_\inp)\to\L(\H_\out)$ is TP when $\tr[\map{T}[\rho]]=\tr[\rho]$ for every linear operator $\rho\in \L(\H_\inp)$.} (TP) linear map.

While the mathematical characterisation of matrices that represent quantum states $\rho$ is clear, it is, a priori, unclear, what the corresponding properties of a representation of a quantum \textit{channel} -- a transformation between sets of quantum states -- are.  Here, we aim to provide a simple way of characterising mappings between objects that routinely occur in quantum mechanics. To illustrate the general concept, we first provide an answer for the well-known case of CPTP maps. To this end, we exploit the fact that linear maps admit a convenient representation as linear operators via the \textit{Choi-Jamio{\l}kowski isomorphism} (CJI)~\cite{pills67, jamiolkowski72, choi75}: Let $\{\ket{j}\}_j$ be the canonical computational basis for $\H_\inp$. The \textit{Choi operator/matrix} $T\in\L(\H_\inp\otimes\H_\out)$ of a linear map $\map{T}:\L(\H_\inp)\to\L(\H_\out)$ is then defined as
\begin{equation}
\label{eqn::Choi}
    T:= \sum_{jk} \ketbra{j}{k}\otimes \map{T}\big[\ketbra{j}{k}\big].
\end{equation}
Direct calculation shows that the action of $\map T$ can be written in terms of its Choi matrix $T$ as
\begin{equation}
\label{eqn::ChoiDef}
    \map{T}[\rho]=\tr_{\inp}[(\rho^\tau \otimes \id_\out) \; T]
\end{equation}
where $\rho^\tau$ is the transpose of $\rho$ in the computational basis and $\tr_\inp$ is the partial trace over $\H_\inp$.
\begin{figure}[t]
    \centering
    \subfigure[~]
    {
    \includegraphics[scale=1.15]{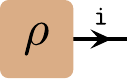}
    \label{fig::State}} \hspace{2cm}
    \subfigure[~]
    {
    \includegraphics[scale=1.15]{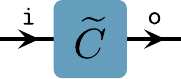}
    \label{fig::Channel}} 
    \caption{\textbf{Quantum states and Quantum Channels.} (a) Quantum state $\rho \in \L(\H_\inp)$. (b) Quantum channel $\map{C}: \L(\H_\inp) \rightarrow \L(\H_\out)$ that maps quantum states in $\L(\H_\inp)$ onto quantum states in $\L(\H_\out)$. Throughout, lines are labelled by the Hilbert space/space of linear operators they correspond to.}
\end{figure}

To characterise the properties of $T$, we note that a linear map $\map{T}:\L(\H_\inp)\to\L(\H_\out)$ is TP if and only if, $\tr_\out[T]=\id_\inp$ and CP if and only if $T\geq0$. Hence, CPTP maps (quantum channels) $\map{C}$ are described by a Choi matrix $C\in\L(\H_\inp \otimes\H_\out)$ that satisfies
\begin{align}
    C\geq& 0 \quad \text{and} \quad \tr_\out[C]=\id_\inp.
\end{align}
    In anticipation of later considerations, we can phrase this equivalently as
\begin{align}
    C\geq& 0 \label{eq:pos1}\\
   {}_\out C=&_{\inp\out}C \label{eq:lin1} \\
    \tr[C]=&d_\inp, \label{eq:aff1}
\end{align}
where $_\xt C:=\tr_\xt[C] \otimes \frac{\id_\xt}{d_\xt}$ is the trace-and-replace map and $d_\xt$ is the dimension of $\H_\xt$. 
Notice that, for consistency, one should keep track of the ordering of the operators, for instance, if $C\in\L(\H_\inp\otimes\H_\out)$, $_\inp C=\frac{\id_\inp}{d_\inp}\otimes\tr_\inp[C]$ and $_\out C=\tr_\out[C]\otimes \frac{\id_\out}{d_\out}$. Whenever there is risk of ambiguity, or we desire to emphasise some property, we will use subscripts (rather than explicit ordering) to indicate what space an object is defined/acts on. 

The characterisation of quantum channels given by Eqs.~\eqref{eq:pos1}--~\eqref{eq:aff1} has an interesting structure, which will be the starting point to analyse the structure of more general transformations throughout this work. Eq.~\eqref{eq:pos1} is a positivity constraint, Eq.~\eqref{eq:lin1} is a linear constraint, and Eq.~\eqref{eq:aff1} is an affine constraint. Consequently, the set of linear operators satisfying Eq.~\eqref{eq:lin1} form a linear subspace of $\L(\H_\inp \otimes\H_\out)$, and can thus be described by a projective map $\map{P}:\L(\H_\inp \otimes\H_\out)\to\L(\H_\inp\otimes\H_\out)$, where
\begin{equation}
\label{eqn::ChannProj}
   {}_\out C={}_{\inp\out}C  \iff C=\map{P}[C] \quad \text{with} \quad \map{P}[C]:=C-{}_\out C+{}_{\inp\out}C\, .
\end{equation}
We can easily verify that $\map{P}$ is indeed a projective map, that is, ${\map{P}}^2 := \map{P}\circ\map{P} = \widetilde P$. 
Putting everything together, an operator $C\in\L(\H_1\otimes\H_2)$ is the Choi operator of a quantum channel if and only if
\begin{alignat}{2}
\label{eqn::1}
     C&\geq 0, \quad \qquad&  &\text{(Positive semidefinite)}\\
     \label{eqn::2}
    C&=\map{P}[C], &  &\text{(Linear subspace)} \\
     \label{eqn::3}
   \tr[C]&=d_\inp, &  &\text{(Fixed trace)}.
\end{alignat}

Put differently, besides the positivity and overall trace constraint, the set of quantum channels is fully defined by the projector $\map{P}$.

While positivity of the Choi matrix simply follows from the requirement of complete positivity for the map (a property that we will assume throughout), for more general mappings, working out both the trace constraint and the correct projector can be somewhat cumbersome. For example, it is a priori unclear what properties a mapping from quantum channels to quantum channels (a so-called supermap~\cite{chiribella08quantum_supermaps} or superchannel) would possess, and similar for any mapping `higher up in the hierarchy'. Below, we extend the above  concepts and methods to provide a direct and systematic way to derive the linear and affine constraints for transformations between general quantum objects (see Fig.~\ref{fig::gen_setup} for a graphical depiction).
\begin{figure}[ht!]
    \centering
    \includegraphics[width = 0.96\linewidth]{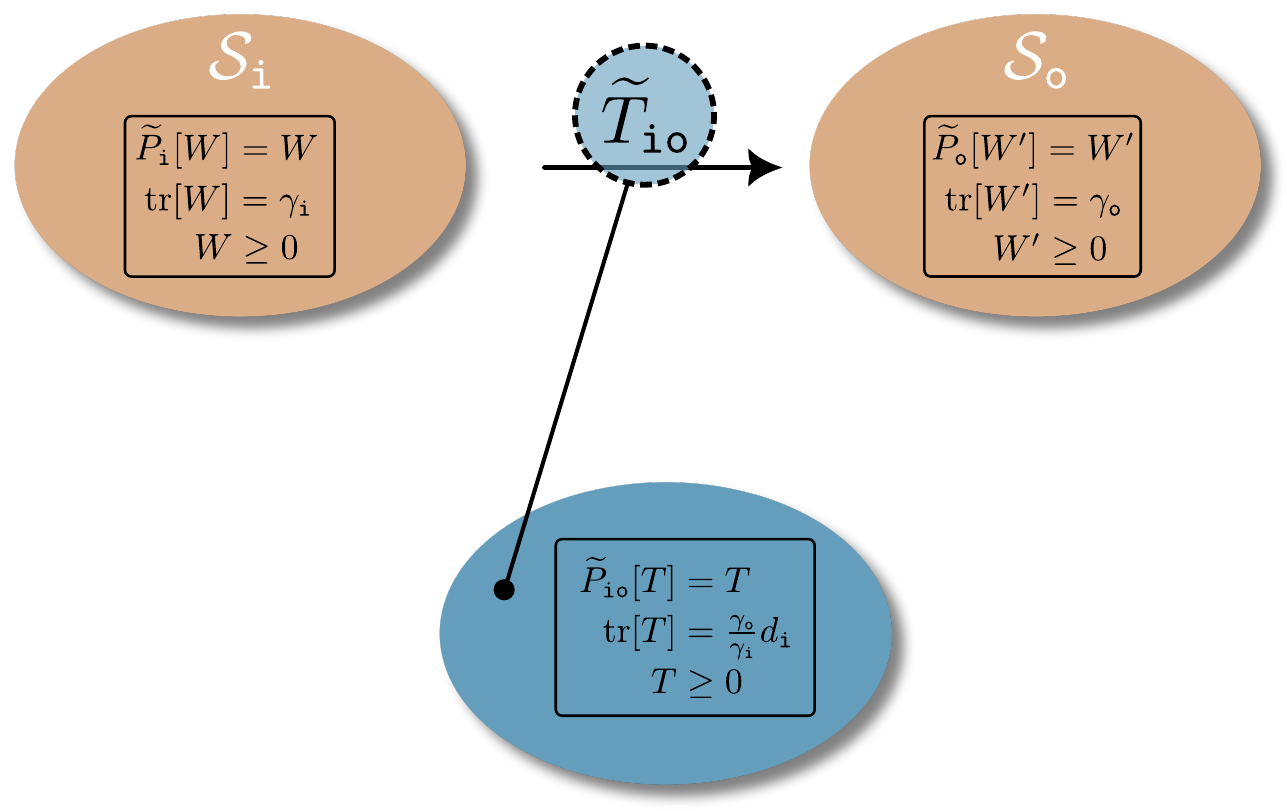}
    \caption{\textbf{General Transformations between affine (quantum) sets} All sets $\Scal_\inp$ and $\Scal_\out$ we consider are defined by a linear constraint (given by $\map{P}_\inp$ and $\map{P}_\out$, respectively) and in most cases an affine and a positivity constraint. Our aim is to characterise the set of linear transformations $\map{T}_{\inp\out}$ between them. Predominantly, this characterisation will be carried out for the Choi matrices $T$ of said transformations, and as it turns out, the corresponding set of matrices is, again, defined by a projector $\map{P}_{\inp\out}$, as well as an affine and a positivity constraint. The concrete construction of $\map{P}_{\inp\out}$ depends on the respective properties of $\map{P}_\inp$ and $\map{P}_{\out}$. Thm.~\ref{thm::genFormProj} provides this construction for the special case most often encountered in quantum mechanics, while the general case is discussed in Thm.~\ref{thm::FullyGenProj}. Likewise, the case where there are no affine constraints on $\Scal_\inp$ and $\Scal_\out$ is discussed in Thms.~\ref{thm::LinSpaceProj} and~\ref{thm::TransChoiGen}. }
    \label{fig::gen_setup}
\end{figure}

\section{Linear transformations between quantum objects}
\label{sec::GenTrans}
\subsection{Sets of quantum objects}
The characterisation~\eqref{eqn::ChannProj} of the set of quantum channels via a projector provides a potent way to derive the structural properties of mappings that occur frequently in quantum mechanics. We can use this structure to represent a very general class of (deterministic) quantum objects, such as quantum states, quantum channels, quantum superchannel~\cite{chiribella08quantum_supermaps}, quantum combs~\cite{chiribella07circuit_architecture,chiribella09networks}, channels with memory~\cite{kretschmann05channels_memory}, quantum strategies~\cite{gutoski07quantum_games}, non-Markovian processes~\cite{milz21markov}, causal quantum operations~\cite{backman01causal_operations}, non-signalling channels~\cite{piani06nonsignalling}, process matrices~\cite{oreshkov11,araujo15witnessing}, and more generally mappings between \textit{any} kinds of linear spaces~\cite{perinotti16higher,bisio_theoretical_2019, simmons_higher-order_2022, hoffreumon_projective_2022,apadula22nosignalling}. Before discussing the fully general case, i.e., mappings between general linear spaces (see Sec.~\ref{sec::GenAppr}), we start with a discussion of scenarios that are commonly encountered in quantum mechanics.
\begin{definition}[Set of quantum objects] \label{def:quantum_object}
A set of linear operators $\mathcal{S}\subseteq\L(\H)$ is a quantum object set if it can be described by:

\begin{subequations}
\begin{empheq}[box=\fcolorbox{brow}{white}]{alignat=3}
\text{A linear operator }  W&\in\L(\H) &&\text{ belongs to $\mathcal{S}$ if and only if:} \nonumber \\
  W &\geq0 && \text{(Positive semidefinite),}\\
\map{P}[W]&=W &&  \text{(Belongs to a particular linear subspace),}\\
\tr[W]&=\gamma && \text{(Fixed trace)},
\end{empheq}
\end{subequations}
where $\map{P}:\L(\H)\to\L(\H)$ is a linear projective map, that is $\map{P}^2:=\map{P}\circ\map{P}=\map{P}$.
\end{definition}
For example, both quantum states and the set of Choi matrices of quantum channels satisfy the above definition. For quantum states, we have $\map{P} = \map{\id}$ (where $\map{\id}$ denotes the identity map\footnote{The identity map is defined by $\map{\id}[X] = X$ for all $X\in \L(\H)$.}) and $\gamma = 1$, while for quantum channels, $\map{P}$ is given by Eq.~\eqref{eqn::ChannProj} and $\gamma = d_\inp$. 

\subsection{Transformations between quantum objects}

Let us consider two arbitrary sets of quantum objects $\Scal_\inp \subseteq \L(\H_\inp) $ and $ \Scal_\out \subseteq \L(\H_\out)$ where we use \inp\,  and \out\, as general placeholders for 'input' and 'output'. 
Our main question then is:
{\begin{center}
\textit{How can the set of quantum transformations $\widetilde T_{\inp\out}$ from $\Scal_\textup{\inp}$ to $ \Scal_\textup{\out}$ be characterised?}     
\end{center}}
Since we require $\map{T}_{\inp\out}$  to map elements from $\Scal_\textup{\inp}$ to $ \Scal_\textup{\out}$, we require that for every $W\in\Scal_\inp$, we have that $ \map{T}_{\inp\out}[W]\in\Scal_\out$, where we use additional subscripts on $\map{T}$ to signify its input and output space. Also, in order to be consistent with the linearity of quantum theory, the transformation $\map{T}_{\inp\out}$ is required to be a linear map from the linear space spanned by $\Scal_\inp$ to the linear space spanned by $\Scal_\out$.  Additionally, since all elements of $\Scal_\inp$ and $\Scal_\out$ are positive, we would at least require that $\map{T}_{\inp\out}$ is positive on all $W\in \Scal_\inp$. In line with standard considerations in quantum mechanics, throughout, we go beyond this minimal requirement\footnote{Positivity can also be argued for in order to ensure that all \textit{probabilistic} quantum object (see Sec.~\ref{sec::Probabilistic} for more details) are mapped to positive objects as well. However, this argument requires that the probabilistic quantum objects, i.e., all $W^\sharp\in \L(\H_\inp)$ which satisfy $W^\sharp \leq W$ for some $W\in \Scal_\inp$, span the full space $\L(\H_\inp)$. This is the case if $\Scal_i$ contains at least one full rank state, which we generally assume (see below).} and demand that $\map{T}_{\inp\out}$ is a positive map on all of $\L(\H_\inp)$. Finally, similarly to quantum channels acting on quantum states, we desire $\map{T}_{\inp\out}$ to be a valid transformation even when it is applied to only a part of a larger quantum object, which (at least) requires that $\map{T}_{\inp\out}:\L(\H_\inp)\to\L(\H_\out)$ is completely positive. In turn, this implies that all Choi matrices we encounter throughout are positive semidefinite (see Sec.~\ref{sec::ChoiCharSpec}). 

\begin{definition}[Quantum Transformations]
Let $\map{P}_\inp:\L(\H_\inp)\to\L(\H_\inp)$ and  $\map{P}_\out:\L(\H_\out)\to\L(\H_\out)$ be linear projective maps and $\Scal_\inp \subseteq \L(\H_\inp) $ and $ \Scal_\out \subseteq \L(\H_\out)$ be sets of quantum objects defined by
\begin{align}
\setlength{\arrayrulewidth}{1pt}
\renewcommand{\arraystretch}{1.2}
\begin{tabular}{|c|c|c|}
\arrayrulecolor{brow}
\hhline{|-|~|-|}
$W\in \L(\H_\inp) $ belongs to $\Scal_\inp$ iff  & \hspace*{20mm}& $W'\in \L(\H_\out) $ belongs to $\Scal_\out$ iff  \\
$W\geq0$, &\hspace*{5mm} \Large{$\mathbf{\longrightarrow}$}  \hspace*{5mm} & $W'\geq0$, \\
$\map{P}_\inp[W]=W$, && $\map{P}_\out[W']=W'$, \\
$\tr[W]=\gamma_\inp$. && $\tr[W']=\gamma_\out$. \\
\hhline{|-|~|-|}
\end{tabular}
\end{align} 
A linear map $\map{T}_{\inp\out}:\L(\H_\inp)\to\L(\H_\out)$ is a quantum transformation from $\Scal_\inp$ to $\Scal_\out$ if:
\begin{subequations}
\arrayrulecolor{blu}
\begin{empheq}[box=\fcolorbox{blu}{white}]{align}
	i:& \quad \map{T}_{\inp\out} \text{ is completely positive }\\
	ii:& \quad \forall W\in\Scal_\inp, \text{ we have that } \map{T}_{\inp\out}[W]\in\Scal_\out
\end{empheq}
\end{subequations}
\end{definition}
General linear mappings of this type have previously been employed in the quantum information literature, for example for the analysis of the dynamics of quantum causal structures~\cite{castro-ruiz_dynamics_2018} as well as, under the guise of 'admissible adapters', in the resource theory of causal connection~\cite{milz_resource_2022}, or as 'structure-preserving maps'~\cite{hoffreumon_projective_2022}, in the study of transformations between general quantum objects. More detailed structural investigations of the hierarchy of transformations such maps engender have been carried out in~\cite{perinotti16higher, bisio_theoretical_2019, simmons_higher-order_2022, hoffreumon_projective_2022}.

Importantly, for the concrete characterisation of $\map{T}_{\inp\out}$ (or, equivalently, its Choi matrix $T_{\inp\out}$) only the linear and affine constraints on $\Scal_\inp$ and $\Scal_\out$ play a role. The positive semidefiniteness constraint on both sets on the other hand only enters in the requirement for $\map{T}_{\inp\out}$ to be CP (or, equivalently, its Choi matrix $T_{\inp\out}$ to be positive semidefinite). Concretely, this holds true, since in the cases we consider, the positivity restriction does not alter the span of the sets, i.e., both $\Scal_\inp$ and $\Scal_\out$ span the same spaces ($\map{P}_\inp[\L(\H_\inp)]$ and $\map{P}_\out[\L(\H_\out)]$, respectively) with or without the positivity constraints imposed on their elements. Consequently, positivity of the respective elements does not enter as an additional constraint on $\map{T}_{\inp\out}$. As a result, in what follows, we rarely ever explicitly assume positivity for the elements of $\Scal_\inp$ and $\Scal_\out$ and mostly consider transformations between affine sets. Positivity of the respective elements, as well as complete positivity of the maps between $\Scal_\inp$ and $\Scal_\out$ can then always be imposed by hand without any added complications. On the other hand, while a similar argument could seemingly be made for the affine constraints -- since they generally do not change the span of $\Scal_\inp$ and $\Scal_\out$, either -- they fix a rescaling factor, in the sense that $\tr[\map{T}_{\inp\out}[W]] = \gamma_\out/\gamma_\inp\tr[W]$ for all $W \in \Scal_\inp$, thus playing a crucial role for the specific properties of $\map{T}_{\inp\out}$.
\vspace{-2mm}
\subsection{Map characterisation of quantum transformations}
We now present our first theorem -- which has, in slightly different form, already been derived in Refs.~\cite{castro-ruiz_dynamics_2018, milz_resource_2022, simmons_higher-order_2022, hoffreumon_projective_2022} -- to characterise quantum transformations. In this first characterisation, we aim to completely characterise the linear map $\map{T}_{\inp\out}$ without making reference to its Choi operator, but directly to its map properties. \vspace{-1.5mm}
\begin{theorem}[Transformation between affine sets: map version]
\label{thm::genFormProjOp}
Let $\map{P}_\inp:\L(\H_\inp)\to\L(\H_\inp)$ and  $\map{P}_\out:\L(\H_\out)\to\L(\H_\out)$ be linear projective maps and $\Scal_\inp \subseteq \L(\H_\inp) $ and $ \Scal_\out \subseteq \L(\H_\out)$ be affine sets defined by \vspace{-2.5mm}
\begin{align}
\setlength{\arrayrulewidth}{1pt}
\renewcommand{\arraystretch}{1.2}
\arrayrulecolor{brow}
\begin{tabular}{|c|c|c|}
\hhline{|-|~|-|}
$W\in \L(\H_\inp) $ belongs to $\Scal_\inp$ iff  & \hspace*{20mm}& $W'\in \L(\H_\out) $ belongs to $\Scal_\out$ iff  \\
$\map{P}_\inp[W]=W$, &\hspace*{5mm} \Large{$\mathbf{\longrightarrow}$}  \hspace*{5mm}& $\map{P}_\out[W']=W'$, \\
$\tr[W]=\gamma_\inp$. && $\tr[W']=\gamma_\out$.  \\
\hhline{|-|~|-|}
\end{tabular}
\end{align} \vspace{-1.5mm}
For $\gamma_\inp \neq 0$\footnote{We emphasise that assuming that $\tr[W]\neq0$ is not a strong restriction, and is quite natural for practical applications. Since all quantum objects are positive semidefinite, the only traceless object is the zero operator. In Sec.~\ref{sec::GenAppr} we discuss a more general version of this theorem.}, a linear map $\map{T}_{\textup{\inp \out}}:\L(\H_\textup{\inp}) \rightarrow \L(\H_\textup{\out})$ satisfies
$
        \map{T}_{\textup{\inp \out}}[W]\in \mathcal{S}_\textup{\out},$ for all $ W \in \mathcal{S}_\textup{\inp} 
$
iff \vspace{-1mm}
\begin{subequations}
\begin{empheq}[box=\fcolorbox{blu}{white}]{align}
\label{eqn::GenProjOp1}
   \map{P}_\textup{\out} \circ \map{T}_{\textup{\inp\out}} \circ \map{P}_\textup{\inp} &= \map{T}_{\textup{\inp\out}} \circ \map{P}_\textup{\inp},\\
    \label{eqn::GenProjOp2}
   \text{and} \quad \tr \circ \map{T}_{\textup{\inp\out}} \circ \map{P}_\textup{\inp} &= \frac{\gamma_\out}{\gamma_\inp} \tr \circ \map{P}_\textup{\inp}\, .
\end{empheq}
\end{subequations}
\end{theorem}
\begin{proof}
We start by showing that if Eqs.~\eqref{eqn::GenProjOp1} and ~\eqref{eqn::GenProjOp2} hold, then
$       \map{T}_{\textup{\inp \out}}[W]\in \mathcal{S}_\textup{\out},$ for all $ W \in \mathcal{S}_\textup{\inp} $. 
Let $W\in\Scal_\inp$. Then, by definition, we have $\map{P}_\inp[W] = W$. Now, since Eq.~\eqref{eqn::GenProjOp1} holds, for all $W \in \Scal_\inp$ we have
\begin{align}
 \map{P}_\textup{\out} [\map{T}_{\textup{\inp\out}}[W]] =  \map{P}_\textup{\out} \circ \map{T}_{\textup{\inp\out}} \circ \map{P}_\textup{\inp}[W] = \map{T}_{\textup{\inp\out}} \circ \map{P}_\textup{\inp}[W] = \map{T}_{\textup{\inp\out}}[W].
\end{align}
Additionally, since $\tr[W] = \gamma_\inp$ for every $W\in\Scal_\inp$, from Eq.~\eqref{eqn::GenProjOp2} we obtain
\begin{align}
\tr\circ\map{T}_{\textup{\inp\out}}[W] =\tr\circ\map{T}_{\textup{\inp\out}}\circ\map{P}_\inp[W] = \frac{\gamma_\out}{\gamma_\inp}\tr[W] =\gamma_\out,
\end{align}
and hence, Eqs.~\eqref{eqn::GenProjOp1}  and~\eqref{eqn::GenProjOp2} together imply that  $\map{T}_{\textup{\inp \out}}[W]\in \mathcal{S}_\textup{\out},$ for all $ W \in \mathcal{S}_\textup{\inp} $. 

For the converse direction, we note that since $\gamma_\inp \neq 0$, the affine constraint has no influence on the span of $\Scal_\inp$, such that $\text{span}(\Scal_\inp) = \map{P}_\inp[\L(\H_\inp)]$. By assumption, $\map{P}_\out\circ \map{T}_{\inp\out}[W] = \map{T}_{\inp\out}[W]$ holds for all $W \in \Scal_\inp$, and by linearity, we have $\map{P}_\out\circ \map{T}_{\inp\out}[M] = \map{T}_{\inp\out}[M]$ for all $M \in \text{span}(\Scal_\inp)$. For any arbitrary $X\in \L(\H_\inp)$ we have $M:=\map{P}_\inp[X] \in \text{span}(\Scal_\inp)$, and thus 
\begin{gather}
    \map{P}_\out\circ \map{T}_{\inp\out} \circ \map{P}_\inp[X] = \map{T}_{\inp\out} \circ  \map{P}_\inp [X].
\end{gather}
Since this holds for arbitrary $X \in \L(\H_\inp)$, we see that 
Eq.~\eqref{eqn::GenProjOp1} is satisfied. Similarly, if $\map{T}_{\inp \out}$ is a map from $\Scal_\inp$ to $\Scal_\out$, by linearity, we see that $\tr[\map{T}_{\inp \out}[M]] = \gamma_\out/\gamma_\inp \tr[M]$ for all $M \in \text{span}(\Scal_\inp)$. Thus, for arbitrary $X \in \L(\H_\inp)$ we have 
\begin{gather}
\tr\circ \map{T}_{\inp \out} \circ \map{P}_\inp[X] = \frac{\gamma_\out}{\gamma_\inp} \tr \circ \map{P}_\inp[X]\, ,
\end{gather}
where, again, we have used that $\map{P}_\inp[X] \in \text{span}(\Scal_\inp)$. Since the above equation holds for all $X \in \L(\H_\inp)$, we thus recover Eq.~\eqref{eqn::GenProjOp2}, concluding the proof.
\end{proof}
We emphasise that the above Theorem covers the case where $\gamma_\out = 0$, for which we have $\map{P}_\textup{\out} \circ \map{T}_{\textup{\inp\out}} \circ \map{P}_\textup{\inp} = \map{T}_{\textup{\inp\out}} \circ \map{P}_\textup{\inp}$ and $\tr \circ \map{T}_{\textup{\inp\out}} \circ {P}_\textup{\inp} = 0$. However, in this case, imposing positivity on the elements of $\Scal_\out$ would, unlike in all cases we consider, lead to explicit further simplifications (see App.~\ref{app::gam0}). On the other hand, the scenario $\gamma_\inp =0$ is not directly covered by the above theorem and in principle requires special consideration. For this scenario, it is easy to see that $\gamma_\inp = 0$ implies $\gamma_\out = 0$. Then, one can readily define a new projector $\map{P}_\inp'$ that projects onto a vector space of traceless matrices (thus incorporating the requirement $\gamma_\inp = 0$), such that $W\in \Scal_\inp$ iff $\map{P}[W] = W$. With this, a map $\map{T}_{\inp\out}$ maps between $\Scal_\inp$ and $\Scal_\out$ if and only if $\map{P}_\textup{\out} \circ \map{T}_{\textup{\inp\out}} \circ \map{P}_\textup{\inp}' = \map{T}_{\textup{\inp\out}} \circ \map{P}_\textup{\inp}'$ and $\tr \circ \map{T}_{\textup{\inp\out}} \circ {P}_\textup{\inp}' = 0$. The details can be found in App.~\ref{app::gam0}. From now on, whenever not explicitly mentioned, we will exclude both of these scenarios to avoid unnecessary technical complications and assume $\gamma_\inp, \gamma_\out \neq 0$.

While providing necessary and sufficient conditions for quantum transformations between arbitrary quantum sets, Thm.~\ref{thm::genFormProjOp} is not particularly insightful when it comes to the structural properties of $\map{T}_{\inp\out}$ and does not easily allow for an incorporation of the properties that many projectors $\map{P}_\inp$ encountered in quantum mechanics have. In the following subsection, we provide a specialised version of Thm.~\ref{thm::genFormProjOp} in terms of the Choi state $T_{\inp\out}$ that takes commonly assumed properties of $\map{P}_\inp$ into account and will be of more direct use.

\subsection{Choi characterisation of particular quantum transformations}
\label{sec::ChoiCharSpec}
	For most of practical cases, the projectors associated to the sets of quantum objects respect additional properties which allow us to present a more specialised and useful characterisation of quantum transformations. In particular, there are three properties which the projector $\map{P}$ associated to a quantum set $\Scal\subseteq\L(\H)$ often respects
\begin{enumerate}
	\item Unitality: $\map{P}[\id]=\id$
	\item Self-adjointness\footnote{
Let $\map{P}: \L(\H_\textup{\xt}) \rightarrow \L(\H_\textup{\yt})$ be a linear operator. Its adjoint is the unique map $\map{P}^\dagger: \L(\H_\textup{\yt}) \rightarrow  \L(\H_\textup{\xt})$ respecting
$\tr[B^\dagger \widetilde P[A]] = \tr[\map{P}^\dagger[B]^\dagger A]$ for all  $A \in \L(\H_\xt)$ and all $B \in \L(\H_\yt).
$
} $\map{P}=\map{P}^\dagger$
	\item Commutation with the transposition:  $\map{P}[W^\tau]={\map{P}[W]}^\tau$, for every $ W\in\L(\H)$ 
\end{enumerate}
We notice that these three properties are respected by the projectors onto the sets of quantum states, quantum channels, superchannels, quantum combs, process matrices, and non-signalling channels, to name but a few.\footnote{
For a quantum set $\Scal \subseteq \L(\H)$ defined by a unital and self-adjoint projector, it holds that $W \in \Scal$ if and only if 
$
    W = \map{P}(X) + \frac{\id}{\tr{(\id)}}\Big(\gamma-\tr(X)\Big)
$ for some matrix $X\in\L(\H)$,   a parametrisation which may be very convenient. The particular case of this parametrisation for quantum channels was already considered in Ref.~\cite{araujo15witnessing} to obtain a projective characterisation of process matrices. }
 We now present a characterisation theorem tailored for this particular case in terms of Choi matrices, while the more general case with no extra assumptions is discussed in Sec.~\ref{sec::GenAppr}.
 
As a first step for this characterisation, we recall that the action of a map $\map{T}_{\inp\out}:\L(\H_\inp) \rightarrow \L(\H_\out)$ in terms of its Choi matrix $T_{\inp\out}:\L(\H_\inp \otimes \H_\out)$ can be written as (see Eq.~\eqref{eqn::ChoiDef}) 
 \begin{gather}
    \label{eqn::LinkIntro}
     \map{T}_{\inp\out}[X_\inp] =: \tr_\inp[(X_\inp^\tau \otimes \id_\out) T_{\inp\out}] =: T_{\inp\out} \star X_\inp , 
 \end{gather}
 where we have employed the \emph{link product} $\star$ (see Sec.~\ref{sec::LinOp} for more details) and added subscripts to emphasise what spaces the respective elements are defined on. As mentioned above, complete positivity of $\map{T}_{\inp\out}$ is equivalent to $T_{\inp\out} \geq 0$~\cite{chiribella09networks}. With this, the characterisation of the map $\map{T}_{\inp\out}$ amounts to a characterisation of the matrix $T_{\inp\out}$, which can be obtained via a projector, denoted by $\map{P}_{\inp\out}$.
This characterisation has also been given in Refs.~\cite{castro-ruiz_dynamics_2018, milz_resource_2022, simmons_higher-order_2022, hoffreumon_projective_2022} and is provided here in the notation we employ.\footnote{We remark that the characterisation presented in Ref.~\cite{castro-ruiz_dynamics_2018} {differs from ours since it} misses two terms that do not necessarily cancel out.}

\begin{theorem}[Transformation between affine sets: specialised Choi version]
\label{thm::genFormProj}
Let $\map{P}_\inp:\L(\H_\inp)\to\L(\H_\inp)$ and  $\map{P}_\out:\L(\H_\out)\to\L(\H_\out)$ be linear projective maps and $\Scal_\inp \subseteq \L(\H_\inp) $ and $ \Scal_\out \subseteq \L(\H_\out)$ be affine sets defined by
\begin{align}
\setlength{\arrayrulewidth}{1pt}
\renewcommand{\arraystretch}{1.2}
\begin{tabular}{|c|c|c|}
\hhline{|-|~|-|}
$W\in \L(\H_\inp) $ belongs to $\Scal_\inp$ iff  & \hspace*{20mm}& $W'\in \L(\H_\out) $ belongs to $\Scal_\out$ iff  \\
$\map{P}_\inp[W]=W$, &\hspace*{5mm} \Large{$\mathbf{\longrightarrow}$}  \hspace*{5mm}& $\map{P}_\out[W']=W'$, \\
$\tr[W]=\gamma_\inp$. && $\tr[W']=\gamma_\out$.  \\
\hhline{|-|~|-|}
\end{tabular}
\end{align} 
Additionally, we assume that the maps $\map{P}_\inp$ and $\map{P}_\out$ are self-adjoint and unital, and that $\map{P}_\inp$ commutes with the transposition map, \ie,
\begin{subequations}
\begin{empheq}[box=\widefbox]{alignat=3}
   &\map{P}_\inp=\map{P}_\inp^\dagger,  && \map{P}_\out=\map{P}_\out^\dagger, \\
   &\map{P}_\inp[\id]=\id,  && \map{P}_\out[\id]=\id, \\
   &\map{P}_\inp[W^\tau]=\map{P}_\inp&&[W]^\tau, \quad\forall W\in\L(\H_\inp).
\end{empheq}
\end{subequations}
For $\gamma_\inp \neq 0$, a linear map $\map{T}_{\textup{\inp \out}}:\L(\H_\textup{\inp}) \rightarrow \L(\H_\textup{\out})$ satisfies
$
        \map{T}_{\textup{\inp \out}}[W]\in \Scal_\textup{\out},$ for all $ W \in \Scal_\textup{\inp} 
$
if and only if
\begin{subequations}
\begin{empheq}[box=\fcolorbox{blu}{white}]{align}
\label{eqn::GenProj1}
    \map{P}_{\textup{\inp\out}}[T_{\textup{\inp \out}}]:=& T_{\textup{\inp \out}} - (\map{P}_\textup{\inp}\otimes\map{\id}_\out) [T_{\textup{\inp\out}}] + (\map{P}_\textup{\inp} \otimes\map{P}_\textup{\out}) [T_{\textup{\inp\out}}]  - (\map{P}_\textup{\inp}\otimes\map{\id}_\out)[{}_\textup{\out}T_{\textup{\inp\out}}] +  {}_{\textup{\inp\out}}T_{\textup{\inp\out}}=T_{\textup{\inp \out}}\, ,\\
\label{eqn::GenProj2}
 \tr [T_{\textup{\inp\out}}] =& \frac{\gamma_\out}{\gamma_\inp} d_\textup{\inp}\,,
\end{empheq}
\end{subequations}
holds for its Choi matrix $T_{\inp\out}$, where $\map{\id}_\out$ is the identity map, $d_\inp$ is the dimension of $\H_\inp$, and $\map{P}_{\inp\out}:\L(\H_\inp\otimes\H_\out)\to\L(\H_\inp\otimes\H_\out)$ is a self-adjoint, unital projector that commutes with the transposition. 
\end{theorem}
\begin{proof} 
The derivation of Eqs.~\eqref{eqn::GenProj1} and~\eqref{eqn::GenProj2}  can be found in Sec.~\ref{sec::ProofLink} where we discuss the link product and the relevant mathematical tools to easily and systematically deal with Choi matrices of general linear transformations. Here, we show the remaining properties of the projector $\map{P}_{\inp\out}$, i.e., that it is self-adjoint, unital, commutes with the transposition and is, indeed a projector. To see this latter property, first note that a self-adjoint, unital projector $\map{P}_\xt$ is trace-preserving, since $\tr[\map{P}_\xt[M]] = \tr[\map{P}_\xt[\mathbbm{1}_\xt] M] =  \tr[M]$ for all $M \in \L(\H_\xt)$ and $\xt \in \{\inp, \out\}$. Consequently, ${}_\xt(\map{P}_\xt[M]) = \map{P}_\xt[{}_\xt M] = {}_\xt M$ for all $M \in \L(\H_\xt)$, and thus ${}_\xt \circ \map{P}_\xt [M] = \map{P}_\xt \circ {}_\xt[M] = {}_\xt M$. Additionally, ${}_{\xt\xt}M = {}_\xt M$, and, by assumption ${\map{P}_\xt}^2 = {\map{P}_{\xt}}$ for $\xt\in \{\inp, \out\}$. With this using Eq.~\eqref{eqn::GenProj1}, it is easy to see that 
\begin{gather}
{\map{P}_{\inp\out}}^2 = (\map{\id} - \map{P}_\inp \otimes \map{\id}_\out + \map{P}_\inp \otimes \map{P}_\out - \map{P}_\inp \otimes {}_\out + {}_\inp \otimes {}_\out)^2 = \map{P}_{\inp\out}\,
\end{gather}
holds, i.e., $\map{P}_{\inp\out}^2 = \map{P}_{\inp\out}$. Finally, 
since both ${}_\inp\sbt$ and ${}_\out\sbt$ are self-adjoint, unital and commute with the transposition, these properties also hold for $\map{P}_{\inp\out}$ whenever they hold for{$\map{P}_\inp$ and} $\map{P}_\out$. 
\end{proof}

Naturally, the above Theorem is not as general as the one for maps given in Thm.~\ref{thm::genFormProjOp}, since it requires additional properties of $\map{P}_\inp$ and $\map{P}_\out$. However, it allows for a direct characterisation of the properties of a concrete representation of linear mappings, and applies to most scenarios that are relevant in quantum mechanics (see Ex.~\ref{Ex::Gibbs} for a concrete example where these properties are not satisfied, though.). Its generalisation, which is equivalent to Thm.~\ref{thm::genFormProjOp}, can be found in Sec.~\ref{sec::GenAppr}. Also in Sec.~\ref{sec::GenAppr}, we provide a version of Thm.~\ref{thm::genFormProj} for mappings that are \textit{not} trace-rescaling, that is, we discuss transformations between \textit{linear} subspaces instead of affine subspaces, which is both of independent interest and highlights the role that the affine constraints on the sets $\Scal_\inp$ and $\Scal_\out$ play for the properties of $T_{\inp\out}$. As was the case for Thm.~\ref{thm::genFormProjOp}, the case $\gamma_\inp = 0$ is explicitly excluded in the above Theorem. It is discussed in detail in App.~\ref{app::gam0Choi}, as a special instance of the general case (i.e., where we impose no restrictions on $\map{P}_\inp$ and $\map{P}_\out$). Before discussing this general case in detail and providing the technical details for the derivation of the above Theorem, we now first show its concrete application for commonly encountered scenarios in quantum mechanics.

\section{Applications to particular quantum transformations and transformations without a fixed causal order}
\label{sec::Emergence}
We now apply Thm.~\ref{thm::genFormProj} to obtain a quantum set characterisation for several quantum transformations used in the literature. Later in this section we also discuss the simplest quantum transformation which may disrespect a standard notion of fixed causal order. 

\begin{example}[Quantum states to quantum states (Quantum Channels)]
\label{ex::QStatesQStates}
\begin{figure}[t!]
    \centering
    \includegraphics{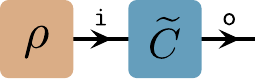}
    \caption{\textbf{Quantum Channels} A quantum channel $\map{C}$ which maps quantum states $\rho$ onto quantum states $\rho$'.} 
    \label{fig::Quantum Channel}
\end{figure}
In Sec.~\ref{sec::QstatesQchannels}, we derived the properties of quantum channels $\map{C}$ that map quantum states $\rho \in \L(\H_\inp)$ onto quantum states $\rho' \in \L(\H_\out)$. Since quantum states are unit trace, we have $\gamma_\inp = \gamma_\out = 1$, and there are no linear constraints on quantum states, such that $\map{P}_\inp = \map{\id}_\inp$ and $\map{P}_\out = \map{\id}_\out$ (i.e., the identity channel). Naturally, $\map{\id}_\xt$ is unital, self-adjoint, and commutes with the transposition, such that Thm.~\ref{thm::genFormProj} applies. Employing Eqs.~\eqref{eqn::GenProj1} and~\eqref{eqn::GenProj1}, we directly obtain (for less cluttered notation, we omit the subscripts on $C_{\inp\out}$):
\begin{align}
\setlength{\arrayrulewidth}{1pt}
\renewcommand{\arraystretch}{1.2}
\begin{tabular}{|c|c|c|}
\hhline{|-|~|-|}
\bf{(Quantum state)} & \hspace*{20mm}& \bf{(Quantum state)}   \\
$\rho\in \L(\H_\inp) $ belongs to $\Scal_\inp$ iff  & \hspace*{20mm}& $\rho'\in \L(\H_\out) $ belongs to $\Scal_\out$ iff  \\
$\rho\geq0$, &\hspace*{5mm} \Large{$\mathbf{\longrightarrow}$}  \hspace*{5mm}& $\rho'\geq0$, \\
$\map{P}_\inp[\rho]$  $:= \map{\mathbbm{1}}[\rho]=\rho$, && $\map{P}_\out[\rho']$ $:= \map{\mathbbm{1}}[\rho']=\rho'$, \\
$\tr[\rho]=1$. && $\tr[\rho']=1$. \\
\hhline{|-|~|-|}
\end{tabular}
\end{align} 
\begin{subequations}
\begin{empheq}[box=\fcolorbox{blu}{white}]{align}
 \text{\bf{(Quantum}} & \text{  \bf{channel)}}\nonumber \\
C\geq\,&0, \\
    \label{eqn::ChanExProj}
     \map{P}_{\inp\out}[C]:=\,&C - {}_\out C + {}_{\inp\out}C =C,
     \\
 \tr [C] =\,& d_\inp\,,
\end{empheq}
\end{subequations}
which coincides exactly with the properties~\eqref{eq:lin1} and~\eqref{eq:aff1} derived in Sec.~\ref{sec::QstatesQchannels}. Additionally, demanding that $\map{C}$ is completely positive then imposes $C_{\inp\out} \geq 0$, i.e., Eq.~\eqref{eq:pos1}. For ease of notation, in the subsequent examples, we denote the projector of Eq.~\eqref{eqn::ChanExProj} by $\map{P}^{(C)}$ and add subscripts whenever we want to clarify what spaces it acts on. \hfill $\blacksquare$
\end{example}
\begin{example}[Quantum channels to states]
\begin{figure}[t!]
    \centering
    \includegraphics[scale=0.8]{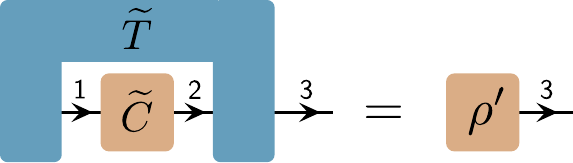}
    \caption{A quantum transformation $\map{T}$ which maps quantum channels $\map{C}$ into quantum states $\rho$.}
    \label{fig::Mmap}
\end{figure}
The next simplest transformation one could consider is a mapping from quantum channels to quantum states, i.e., transformations of the form $\map{T}[C_{12}] = \rho_3$ (see Fig.~\ref{fig::Mmap}).\footnote{Here, and in what follows, whenever there is no risk of confusion, we drop the semantic distinction between transformations and their respective Choi matrices.} Such transformations are particular types of quantum combs~\cite{chiribella07circuit_architecture, chiribella09networks}, and have been considered amongst others in the study of open quantum system dynamics with initial correlations under the name of $\mathcal{M}$-maps~\cite{modi_operational_2012}. Keeping track of the involved spaces, for this case, we have to identify $\H_\textup{\inp} \cong \H_1 \otimes \H_{2}$ and $\H_\textup{\out} \cong \H_3$. Since the resulting quantum states $\rho_3 \in \L(\H_3)$ are unit trace, while $\tr[C_{12}] = d_{1}$, we have $\gamma_\inp = 1$ and $\gamma_\out = d_1$.  Additionally, the role of $\map{P}_\textup{\inp}$ is now played by the projector $\map{P}^{(C)}_{12}$  of Eq.~\eqref{eqn::ChanExProj}, while $\map{P}_\textup{\out}$ is again given by $\map{\id}_3$ (since there are no linear restrictions on quantum states). Given that all involved projectors are self-adjoint, unital and commute with the transposition, Thm.~\ref{thm::genFormProj} applies. With this, using Eqs.~\eqref{eqn::GenProj1} and~\eqref{eqn::GenProj2}, we obtain for the Choi state $T\in \L_{123}(\H_1 \otimes \H_2\otimes \H_3)$ of the map $\widetilde T$: 
\begin{align}
\setlength{\arrayrulewidth}{1pt}
\renewcommand{\arraystretch}{1.2}
\begin{tabular}{|c|c|c|}
\hhline{|-|~|-|}
\bf{(Quantum channel)} & \hspace*{20mm}& \bf{(Quantum state)}   \\
$C\in \L(\H_1\otimes\H_2) $ belongs to $\Scal_\inp$ iff  & \hspace*{20mm}& $\rho'\in \L(\H_3) $ belongs to $\Scal_\out$ iff  \\
$C\geq0$, &\hspace*{5mm} \Large{$\mathbf{\longrightarrow}$}  \hspace*{5mm}& $\rho'\geq0$, \\
$\map{P}_\inp[C]:=C-{}_2C+{}_{12}C=C$, && $\map{P}_\out[\rho']$ $=\map{\mathbbm{1}}[\rho']$ $=\rho'$, \\
$\tr[C]=d_1$. && $\tr[\rho']=1$. \\
\hhline{|-|~|-|}
\end{tabular}
\end{align} 
\begin{subequations}
\begin{empheq}[box=\fcolorbox{blu}{white}]{align}
 \text{\bf{(Quantum ch}} & \text{\bf{annel to quantum state)}}\nonumber \\
T\geq\,&0, \\
     \map{P}_{\inp\out}[T]:=\,&T - {}_3T + {}_{23}T =T,\\
 \tr [T] =\,& d_2\,.
\end{empheq}
\end{subequations}

The above coincides with $_{3}T = {}_{23}T$ and $\tr[T] = d_2$ which, in turn, are the causality/trace constraints of a one-slot comb with a final output line~\cite{chiribella09networks} (we discuss causality constraints in more detail below). Additionally, choosing $\H_1$ to be trivial, \ie, $\H_1\cong\mathbb{C}$, we recover the characterisation of quantum channels. As before, demanding complete positivity from $\map{T}$ translates to the additional requirement $T\geq 0$. \hfill $\blacksquare$
\end{example}

\begin{example}[Quantum channels to quantum channels (Quantum Superchannels)]
\label{ex::supchann}
Let us now consider the question raised at the end of Sec.~\ref{sec::QstatesQchannels}, namely the characterisation of transformations $\map{T}[C_{23}] = C_{14}'$ that map quantum channels $C_{23} \in \L(\H_2 \otimes \H_3)$ onto quantum channels $C_{14}' \in \L(\H_1 \otimes \H_4)$ (see Fig.~\ref{fig::SuperChannel}). In this case, we identify $\H_\textup{\inp} \cong \H_2 \otimes \H_3$ and $\H_\textup{\out} \cong \H_1 \otimes \H_4$. The projectors on the input and output space of $\map{T}$ are, respectively, given by the projectors $\map{P}^{(C)}_{23}$ and $\map{P}^{(C)}_{14}$ of Eq.~\eqref{eqn::ChanExProj}, which are self-adjoint, unital, and commute with the transposition, such that Thm.~\ref{thm::genFormProj} applies. In addition, for channels, we have $\gamma_\inp = \tr[C_{23}] = d_2$ and  $\gamma_\out = \tr[C_{14}'] = d_1$. Thus, employing Eqs.~\eqref{eqn::GenProj1} and~\eqref{eqn::GenProj2}, we obtain for the properties of $T\in \L(\H_1 \otimes \H_2 \otimes\H_3 \otimes\H_4)$: 
\begin{align}
\setlength{\arrayrulewidth}{1pt}
\renewcommand{\arraystretch}{1.2}
\begin{tabular}{|c|c|c|}
\hhline{|-|~|-|}
\bf{(Quantum channel)} & \hspace*{20mm}& \bf{(Quantum channel)}   \\
$C_{}\in \L(\H_2\otimes\H_3) $ belongs to $\Scal_\inp$ iff  & \hspace*{20mm}& $C'_{}\in \L(\H_1\otimes\H_4) $  belongs to $\Scal_\out$ iff  \\
$C_{}\geq0$, &\hspace*{5mm} \Large{$\mathbf{\longrightarrow}$}  \hspace*{5mm}& $C_{}'\geq0$, \\
$\map{P}_\inp[C]:=C_{}-{}_3C_{}+{}_{23}C_{}=C_{}$, && $\map{P}_\out[C']:=C'_{}-{}_4C'_{}+{}_{14}C'_{}=C'_{}$, \\
$\tr[C_{}]=d_2$. && $\tr[C']=d_1$. \\
\hhline{|-|~|-|}
\end{tabular}
\end{align} 
\begin{subequations}\label{eqn::Superchannel}
\begin{empheq}[box=\fcolorbox{blu}{white}]{align}
 \text{\bf{(Quantum}} & \text{  \bf{superchannel)}}\nonumber \\
T\geq&0, \\
\label{eqn::supChanProj}
     \map{P}_{\inp\out}[T]:=&T - {}_4T   + {}_{34}T  - {}_{234}T + {}_{1234}T =T,\\
 \tr [T] =& d_1d_3\,,
\end{empheq}
\end{subequations}
It is easy to see that the above is, in addition to $\tr[T] = d_{1}d_3$, equivalent to ${}_4T = {}_{34}T$ and ${}_{234}T = {}_{1234}T$, which, in the ordering of spaces we have chosen, coincides with the causality/trace constraints of a quantum comb with one slot (corresponding to the spaces labelled by $2$ and $3$), and an initial input (labelled by $1$) and final output (labelled by $4$)~\cite{chiribella2009theoretical}. This, in turn, reflects the well-known fact that there are no causally disordered superchannels~\cite{chiribella08quantum_supermaps}. Additionally, choosing $\H_1$ to be trivial, we recover the conditions~\eqref{eqn::Superchannel} on transformations of channels to states from above. 

Finally, here, it is insightful to discuss in what way the properties of $T$ would change if the trace conditions on the elements of $\Scal_\inp$ and $\Scal_\out$ were dropped. Then, the transformation ${\map{T}}': \L(\H_2\otimes \H_3) \rightarrow \L(\H_1\otimes \H_4)$ would still have to satisfy $(\map{P}^{(C)}_{14} \circ \map{T}') [C_{23}] = \map{T}' [C_{23}]$ for all $C_{23}= \map{P}^{(C)}_{23}[C_{23}]$, but it is not necessarily trace-rescaling. The corresponding characterisation for this case will be given in Thm.~\ref{thm::LinSpaceProj}. Using Eq.~\eqref{eqn::LinProj1} from Thm.~\ref{thm::LinSpaceProj} (see Sec.~\ref{sec::ProofLink}), one obtains
\begin{gather}
\begin{split}
 T' &= T' - \map{P}^{(C)}_{23}[T'] + (\map{P}^{(C)}_{23} \otimes \map{P}^{(C)}_{14})[T'] \\
 &= T' -{}_4T' + {}_{14}T' + {}_{34}T' - {}_{134}T' - {}_{234}T' + {}_{1234}T',
 \end{split}
\end{gather}
with no additional restriction on the trace of $T'$. Even setting aside the absence of an additional trace constraint on $T'$, the above Equation is significantly different from Eq.~\eqref{eqn::supChanProj}, underlining the importance of the affine constraints on $\Scal_\inp$ and $\Scal_\out$ for the properties of the transformations between them.\hfill $\blacksquare$
\begin{figure}[ht]
    \centering
    \includegraphics[scale=0.8]{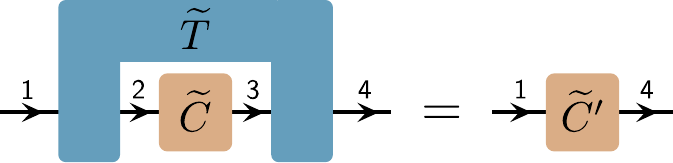}
    \caption{\textbf{Superchannels.} A superchannel $\map{T}$ maps a CPTP map $\map{C}$ onto a CPTP map $\map{C}'$.}
    \label{fig::SuperChannel}
\end{figure}
\FloatBarrier
\end{example}

\begin{example}[Quantum superchannels to quantum channels (Two-slot combs)]
Let us now consider the question of transforming superchannels into channels (see Fig.~\ref{fig::2slot_comb}), such mathematical objects are known in the literature as two-slot quantum combs~\cite{chiribella07circuit_architecture,chiribella08quantum_supermaps,chiribella09networks}, two-round quantum strategies~\cite{gutoski07quantum_games}, and two-slot sequential superchannels~\cite{quintino22deterministic}. In this case, we identify $\H_\textup{\inp} \cong \H_2 \otimes \H_3\otimes \H_4\otimes \H_5$ and $\H_\textup{\out} \cong \H_1 \otimes \H_6$.  Thus, employing the projectors from the previous examples, we obtain for the properties of $T\in \L(\H_1 \otimes \H_2 \otimes\H_3 \otimes\H_4\otimes\H_5\otimes\H_6)$: 
\begin{align}
\setlength{\arrayrulewidth}{1pt}
\renewcommand{\arraystretch}{1.2}
\begin{tabular}{|c|c|c|}
\hhline{|-|~|-|}
\bf{(Quantum superchannel)} & \hspace*{-3mm}& \bf{(Quantum channel)}   \\
$C_{}\in \L(\H_2\otimes\H_3\otimes\H_4\otimes\H_5) $  & & $C'_{}\in \L(\H_1\otimes\H_6) $   \\
belongs to $\Scal_\inp$ iff &\hspace*{-1mm} \Large{$\mathbf{\longrightarrow}$}  \hspace*{-1mm}&  belongs to $\Scal_\out$ iff \\
$C_{}\geq0$, && $C_{}'\geq0$, \\
$\map{P}_\inp[C]:=C - {}_5C   + {}_{45}C  - {}_{345}C + {}_{2345}C = C$, && $\map{P}_\out[C']:=C'_{}-{}_6C'_{}+{}_{16}C'_{}=C'_{}$, \\
$\tr[C_{}]=d_2d_4$. && $\tr[C']=d_1$. \\
\hhline{|-|~|-|}
\end{tabular}
\end{align} 
\begin{subequations}\label{eqn::2slot_comb}
\begin{empheq}[box=\fcolorbox{blu}{white}]{align}
 \text{\bf{(Quantum}} & \text{  \bf{two-slot comb)}}\nonumber \\
T\geq&0, \\
\label{eqn::Two_comb_Proj}
     \map{P}_{\inp\out}[T]:=&T - {}_6T   + {}_{56}T  - {}_{456}T + {}_{3456}T - {}_{23456}T + {}_{123456}T =T,\\
 \tr [T] =& d_1d_3d_5.
\end{empheq}
\end{subequations}

\begin{figure}[ht]
    \centering
    \includegraphics[width=0.65\linewidth]{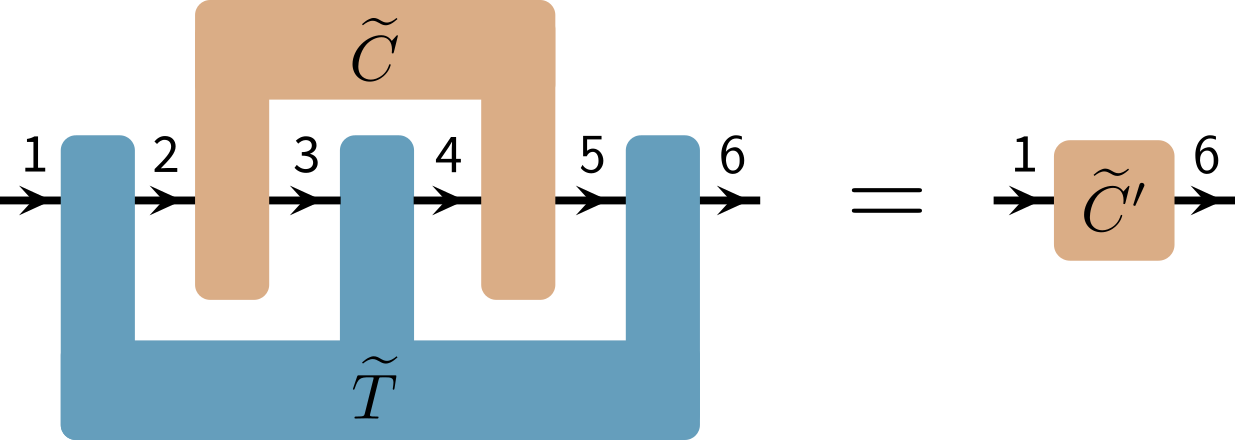}
    \caption{\textbf{Two-slot combs.} A two-slot comb $\map{T}$ maps a superchannel $\map{C}$ onto a CPTP map $\map{C}'$.}
    \label{fig::2slot_comb}
\end{figure}
\FloatBarrier

Two-slots quantum combs are known to be causally ordered since they may always be decomposed as a quantum circuit~\cite{chiribella07circuit_architecture,chiribella08quantum_supermaps,chiribella09networks,gutoski07quantum_games}, that is, as an ordered sequence of quantum channels, as illustrated in Fig.~\ref{fig::2slot_comb_ordered}. As discussed in the next example, this causally ordered property also holds for any $(k+1)$-slot quantum comb, which is a quantum transformation from a $k$-slot comb into a quantum channel~\cite{chiribella09networks}. \hfill $\blacksquare$
\begin{figure}[ht]
    \centering
    \includegraphics[width=0.9\linewidth]{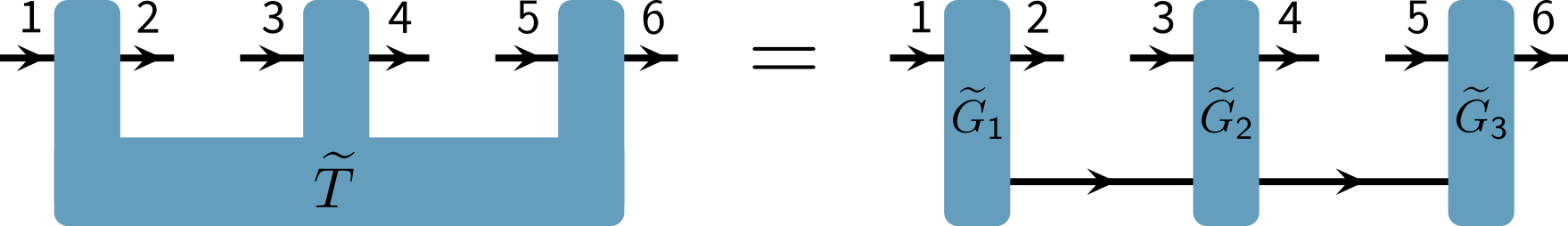}
    \caption{\textbf{Two-slot combs.} A two-slot comb $\map{T}$ can always be decomposed as a sequence of CPTP maps $\{\map{G}_1, \map{G}_2, \map{G}_3\}$.}
    \label{fig::2slot_comb_ordered}
\end{figure}
\end{example}

\begin{example}[Quantum $(k-1)$-slot combs to quantum channels ($k$-slot combs)]
Let us now consider the question of transforming $(k-1)$-slot combs into channels (see Fig.~\ref{fig::kslot_comb}), such mathematical objects are known in the literature as $k$-slot quantum combs~\cite{chiribella07circuit_architecture,chiribella08quantum_supermaps,chiribella09networks}, $k$-round quantum strategies~\cite{gutoski07quantum_games}, and $k$-slot sequential superchannels~\cite{quintino22deterministic}.
\begin{definition}
A $0$-slot quantum comb is a quantum channel. For an arbitrary $k\in\mathbb{N}$, a $k$-slot quantum comb is recursively defined as a quantum transformation which maps $(k-1)$-quantum combs into quantum channels.
\end{definition}
In this case, we identify $\H_\textup{\inp} \cong \H_2 \otimes \H_3\otimes \ldots \H_{2{k}+1}$ and $\H_\textup{\out} \cong \H_1 \otimes \H_{{2k+2}}$.  Thus, employing the results from previous examples, we obtain for the properties of $T\in \L(\H_1 \otimes \H_2 \ldots \otimes\H_{2k+2})$: 
\begin{align}
\setlength{\arrayrulewidth}{1pt}
\renewcommand{\arraystretch}{1.2}
\begin{tabular}{|c|c|c|}
\hhline{|-|~|-|}
\bf{(($k-1$)-slot comb)} && \bf{(Quantum channel)}   \\
$C_{}\in \L(\H_2\otimes\H_3\otimes \ldots \otimes\H_{2k+1}) $  & & $C'_{}\in \L(\H_1\otimes\H_{2k+2}) $   \\
 belongs to $\Scal_\inp$ iff &\hspace*{2mm}\Large{$\mathbf{\longrightarrow}$} \hspace*{2mm}&  belongs to $\Scal_\out$ iff  \\
$C_{}\geq0$, && $C_{}'\geq0$, \\
$\map{P}_\inp[C]:=C - {}_{(2k+1)}C   + {}_{(2k)(2k+1)}C  -\ldots$  && $\map{P}_\out[C']:=C'_{}-{}_{(2k+2)}C'_{}$ \\
\phantom{asdfffffff}$- {}_{3\ldots(2k+1)}C + {}_{2\ldots(2k+1)}C = C,$  && \phantom{asdfffffffffffffff}$+{}_{1(2k+2)}C'_{}=C'_{}$,\\
$\tr[C_{}]=d_2d_4\ldots d_{2k}$. && $\tr[C']=d_1$.\\
\hhline{|-|~|-|}
\end{tabular}
\end{align} 
\begin{subequations}\label{eqn::kslot_comb}
\begin{empheq}[box=\fcolorbox{blu}{white}]{align}
 \text{\bf{(Quantum}} & \text{  $\mathbf{k}$\bf{-slot comb)}}\nonumber \\
T\geq&0 \\ \nonumber
\label{eqn::k_slot_comb}
    \map{P}_{\inp\out}[T]:=&T - {}_{(2k+2)}T   + {}_{(2k+1)(2k+2)}T  - {}_{(2k)(2k+1)(2k+2)}T + \phantom{f} \\
    &\ldots - {}_{23\ldots (2k+2)}T + {}_{12\ldots(2k+2)} T =T, \nonumber \\
 \tr [T] =& d_1d_3d_5\ldots d_{(2k+1)} \nonumber\,,
\end{empheq}
\end{subequations}

\begin{figure}[ht]
    \centering
    \includegraphics[width=0.95\linewidth]{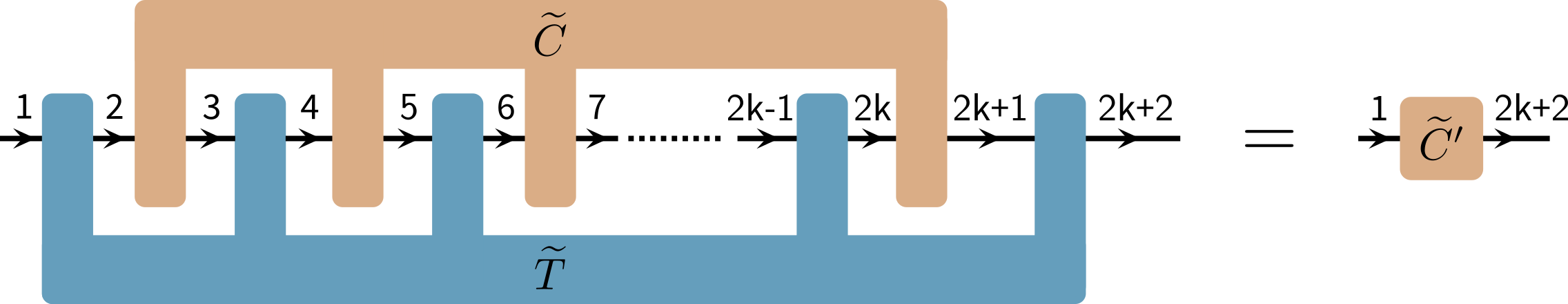}
    \caption{\textbf{$k$-slot combs.} A $k$-slot comb $\map{T}$ maps a $(k-1)$-slot comb $\map{C}$ onto a CPTP map $\map{C}'$.}
    \label{fig::kslot_comb}
\end{figure}
\FloatBarrier

All $k$-slots quantum combs may be decomposed as a quantum circuit~\cite{chiribella07circuit_architecture,chiribella08quantum_supermaps,chiribella09networks,gutoski07quantum_games}, that is, as an ordered sequence of quantum channels, as illustrated in Fig.~\ref{fig::2slot_comb_ordered}.\hfill $\blacksquare$
\begin{figure}[ht]
    \centering
    \includegraphics[width=0.99\linewidth]{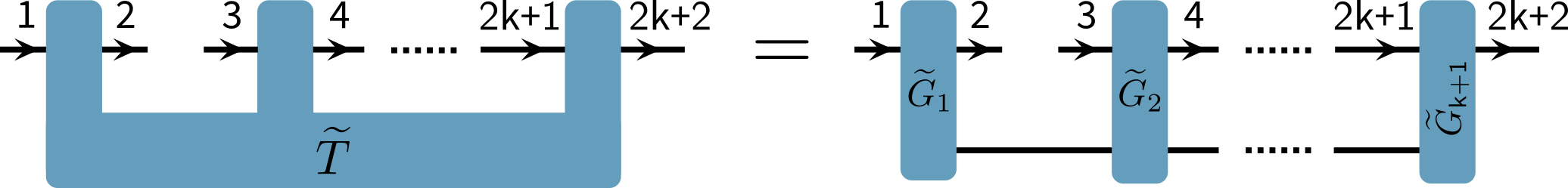}
    \caption{\textbf{$k$-slot combs.} A $k$-slot comb can always be decomposed as an ordered sequence of quantum channels.}
    \label{fig::kslot_comb_ordered}
\end{figure}
\end{example}

\subsection{Transformations without a fixed causal order}

All quantum transformations addressed in previous could be decomposed in terms of an ordered quantum circuit. For this reason, they are referred to as objects with a fixed causal order.
As proven in Ref.~\cite{chiribella09networks}, a quantum transformation admits a decomposition as a fixed ordered quantum circuit if it is a quantum $k$-slot comb. This motivates the definition of transformations with a fixed order.
\begin{definition}
A quantum transformation $\map{T}$ has a fixed causal order if it can be written as a $k$-slot comb for some well-suited dimensions, where we may also set some of the dimensions to be one.
\end{definition}
While the concept of fixed causal order is well-established, there are non-equivalent definitions for which transformations have a \textit{definite} causal order. One of the first articles to define definite causal order considered the scenario of process matrices which transform a pair of independent channels to unit probability~\cite{oreshkov11} (this case is discussed in the next example). There, a quantum transformation does not have a definite causal order if it cannot be written as a convex combination of quantum transformations with fixed order, a definition which is well-established for this `bipartite' process matrix scenario. When more general quantum transformations are considered, e.g., transformations of a pair of quantum channels into a channel, different arguments led to alternative notions of indefinite causal order ~\cite{oreschkov2016definition,Wechs2019MultipartiteCausality}. Additionally, when more than two parties are involved, the concept of classical dynamical control of causal order plays a non-trivial role~\cite{wechs21control}. Here, we will stay clear of these (important) subtleties in the definition of causal definiteness, and only examine if the obtained transformations can lie outside the set of \textit{fixed} causal order processes, i.e., outside the set of quantum combs.

\begin{example}[Non-signalling channels to unit probability (Process Matrix)]
As a pertinent example, let us consider the well-studied case of (bipartite) process matrices~\cite{oreshkov11, araujo15witnessing}, i.e., the set of transformations $\map{T}$ that map pairs of CPTP maps to unit probability (this, in turn, implies that they constitute the dual affine set of the set of (bipartite) non-signalling CPTP maps, see Sec.~\ref{sec::DualAff}). Specifically, let $R_{12}$ and $N_{34}$ be the Choi states of CPTP maps $\map{R}: \L(\H_1) \rightarrow  \L(\H_2)$ and $\map{N}: \L(\H_3) \rightarrow  \L(\H_4)$. Then, the set of process matrices is given by all linear maps $\map{T}: \L(\H_1 \otimes \H_2\otimes \H_3 \otimes \H_4) \rightarrow \mathbbm{C}$ such that $\map{T}[R_{12} \otimes N_{34}] = 1$ for all CPTP maps $R_{12}, N_{34}$. In this case, the input space is given by $\H_\textup{\inp} = \H_1\otimes \H_2 \otimes \H_3 \otimes \H_4$, while $\H_\textup{\out} = \mathbbm{C}$. The corresponding projectors simply follow from the previous examples as $\map{P}_\textup{\inp} = \map{P}^{(C)}_{12} \otimes \map{P}^{(C)}_{34}$ and $\map{P}_\textup{\out} = \map{\id}$ (these are, again, self-adjoint and unital projectors that commute with the transpose, such that Thm.~\ref{thm::genFormProj} can be applied). More explicitly, we have
\begin{align}
\map{P}_\inp[M]=& \map{P}^{(C)}_{12} \otimes \map{P}^{(C)}_{34}[M] \\
=& \map{P}^{(C)}_{12}[M-{}_4M+{}_{34}M] \\
=& (M-{}_4M+{}_{34}M) -{}_2(M-{}_4M+{}_{34}M) +{}_{12}(M-{}_4M+{}_{34}M)  \\
=&M-{}_4M+{}_{34}M -{}_2M +{}_{24}M -{}_{234}M +{}_{12}M -{}_{124}M +{}_{1234}M.
\end{align}
 With this, we can employ Eqs.~\eqref{eqn::GenProj1} and~\eqref{eqn::GenProj2} to obtain the properties of process matrices $T\in \L(\H_1\otimes \H_2 \otimes \H_3 \otimes \H_4)$ that send pairs of CPTP maps to unit probability:
\begin{align}
\hspace*{-12mm}
\setlength{\arrayrulewidth}{1pt}
\renewcommand{\arraystretch}{1.2}
\hspace{0.6cm}\begin{tabular}{|c|c|c|}
\hhline{|-|~|-|}
\bf{(Non-signalling channel)} & \hspace*{6mm}& \bf{(Complex number)}\\
$M_{}\in \L(\H_1\otimes\H_2\otimes\H_3\otimes\H_4) $ belongs to $\Scal_\inp$ iff  & \hspace*{6mm}& $c\in \L(\mathbb{C}) $  belongs to $\Scal_\out$ iff  \\
$M_{}\geq0$, && $c\geq0$, \\
$\map{P}_\inp[M]:=M-{}_2M  - {}_4M + {}_{24}M + {}_{34}M- $ &\hspace*{2mm} \Large{$\mathbf{\longrightarrow}$}  \hspace*{2mm}& $\map{P}_\out[c]:=c$, \\
$\phantom{\map{P}_\inp[M]}{}_{234}M + {}_{12}M - {}_{124}M+{}_{1234}M =M_{}$, && $\tr[c]=1$.\\
$\tr[M_{}]=d_1d_3$. &&  \\
\hhline{|-|~|-|}
\end{tabular}
\end{align} 
\begin{subequations}\label{eqn::ProcMat}
\begin{empheq}[box=\fcolorbox{blu}{white}]{align}
 \text{\bf{(Process}} & \text{  \bf{matrix)}}\nonumber \\
T\geq\,&0 ,\\
     \map{P}_{\inp\out}[T]:=\,& {}_2T  + {}_4T - {}_{24}T - {}_{34}T   + {}_{234}T - {}_{12}T + {}_{124}T =T,\\
 \tr [T] =\,& d_2d_4\,.
\end{empheq}
\end{subequations}
The above properties of $T$ exactly coincide with the characterisation of process matrices given in Ref.~\cite{araujo15witnessing}. \hfill $\blacksquare$
\label{ex::ProcMat}
\begin{figure}[t!]
    \centering
    \subfigure[~]
    {
    \includegraphics[scale=0.65]{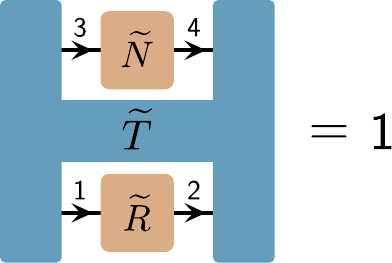}
    \label{fig::ProcMat1}} \hspace{2cm}
    \subfigure[~]
    {
    \includegraphics[scale=0.65]{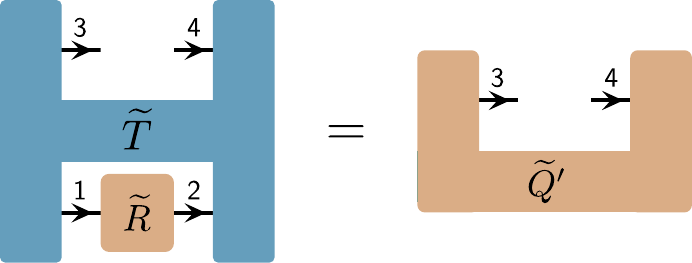}
    \label{fig::ProcMat2}} 
    \caption{(a) \textbf{Process Matrix.} As a mapping from two independent channels $\map{R}$ and $\map{N}$ to the number $1$ (b) \textbf{Process Matrix.} As a mapping from a channel $\map{R}$ to a one-slot superchannel $\map{Q}'$ without past and future.}
    \label{fig::ProcMa}
\end{figure}
\end{example}
In particular this latter result is of interest, since the set of process matrices can be considered the dual affine set of the set of all tensor products of CPTP maps, where the dual affine of a set are all operators that map the elements of the set to $1$.\footnote{More generally, any process matrix $T$ will map any matrix of the form $\sum_i \lambda_i M^{(i)} \otimes N^{(i)}$, where $\sum_i \lambda_i = 1$, and $M^{(i)}, N^{(i)}$ CPTP to $1$. Since the set of all valid CPTP maps that can be decomposed in this way is the set of non-signalling maps~\cite{Gutoski09, chiribella09_switch}, process matrices form exactly the dual set of non-signalling maps. We will further investigate this connection in Sec.~\ref{sec::DualAff}.} Such dual affine sets play an important role in quantum mechanics (and more generally, linear algebra), and evidently, the projectors we introduced can be used to characterise them comprehensively. Below, we will analyse the characterisation of dual sets (affine or not) in more detail. Before doing so, we provide an alternative characterisation of process matrices. 

\begin{example}[Quantum channels to superchannels without initial input and final output (Process matrices revisited)]
As a penultimate example, let us consider process matrices from a different perspective. Interestingly, they can be characterised in an alternative (yet equivalent) way, namely as mappings that map quantum channels (say, $\map{R}: \L(\H_1) \rightarrow \L(\H_2)$) to one-slot combs with no initial input and final output, i.e., superchannels with trivial input and output space (see Fig.~\ref{fig::ProcMat2}). Concretely, this requirement reads $\map{T}[R] = {\map{Q}}'$, where ${\map{Q}}'$ is a one-slot comb whenever $\map{R}$ is a CPTP map. Since such one-slot combs are special cases of superchannels, they are characterised by the projector of Eq.~\eqref{eqn::supChanProj} and they have a fixed causal order, that is, they can be represented by a causally ordered quantum circuit. This latter fact, in turn, chimes nicely with the intuitive definition of process matrices as mappings that obey local but not necessarily global causality~\cite{oreshkov11}; considering the slot corresponding to $\H_1$ and $\H_2$ as Alice's laboratory, independent of what deterministic operation (i.e., CPTP map) she performs locally, Bob (i.e., the slot corresponding to $\H_3$ and $\H_4$) will always encounter a causally ordered scenario (given by the one-slot comb $\map{Q}'$). Naturally, one would obtain the same definition of process matrices with the roles of Alice and Bob reversed. 

Let us now show that this alternative definition of process matrices indeed leads to the same characterisation as the one provided in the previous Example. First, since $R_{12} \in \L(\H_1\otimes \H_2)$ and $Q'_{34} \in \L(\H_3\otimes \H_4)$, we identify $\H_\textup{\inp} \cong \L(\H_1\otimes \H_2)$ and $\H_\textup{\out} \cong \L(\H_3\otimes \H_4)$. The projector $\map{P}_\textup{\inp}$ is given by the projector $\map{P}^{(C)}_{12}$ on the space of channels, while $\map{P}_\textup{\out}$ follows directly from the projector onto the set of superchannels provided in Eq.~\eqref{eqn::supChanProj} by setting $\H_1 \cong \H_4 \cong \mathbbm{C}$, such that $\map{P}_\textup{\out}[T] = {}_4T$. Additionally, we have that $\gamma_\inp = \tr[R] = d_2$, while $\gamma_\out = \tr [Q'] = d_3$ (see Eq.~\eqref{eqn::Superchannel}). Since all involved projectors are self-adjoint, unital, and commute with the transposition, Thm.~\ref{thm::genFormProj} applies, and we obtain the characterisation of T as
\begin{align}
\setlength{\arrayrulewidth}{1pt}
\renewcommand{\arraystretch}{1.2}
\hspace{-0.3cm}\begin{tabular}{|c|c|c|}
\hhline{|-|~|-|}
\bf{(Quantum channel)} & \hspace{0mm}& \bf{(Superchannel without past and future)}   \\
$R_{}\in \L(\H_1\otimes\H_2) $ belongs to $\Scal_\inp$ iff  & \hspace*{10mm}& $Q'_{}\in \L(\H_3\otimes\H_4) $  belongs to $\Scal_\out$ iff  \\
$R_{}\geq0$, & \Large{$\mathbf{\longrightarrow}$} & $Q_{}'\geq0$, \\
$\map{P}_\inp[R]:=R_{}-{}_2R_{}+{}_{12}R_{}=R_{}$, && $\map{P}_\out[Q']:= {}_4Q'_{}=Q'_{}$, \\
$\tr[R_{}]=d_1$. && $\tr[Q']=d_3$. \\
\hhline{|-|~|-|}
\end{tabular}
\end{align} 
\begin{subequations}\label{eqn::ProcMat2}
\begin{empheq}[box=\fcolorbox{blu}{white}]{align}
 \text{\bf{(Process}} & \text{  \bf{matrix revisited)}}\nonumber \\
T\geq\, &0,\\
\label{eqn::ProcMatProjector}
     \map{P}_{\inp\out}[T]:=\, & {}_2T  + {}_4T - {}_{24}T - {}_{34}T   + {}_{234}T - {}_{12}T + {}_{124}T =T,\\
 \tr [T] =\, & d_2d_4\,,
\end{empheq}
\end{subequations}
which coincides exactly with the characterisation of process matrices given in Eq.~\eqref{eqn::ProcMat}. Besides yielding an equivalent characterisation of process matrices, the above derivation also sheds an interesting light on the emergence of causal indefiniteness; graphically, a mapping from CPTP maps to one-slot combs is very similar to a mapping from CPTP maps to CPTP maps (i.e., superchannels), with the only difference that the incoming and outgoing wires of the former case are inverted with respect to the latter (to see this, compare Figs.~\ref{fig::SuperChannel} and~\ref{fig::ProcMat2}). 
This graphical similarity notwithstanding, \textit{all} one-slot superchannels have a fixed causal order, while process matrices do not always have an ordered quantum circuit decomposition (since the projector in Eq.~\eqref{eqn::ProcMatProjector} is not onto the space of $2$-slot combs). In particular, process matrices can be causally non-separable, that is, they cannot be written as a convex combination of ordered quantum circuits, or even as a quantum circuit with classical dynamical order~\cite{wechs21control}, and may even violate causal inequalities~\cite{oreshkov11, branciard_simplest_2015}.

Finally, let us remark that the equivalence between the two characterisations of process matrices ceases to hold if the trace-rescaling property is dropped. In this case, the requirement that process matrices map non-signalling maps to $\mathbbm{C}$(i.e., the case considered in the previous Example) yields no restrictions on the corresponding map $T'$, i.e., $\map{P}_\textup{\inp \out} = \map{\id}_\textup{\inp \out}$ (as can be seen by direct insertion into Eq.~\eqref{eqn::LinProj1} in Sec.~\ref{sec::ProofLink}). On the other hand, dropping the trace-rescaling conditions on maps that map CPTP maps to the space spanned by one-slot combs (i.e., the ones considered in this Example), we obtain, using Thm.~\ref{thm::LinSpaceProj}: 
\begin{gather}
    T' = {}_2T' + {}_4T' -{}_{24}T' - {}_{12}T' + {}_{124}T',
\end{gather}
which is a non-trivial constraint on the map $T'$.\hfill $\blacksquare$
\end{example}

\begin{example}[Superchannels to superchannels]
We now discuss an a priori more involved case that features less prominently in the literature: mappings from superchannels to superchannels (see Fig.~\ref{fig::suptosup}). Above, we already derived the projector onto the space of superchannels as well as their trace (see Eqs.~\eqref{eqn::Superchannel}). 
\begin{figure}[ht!]
    \centering
    \includegraphics[scale=0.65]{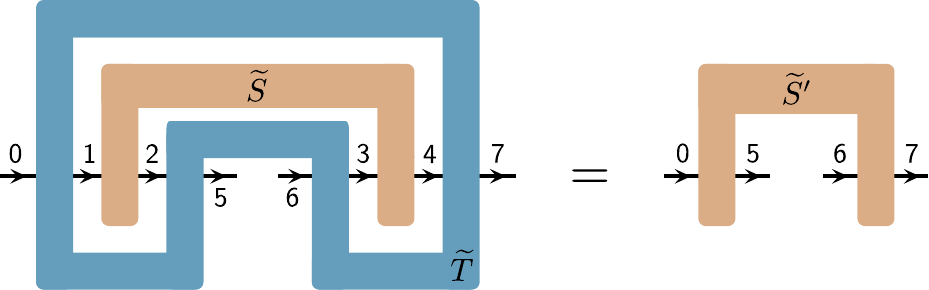}
    \caption{\textbf{Superchannel to superchannel.} The higher-order quantum map $\widetilde T$ (with Choi matrix $T \in \L(\H_0 \otimes \H_1 \otimes\H_2 \otimes\H_5 \otimes\H_6 \otimes\H_3 \otimes\H_4 \otimes \H_7)$ maps superchannels $\widetilde S$ (with Choi state $S \in \L(\H_1 \otimes \H_2 \otimes\H_3 \otimes\H_4)$) onto superchannels ${\widetilde S}'$ (with Choi state $S'\in \L(\H_0 \otimes \H_5 \otimes \H_6 \otimes \H_7$)).}
    \label{fig::suptosup}
\end{figure}
Here, proper bookkeeping of the involved spaces becomes slightly involved, but the respective properties of mappings from superchannels to superchannels can be readily deduced, using Thm.~\ref{thm::genFormProj}. Specifically, following the labelling convention of Fig.~\ref{fig::suptosup}, we set $\H_\inp := \H_1 \otimes \H_2 \otimes \H_3 \otimes \H_4$ and $\H_\out := \H_0 \otimes \H_5 \otimes \H_6 \otimes \H_7$. Consequently, for the Choi matrices $S$ ($S'$) of the input (output) superchannels we have $S \in \L(\H_\inp)$ ($S' \in \L(\H_\out)$ ), while the Choi matrix $T$ of the transformation between them acts on $\H_\inp \otimes \H_\out$. Now, using Thm.~\ref{thm::genFormProj}, we obtain
\begin{align}
\setlength{\arrayrulewidth}{1pt}
\renewcommand{\arraystretch}{1.2}
\hspace{-0.3cm}\begin{tabular}{|c|c|c|}
\hhline{|-|~|-|}
\bf{(Superchannel)} & \hspace{0mm}& \bf{(Superchannel)}   \\
$S_{}\in \L(\H_1 \otimes \H_2 \otimes \H_3 \otimes \H_4) \in \Scal_\inp$ iff  & \hspace*{10mm}& $S'_{}\in \L(\H_0 \otimes \H_5 \otimes \H_6 \otimes \H_7) \in \Scal_\out$ iff  \\
$S_{}\geq0$, & \Large{$\mathbf{\longrightarrow}$} & $S_{}'\geq0$, \\
$\map{P}_\inp[S]:=S - {}_4S   + {}_{34}S $,
 && $\map{P}_\out[S']:= S' - {}_7S'   + {}_{67}S'  $, \\
 \phantom{$\map{P}_\inp[S]:====$} $- {}_{234}S + {}_{1234}S =S$, && \phantom{$\map{P}_\inp[S]:====$}$- {}_{567}S' + {}_{0567}S' =S'$, \\
$\tr[S_{}]=d_1d_3$. && $\tr[S']=d_0d_6$. \\
\hhline{|-|~|-|} 
\end{tabular}
\end{align} 
\begin{subequations}
\begin{empheq}[box=\fcolorbox{blu}{white}]{align}
 \text{\bf{(Mapping bet}} & \text{\bf{ween superchannels)}}\nonumber \\
T\geq& \ 0, \\
     \map{P}_{\inp\out}[T]:=& \ T - {}_7T + {}_{47}T + {}_{67}T - {}_{347}T - {}_{467}T - {}_{567}T 
     + {}_{2347}T \nonumber \\ 
     \phantom{:=} &+ {}_{3467}T + {}_{4567}T - {}_{12347}T - {}_{23467}T - {}_{34567}T + {}_{123467}T, \nonumber \\
      \phantom{:=} &+ {}_{234567}T - {}_{1234567}T + {}_{01234567}T = T. \label{eqn::suptosupproj}\\
 \tr [T] =& d_0d_6d_2d_4\,.
\end{empheq}
\end{subequations}
While not a priori particularly insightful (albeit indispensable when numerically optimising over transformations of superchannels) in its own right, Eq.~\eqref{eqn::suptosupproj} allows one to directly deduce that transformations from superchannels to superchannels do not necessarily display a fixed causal order, i.e., they are not limited to quantum combs and cannot necessarily be implemented by means of a quantum circuit. In particular, counting input and output spaces, if $T$ corresponded to a supermap with a fixed causal order, it would have to satisfy the properties of a $3$-slot comb (see Fig.~\ref{fig::suptosup}). For example, for the space $\H_7$ to be the final output space of such a $3$-slot comb, $T$ would have to satisfy ${}_7T = {}_{\xt7}T$, where $\xt \in \{0,2,4,6\}$. From Eq.~\eqref{eqn::suptosupproj}, we directly see that this is not the case for any $\xt$, and, for instance, we have ${}_7T - {}_{47}T = {}_{67}T - {}_{467}T - {}_{567}T + {}_{4567}T \neq 0$ and analogously for $\xt = 0,2,6$. In a similar vein, this can be checked for the other potential final output spaces $\{1, 5, 3\}$, with the same result. Consequently, there exists  valid general maps from superchannels to superchannels which do not have a fixed causal order. 
\hfill $\blacksquare$
\end{example}

\begin{example}[Unital channels to unital channels]
As a last example, we will now consider transformations mapping unital channels to unital channels. A linear map $\map{C}:\L(\H_\inp)\to\L(\H_\out)$ is unital if it preserves the identity operator, i.e., $\map{C}(\id_\inp)=\id_\out$, in Choi picture if $C\in\L(\H_\inp\otimes\H_\out)$ is the Choi operator of the channel $\map{C}$, $\map{C}$ is unital if and only if $\tr_\inp(C)=\id_\out$. Quantum channels which are unital are also referred to it as \textit{bistochastic channels}. Transformations between them were first analysed in Ref.~\cite{chiribella20timeflip}, where the authors introduced the ``quantum time flip'', object which was explored in an experimental context in Ref.~\cite{Stromberg2024timeflip,Guo2024timeflip}.

Here, it is important to notice that, since unital maps do not span the set of all CPTP maps, the most general quantum transformation mapping unital channels into unital channels may transform quantum channels into objects which are \textit{not} quantum channels. For concreteness, we now illustrate this fact with the example discussed in Ref.~\cite{chiribella20timeflip}. Let $\H_\inp := \H_1 \otimes \H_2$ and $\H_\out := \H_0 \otimes\H_3$, and $\map{T}:\L(\H_\inp)\to\L(\H_\out)$ be defined as $\map{T}(C)=FCF^\dagger$, where $F:\H_1\to\H_2$ is the flip operator (also known as swap operator) defined as $F:=\sum_{ij}\ketbra{ij}{ji}$. We notice that, if we apply $\map{T}$ on unital channels $B\in\L(\H_\inp)$, the output $\map{T}(B)$ is a quantum channel. However, if we apply the map $\map{T}$ to non-unital channels, such as a trace and replace one $C=\id_1\otimes \ketbra{0}{0}$, we have that $\map{T}(\id_1\otimes \ketbra{0}{0}_2)=\ketbra{0}{0}_1\otimes \id_3$, which is not a quantum channel, because $\tr_3(\ketbra{0}{0}_0\otimes \id_3)\neq \id_0$. Hence, such class of transformations go beyond the process matrix formalism~\cite{oreshkov11} but can still be characterised via Thm.~\ref{thm::genFormProj}, illustrating the generality of our methods.

Similarly to the other cases, we now invoke Thm.~\ref{thm::genFormProj}, to obtain
\begin{align}
\setlength{\arrayrulewidth}{1pt}
\renewcommand{\arraystretch}{1.2}
\hspace{-0.3cm}\begin{tabular}{|c|c|c|}
\hhline{|-|~|-|}
\bf{(Unital channel)} & \hspace{0mm}& \bf{(Unital channel)}   \\
$C_{}\in \L(\H_1 \otimes \H_2) \in \Scal_\inp$ iff  & \hspace*{10mm}& $C'_{}\in \L(\H_0 \otimes \H_3) \in \Scal_\out$ iff  \\
$C_{}\geq0$, & \Large{$\mathbf{\longrightarrow}$} & $C_{}'\geq0$, \\
$\map{P}_\inp[C]:=C - {}_1C - {}_2C  + 2{}_{12}C =C$,
 && $\map{P}_\out[S']:= C'- {}_0C' - {}_3C'  + 2{}_{03}C' =C'  $, \\
$\tr[S_{}]=d_1$. && $\tr[S']=d_0$. \\
\hhline{|-|~|-|} 
\end{tabular}
\end{align} 
\begin{subequations}
\begin{empheq}[box=\fcolorbox{blu}{white}]{align}
 \text{\bf{(Mapping b}} & \text{\bf{etween unital channels)}}\nonumber \\
T\geq& \ 0, \\
     \map{P}_{\inp\out}[T]:=& \ T-_0T+_{01}T+_{02}T-2_{012}T-_3T+_{13}T+_{23}T-2_{123}T\nonumber \\ 
     \phantom{:=} &+_{03}T-_{013}T-_{023}T+3_{0123}T  =T \\
 \tr [T] =& d_0d_3\,.
\end{empheq}
\end{subequations}

When compared to the appendix ``Characterisation of the operations on bistochastic channels'' of Ref.~\cite{chiribella20timeflip}, the equations above provide an alternative (but equivalent) characterisation of transformations between unital channels 
\hfill $\blacksquare$
\end{example}

\section{`Completeness' of Quantum properties}
\label{sec::Completeness}
\subsection{Completely admissible transformations}

Up to this point, we have discussed the properties of transformations $\widetilde T_{\inp\out}$ whenever they act non-trivially on the \textit{full} input space $\L(\H_\inp)$. However, just like in the case of positivity, one may wonder if a transformation is also admissible (in a sense to be defined below) when only acting non-trivially on a part of an input object $W_{\inp\att} \in \L(\H_\inp \otimes \H_\att)$. To give a simple concrete example, a trace-preserving map $\widetilde T_{\inp \out}$ is also trace preserving when only acting non-trivially on a part of a quantum state $\rho_{\inp \att} \in \L(\H_{\inp} \otimes \H_{\att})$, i.e., $\tr[\widetilde T_{\inp \out}\otimes \widetilde{\id}_\att[\rho_{\inp \att}]] = \tr[\rho_{\inp \att}]$ for all auxiliary spaces $\L(\H_\att)$. A priori, though, it is unclear if, and even in what sense, this `completeness' holds for more involved cases, like the ones discussed in the previous section. To tackle this question, we first require the notion of an \textit{extension} of a projector:
\begin{definition}[Extension of a projector]
    \label{def::Extension}
    For a given projector $\widetilde{P}_\xtt$ on $\L(\H_\xtt)$, we call a family of projectors $\{\widetilde{P}_{\xtt\att}\}_\att$ on the respective spaces $\L(\H_{\xtt} \otimes \H_\att)$ an extension of $\widetilde{P}_\xtt$ if it satisfies $\widetilde P_{\xtt \att} \cong \widetilde{P}_\xtt$ when $\H_\att \cong \mathbbm{C}$. Additionally, whenever $\widetilde{P}_\xtt$ is self-adjoint (unital, commuting with the transposition), all projectors of the extension are assumed to be self-adjoint (unital, commuting with the transposition)
\end{definition}
 Whenever there is no risk of confusion, we also call the individual projectors $\widetilde P_{\xtt \att}$ an extension of $\widetilde P_\xtt$, with the understanding that they are an element of a whole set of projection operators. A priori, the above definition occurs somewhat void of meaning, since it does not really restrict the set of extensions $\{\widetilde{P}_{\xtt\att}\}_\att$. We emphasise that this is by design, since we aim to remain agnostic with respect to how one chooses a particular extension in a concrete physical scenario; choosing a more restrictive definition of an extension\footnote{A seemingly obvious additional restriction would be of the form $\tr_\att \circ \map{P}_{\xtt\att} \propto \map{P}_{\xtt}$. While this is the case in all the examples we provide, this assumes a particular relationship between all the spaces defined by $\{\map{P}_{\xtt\att}\}_\att$, which we do not fundamentally require for our proofs/constructions.} might  exclude many cases of interest. Since this level of generality does not pose any additional technical issues, we thus postpone more explicit definitions of the extensions to the examples below (see Exs.~\ref{ex::supchann_comp} and~\ref{ex::ProcMatComp}). 
 
 For the moment, as a concrete simple example, one can consider the set of quantum states on $\L(\H_\xtt) \equiv \L(\H_1)$ defined by the projector $\widetilde P_\xtt = \widetilde{\id}_1$. The natural extension of this projector to projectors on the spaces $\L(\H_1 \otimes \H_\att)$ for arbitrary auxiliary spaces $\L(\H_\att)$ is the set $\{\widetilde{\id}_{1} \otimes \widetilde{\id}_{\att}\}_\att$. 
 
 Analogously, considering the set of CPTP maps $\widetilde T: \L(\H_1) \rightarrow \L(\H_2)$ defined by the projector $\widetilde P_{\xtt}[T] = T - {}_2T + {}_{12}T$, with $\H_{\xtt} \equiv \H_1 \otimes \H_2$, we obtain the `natural' extension to projectors onto the sets of CPTP maps $\widetilde T: \L(\H_1 \otimes \H_{\att_1}) \rightarrow \L(\H_2 \otimes \H_{\att_2})$ acting on an extended space as 
 \begin{gather}
     \{\widetilde P_{\xtt\att}|\widetilde P_{\xtt\att}[X] = X - {}_{2\att_2}X + {}_{1\att_1 2\att_2}X, \quad \forall X \in \L(\H_\xtt \otimes \H_{\att})\}_\att,
 \end{gather}
where $\H_\att \equiv \H_{\att_1} \otimes \H_{\att_2}$. 
 
Intuitively, $\widetilde P_{\xtt \att}$ can be considered as the `version of' $\widetilde P_\xtt$ on a larger space. However, the concrete choice of extension is not always unique and can depend on the physical situation that is considered (see Ex.~\ref{ex::ProcMatComp} below). In the literature, the question of completeness has already 
appeared under the name of `completely admissible', and it was shown that different well-motivated extensions lead to different sets of completely admissible transformations~\cite{milz_resource_2022}. While considering more restricted scenarios, the question of completely admissible was addressed for bipartite process matrices with process and future applied on local bipartite channels~\cite{araujo16purification} and for performing measurements on quantum channels~\cite{burniston_necessary_2020}.
Additionally, different ways to extend and compose quantum systems were introduced, as so-called non-signalling extensions $\widetilde P_\xtt \rightarrow \widetilde P_{\xtt} \otimes \widetilde P_{\att}$ for some projectors $\widetilde P_\att$~\cite{hoffreumon_projective_2022}, the 'prec' $\prec$ extension \cite{hoffreumon_projective_2022}, as well as full signalling 'par' $\rotatebox[origin=c]{180}{\&}$ extension 
\cite{simmons_higher-order_2022} and others~\cite{simmons_higher-order_2022}. While these mentioned examples fit into Def.~\ref{def::Extension}, here, for the moment, we shall not assume an explicit functional form of $\widetilde P_{\xtt \att}$ with respect to $\widetilde P_{\xtt}$ and rather leave their explicit structure a priori unspecified. With this, we are now in a position to define the notion of 'completely admissible' with respect to an extension and generalise the definition first presented at Ref.~\cite{milz_resource_2022}:
\begin{definition}[Completely admissible]
\label{def::CompStrucPres}
Let $\map{P}_\inp:\L(\H_\inp)\to\L(\H_\inp)$ and  $\map{P}_\out:\L(\H_\out)\to\L(\H_\out)$ be linear projective maps with respective extensions $\{\widetilde P_{\inp\att}\}_\att$ and $\{\widetilde P_{\out\att}\}_\att$. Let $\Scal_{\inp\att} \subseteq \L(\H_{\inp\att}) $ and $ \Scal_{\out\att} \subseteq \L(\H_{\out\att})$ be sets of quantum objects defined by
\begin{align}
\setlength{\arrayrulewidth}{1pt}
\renewcommand{\arraystretch}{1.2}
\begin{tabular}{|c|c|c|}
\arrayrulecolor{brow}
\hhline{|-|~|-|}
$W_{\inp\att}\in \L(\H_{\inp\att}) $ belongs to $\Scal_{\inp\att}$ iff  & \hspace*{20mm}& $W'_{\out\att}\in \L(\H_{\out\att}) $ belongs to $\Scal_{\out\att}$ iff  \\
$W_{\inp\att}\geq0$, &\hspace*{5mm} \Large{$\mathbf{\longrightarrow}$}  \hspace*{5mm} & $W'_{\out\att}\geq0$, \\
$\map{P}_{\inp\att}[W_{\inp \att}]=W_{\inp\att}$, && $\map{P}_{\out\att}[W'_{\out\att}]=W'_{\out\att}$,\\
$\tr[W_{\inp \att}] = \gamma_{\inp\att}$. && $\tr[W'_{\out \att}] = \gamma_{\out\att}$.\\
\hhline{|-|~|-|}
\end{tabular}
\end{align} 
for all auxiliary spaces $\H_\att$. A linear map $\map{T}_{\inp\out}:\L(\H_\inp)\to\L(\H_\out)$ is completely admissible with respect to the extensions $\{\widetilde P_{\inp\att}\}_\att$ and $\{\widetilde P_{\out\att}\}_\att$ if:
\begin{subequations}
\arrayrulecolor{blu}
\begin{empheq}[box=\fcolorbox{blu}{white}]{align}
	i:& \quad \map{T}_{\inp\out} \text{ is completely positive }\\
	ii:& \quad \forall W_{\inp\att}\in\Scal_{\inp\att}, \text{ we have that } (\map{T}_{\inp\out}\otimes \map{\id}_\att)[W_{\inp\att}]\in\Scal_{\out\att}
\end{empheq}
\end{subequations}
\end{definition}
In principle, depending on the respective extensions $\{\widetilde P_{\inp\att}\}_\att$ and $\{\widetilde P_{\out\att}\}_\att$, this definition might not yield \textit{any} transformations $\widetilde T_{\inp\out}$ that satisfy it. However, for all cases we consider (and all relevant cases in the literature), this is not the case, and we thus opt for the notational simplicity the generality of the above definition provides.

Notably, Def.~\ref{def::CompStrucPres} restricts the transformations $\map{T}_{\inp\out} \otimes \map{\id}_\att$ similarly to the requirements of a quantum transformation imposed in Def.~\ref{def:quantum_object}. Consequently, we can use a similar reasoning as the one that led to Thm.~\ref{thm::genFormProj} to obtain the properties of $\map{T}_{\inp\out}\otimes \map{\id}_\att$:

\begin{lemma}[Completely admissible transformations: specialised Choi version]
\label{lem::CompStrucPresTransf}
Let $\map{P}_\inp:\L(\H_\inp)\to\L(\H_\inp)$ and  $\map{P}_\out:\L(\H_\out)\to\L(\H_\out)$ be linear projective maps with respective extensions $\{\widetilde P_{\inp\att}\}_\att$ and $\{\widetilde P_{\out\att}\}_\att$. Let $\Scal_{\inp\att} \subseteq \L(\H_{\inp\att}) $ and $ \Scal_{\out\att} \subseteq \L(\H_{\out\att})$ be sets of quantum objects defined by
\begin{align}
\setlength{\arrayrulewidth}{1pt}
\renewcommand{\arraystretch}{1.2}
\arrayrulecolor{brow}
\begin{tabular}{|c|c|c|}
\hhline{|-|~|-|}
$W_{\inp\att}\in \L(\H_{\inp\att}) $ belongs to $\Scal_{\inp\att}$ iff  & \hspace*{20mm}& $W'_{\out\att}\in \L(\H_{\out\att}) $ belongs to $\Scal_{\out\att}$ iff  \\
$\map{P}_{\inp\att}[W_{\inp\att}]=W_{\inp\att}$, &\hspace*{5mm} \Large{$\mathbf{\longrightarrow}$}  \hspace*{5mm}& $\map{P}_{\out\att}[W'_{\out\att}]=W'_{\out\att}$, \\
$\tr[W_{\inp\att}] = \gamma_{\inp\att}$. && $\tr[W'_{\out\att}] = \gamma_{\out\att}$. \\
\hhline{|-|~|-|}
\end{tabular}
\end{align} 
Additionally, we assume that all the maps $\map{P}_{\inp\att}$ and $\map{P}_{\out\att}$ are self-adjoint and unital, and that $\map{P}_{\inp\att}$ commutes with the transposition map, \ie, for all $\H_\att$ we have
\begin{subequations}
\begin{empheq}[box=\widefbox]{alignat=3}
   &\map{P}_{\inp\att}=\map{P}_{\inp\att}^\dagger,  && \map{P}_{\out\att}=\map{P}_{\out\att}^\dagger, \\
   &\map{P}_{\inp\att}[\id]=\id,  && \map{P}_{\out\att}[\id]=\id, \\
   &\map{P}_{\inp\att}[W^\tau_{\inp\att}]=\map{P}_{\inp\att}&&[W_{\inp\att}]^\tau, \quad\forall W_{\inp\att}\in\L(\H_{\inp\att}).
\end{empheq}
\end{subequations}
A linear map $\map{T}_{\textup{\inp \out}}:\L(\H_\textup{\inp}) \rightarrow \L(\H_\textup{\out})$ is completely admissible with respect to the extensions $\{\widetilde P_{\inp\att}\}_\att$ and $\{\widetilde P_{\out\att}\}_\att$ if and only if
\begin{subequations}
\begin{empheq}[box=\fcolorbox{blu}{white}]{align}
\label{eqn::CompStrucPres1}
    &(\map{P}_\textup{\inp\att}\otimes\map{\id}_{\out\att'}) [T_{\textup{\inp \out}} \otimes \Phi^{+}_{\att \att'}] = (\map{P}_\textup{\inp\att} \otimes\map{P}_{\textup{\out}\att'}) [T_{\textup{\inp \out}} \otimes \Phi^{+}_{\att \att'}],  \\
\label{eqn::CompStrucPres2}
   \text{and} \quad &(\map{P}_{\inp\att}\otimes \map{\id}_\out)[{}_{\out\att}T_{\inp\out}] = {}_{\inp\out\att}T_{\inp\out} \quad \& \quad \tr[T_{\inp\out}] = \frac{\gamma_{\out\att}}{\gamma_{\inp\att}} d_\inp
\end{empheq}
\end{subequations}
holds for its Choi matrix $T_{\inp\out}$ for all $\H_\att$, where $\H_\att \cong \H_{\att'}$, $\Phi^{+}_{\att \att'}$ is the Choi matrix\footnote{The concrete form of $\Phi^{+}_{\att \att'}$ depends on the basis that is chosen for the CJI. In the examples below, this choice will always be clear from context.} of the identity map $\map{\id}_\att$ and $\widetilde P_{\out\att'} \cong \map{P}_{\out\att}$. 
\end{lemma}
The proof of the above Lemma follows along the same lines as that of Thm.~\ref{thm::genFormProj} and can be found in App.~\ref{app::Complete}. For improved clarity, here, in contrast to Thm.~\ref{thm::genFormProj}, we have split the properties of the transformation $T_{\inp\out}$ into two parts: Eq.~\eqref{eqn::CompStrucPres2}, which stems from the trace-rescaling properties, and Eq.~\eqref{eqn::CompStrucPres1}, which concerns the remaining structural requirements. We emphasise that Eq.~\eqref{eqn::CompStrucPres2} directly implies that $\gamma_{\out\att}/\gamma_{\inp\att} = \gamma_{\out}/\gamma_{\inp}$ has to hold for all $\H_\att$ in order to allow for the existence of an admissible transformation $T_{\inp\out}$. This requirement will be fulfilled in all the examples we consider below, and we will assume it throughout.
 
While only providing the constraints on $T_{\textup{\inp \out}} \otimes \Phi^{+}_{\att \att'}$, in many relevant cases, Lem.~\ref{lem::CompStrucPresTransf} allows one to directly deduce a characterisation of $T_{\textup{\inp \out}}$ alone. To see this, we make the following Observation: 

\begin{observation}
\label{obs::Complete}
 Let $\widetilde T_{\inp\out}$ be a completely admissible transformation (with respect to the extensions $\{\widetilde P_{\inp\att}\}_\att$ and $\{\widetilde P_{\out\att}\}_\att$) with $\tr[T_{\inp\out}] = \gamma_\out/\gamma_\inp \cdot d_\inp$. For a given $\H_\att$, Eqs.~\eqref{eqn::CompStrucPres1} and~\eqref{eqn::CompStrucPres2} yield a set of linear constraints on $T_{\inp \out}$ via
 \begin{align}
     \Omega_\att[T_{\inp\out}] := (\map{P}_\textup{\inp\att}\otimes\map{\id}_{\out\att'} - \map{P}_\textup{\inp\att} \otimes\map{P}_{\textup{\out}\att'}) [T_{\textup{\inp \out}} \otimes \Phi^{+}_{\att \att'}] = 0 \quad \Leftrightarrow \quad \sum_{\alpha = 1}^{d_\att^4} R_\att^{(\alpha)} \otimes \Delta^{(\alpha)}_{\att\att'} = 0,\\
     \Xi_\att[T_{\inp\out}] := (\map{P}_{\inp\att}\otimes \map{\id}_\out)[{}_{\out\att}T_{\inp\out}] - {}_{\inp\out\att}T_{\inp\out} = 0 \quad \Leftrightarrow \quad \sum_{\beta = 1}^{d_\att^4} Q_\att^{(\beta)} \otimes \Delta^{(\beta)}_{\att\att'} = 0,
 \end{align}
 where $\{\Delta^{(\alpha)}_{\att\att'}\}_\alpha^{d_\att^4}$ is an orthonormal Hermitian basis of $\L(\H_\att \otimes \H_{\att'})$, while $R_\att^{(\alpha)} = \tr_{\att\att'}[\Delta^{(\alpha)} \Omega_{\att}[T_{\inp\out}]]$ and $Q_\att^{(\beta)} = \tr_{\att\att'}[\Delta_{\att\att'}^{(\beta)} \Xi_{\att}[T_{\inp\out}]]$. In principle, this leads to a set of $2d_\att^4$ linear  equations 
 \begin{gather}
 \label{eqn::LinEqComp}
     R_\att^{(\alpha)} = 0  \quad \text{and} \quad Q_\att^{(\beta)} = 0, \quad \forall \alpha, \beta, \H_\att,
 \end{gather} 
 that can often be phrased as restrictions on $T_{\inp\out}$ for every space $\H_\att$, implying a potentially infinite number of linear constraints (that might not be satisfiable simultaneously). As we will see in the concrete examples below, in most practical cases, the physically motivated extensions $\{\widetilde P_{\inp\att}\}_\att$ and $\{\widetilde P_{\out\att}\}_\att$ are such that a) Eq.~\eqref{eqn::LinEqComp} can indeed be read as a direct requirement on $T_{\inp\out}$, b) the ensuing restrictions on $T_{\inp\out}$ can be satisfied concurrently, and c) it suffices to compute them for an arbitrary (non-trivial) fixed extension space $\H_\att$.
\end{observation}
Let us illustrate the above considerations with some concrete examples.
\begin{example}[Superchannels revisited: Completely trace preserving maps\footnote{Although not discussed in detail, App. C of Ref.~\cite{araujo16purification} uses related, but different methods, to obtain the same results as in this example.}]
\label{ex::supchann_comp}
\begin{figure}[ht!]
    \centering
    \includegraphics[width = 0.4\linewidth]{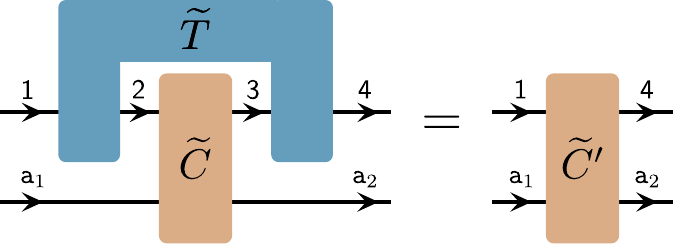}
    \caption{\textbf{Complete trace preservation.} A superchannel $\widetilde T$ should map TP maps to TP maps, even when only acting non-trivially on a part of them (here, the spaces $\H_1$ and $\H_2$). For the particular case of superchannels, this requirement does not yield any additional constraints on $\widetilde T$.}
    \label{fig::superchannel_Comp}
\end{figure}
For the case of superchannels, a natural requirement of completely admissibility is to demand that a superchannel $\widetilde T$ yields a TP map even when acting non-trivially on only a part of a TP map $\widetilde C: \L(\H_1\otimes \H_{\att_1}) \rightarrow \L(\H_2 \otimes \H_{\att_2})$ (see Fig.~\ref{fig::superchannel_Comp}). Thus, we have $\H_\inp = \H_2\otimes \H_3$, $\H_\out = \H_1\otimes \H_4$, and $\H_\att = \H_{\att_1} \otimes \H_{\att_2}$. As mentioned above, the projectors $\map{P}_{\inp\att}$ and $\map{P}_{\out\att'}$ are of the form 
\begin{gather}
    \map{P}_{\inp\att} [X]= X - {}_{3\att_2}X + {}_{2a_13\att_2}X \quad \text{and} \quad \map{P}_{\out\att'} [Y]= Y - {}_{4\att_2'}Y + {}_{1a'_14\att'_2}Y, 
\end{gather}
where $\H_{\att_1} \cong \H_{\att_1'}$  and $\H_{\att_2}  \cong \H_{\att_2'}$. The Choi state of the identity channel $\map{\id}_{\att\att'}$ is given by $\Phi^+_{\att\att'} := \Phi^+_{\att_1\att_1'} \otimes \Phi^+_{\att_2\att_2'}$, where $\Phi^+_{\att_1\att_1'}$ ($\Phi^+_{\att_2\att_2'}$) is the maximally entangled state in the computational basis of $\H_{\att_1}$ ($\H_{\att_2}$). We have $\gamma_{\inp\att} = d_2{d_{\att_1}}$ and $\gamma_{\out\att} = d_1{d_{\att_1}}$, such that $\gamma_{\out\att}/\gamma_{\inp\att} = d_1/d_2 = \gamma_\out/\gamma_\inp$ holds and we have $\tr[T] = d_1d_3$.  In addition, Eq.~\eqref{eqn::CompStrucPres1} yields 
\begin{gather}
    (-{}_4T + {}_{34}T) \otimes \Phi^+_{\att_1\att_1'} \otimes \id_{\att_2\att_2'}/d_{\att_2} + ({}_{14}T - {}_{134}T - {}_{234}T + {}_{1234}T) \otimes \id_{\att_1 \att_1'\att_2\att_2'} /(d_{\att_1}d_{\att_2}) = 0.
\end{gather}
By comparing linearly independent terms, as detailed in Obs.~\ref{obs::Complete}, we see that this implies:
\begin{gather}
    -{}_4T + {}_{34}T = 0 \quad \text{and} \quad {}_{14}T - {}_{134}T - {}_{234}T + {}_{1234}T = 0.
\end{gather}
It is easy to deduce that this is equivalent to 
\begin{gather}
\label{eqn::SupComp}
    T = T - {}_{4}T + {}_{34}T - {}_{234}T + {}_{1234}T, 
\end{gather}
which corresponds exactly to the projector onto the set of superchannels derived in Eq.~\eqref{eqn::supChanProj}. Following the same logic, from the trace-rescaling property Eq.~\eqref{eqn::CompStrucPres1} we obtain ${}_{14}T = {}_{134}T$, which is already implied by Eq.~\eqref{eqn::SupComp}.

Consequently, in this case, the trace-rescaling property does not add any additional restrictions on admissible superchannels. Overall, we thus obtain the (well-known) fact that superchannels as considered in Ex.~\ref{ex::supchann} are already completely admissible (in this case, completely TP-preserving). However, previous proofs of this fact either had to -- additionally -- demand complete positivity (in which case superchannels can be represented by a causally ordered circuit, making them completely TP preserving), or finding other bespoke ways to prove this statement. Here, using Lem.~\ref{lem::CompStrucPresTransf} and Obs.~\ref{obs::Complete}, we can derive the projector onto the set of completely admissible superchannels in a direct and systematic manner. Additionally, by treating the different requirements on completely admissible superchannels individually, we see that demanding a superchannel to be completely trace-rescaling does not yield any additional restrictions on $T_{\inp\out}$. Interestingly, this is in stark contrast to the derivation of the properties of superchannels (without the demand of complete admissibility), where the projector onto the set of superchannels would differ from Eq.~\eqref{eqn::supChanProj} if one did \textit{not} demand the trace-rescaling property.
 \hfill $\blacksquare$
 \end{example}

To see that complete admissibility can, indeed, add new restrictions on transformations, as well as depend on the choice of extension, we provide a second explicit example, this time considering process matrices.
\begin{example}[Process matrix revisited: Completely trace preserving and/or completely non-signalling preserving]
\label{ex::ProcMatComp}
\begin{figure}[ht!]
    \centering
    \subfigure[~]
    {
    \includegraphics[scale = 0.5]{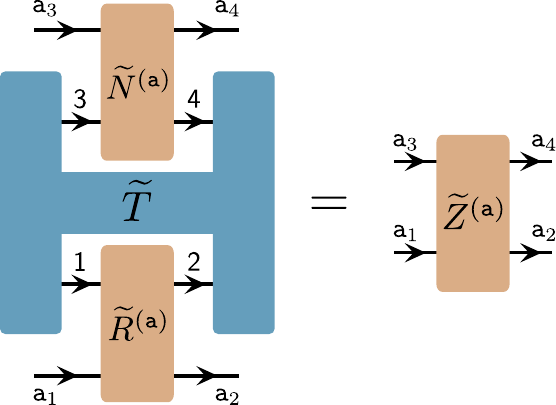}
    \label{fig::proc_mat_Comp}
    } \hspace{2cm}
    \subfigure[~]
    {
    \includegraphics[scale = 0.5]{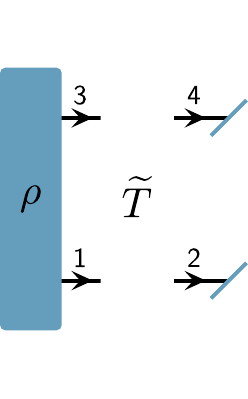}
    \label{fig::proc_mat_Comp2}
    }
    \caption{\textbf{Completely admissible process matrix.} \textbf{(a)} A (bipartite) process matrix maps pairs of channels to unit probability (with corresponding projector $\map{P}_\out = \map{\id}: \mathbbm{C} \rightarrow \mathbbm{C}$). When acting non-trivially only on a part of those channels, the resulting object is a mapping $\map{Z}^{(\att)}: \L(\H_{\att_1}\otimes \H_{\att_3}) \rightarrow \L(\H_{\att_2}\otimes \H_{\att_4})$. One possible extension $\map{P}^{(\mathrm{TP})}_{\out\att}$ follows from the requirement that $\map{Z}^{(\att)}$ is a valid channel. Another, more restrictive option for $\map{P}^{\mathrm{(NS)}}_{\out\att}$ would be the projectors onto channels $\map{Z}^{(\att)}$ that are non-signalling $\att_1 \nrightarrow \att_4$ and $\att_3 \nrightarrow \att_2$. Both are valid extensions and, depending on the physical situation that is considered can be the `correct' choice. \textbf{(b)} The only process matrices that are completely admissible with respect to $\map{P}^{\mathrm{(NS)}}_{\out\att}$ are state preparations followed by a discarding of the output system.}
\end{figure}
We have already discussed (bipartite) process matrices in Ex.~\ref{ex::ProcMat} as the set of transformations that map products of channels $\map{R}: \L(\H_1) \rightarrow \L(\H_2)$ and $\map{Q}: \L(\H_3) \rightarrow \L(\H_4)$ (and thus (bipartite) non-signalling channels\footnote{See Sec.~\ref{sec::NonSigChann} for a detailed discussion of non-signalling channels in the bi- and multi-partite case.}) to unit probability. The corresponding projector $\map{P}_\inp$ on $\L(\H_\inp) := \L(\H_1 \otimes \H_2 \otimes \H_3 \otimes \H_4)$ is given in Eq.~\eqref{eqn::ProcMat}. Its natural extension is the projector $\map{P}_{\inp\att}$ onto the set spanned by non-signalling channels $\L(\H_{\att_1} \otimes \H_1 \otimes \H_{3} \otimes \H_{\att_3}) \rightarrow \L(\H_{\att_2} \otimes \H_2 \otimes \H_{4} \otimes \H_{\att_4})$ with $1\att_1 \nrightarrow 4\att_4$ and $3\att_3 \nrightarrow 2\att_2$ (see Fig.~\ref{fig::proc_mat_Comp}). It can be obtained by the replacement $x \mapsto \att_x x$ for $x\in \{1, 2, 3, 4\}$ in Eq.~\eqref{eqn::ProcMat}, which yields
\begin{gather}
\begin{split}
    \map{P}_{\inp\att}[X]&:=X-{}_{\att_22}X  - {}_{\att_44}X + {}_{\att_22\att_44}X + {}_{\att_{3}3\att_44}X- {}_{\att_22\att_33\att_44}X \\
    &\phantom{;=}+ {}_{\att_11\att_22}X - {}_{\att_11\att_22\att_44}X+{}_{\att_11\att_22\att_33\att_44}X
\end{split}
\end{gather}
While this extension is natural (albeit not the only possible way), the extension of the output projector $\map{P}_\out := \map{\id}: \mathbbm{C} \rightarrow \mathbbm{C}$ is, a priori, not unique. Letting $\map{T}$ act non-trivially on only a part of `extended' channels $\map{R}^{(\att)} \otimes \map{Q}^{(\att)}$, with $\map{R}^{(\att)}: \L(\H_1 \otimes \H_{\att_1}) \rightarrow \L(\H_2 \otimes \H_{\att_2})$  and $\map{N}^{(\att)}: \L(\H_3 \otimes \H_{\att_3}) \rightarrow \L(\H_4 \otimes \H_{\att_4})$, yields a map $\map{Z}^{(\att)}: \L(\H_{\att_1} \otimes \H_{\att_3}) \rightarrow \L(\H_{\att_2} \otimes \H_{\att_4})$. The minimum requirement one could ask for, is that $\map{Z}^{(\att)}$ is trace-preserving, yielding the extended output projector
\begin{gather}
    \map{P}_{\out\att}^{(\mathrm{TP})}[X] = X - {}_{\att_2\att_4}X +{}_{\att_1\att_3\att_2\att_4}X.
\end{gather}
On the other hand, one might be more explicitly interested in the signalling properties of $\map{Z}^{(\att)}$, and require that it is non-signalling\footnote{For example, such a signalling constraint has been used in Ref.~\cite{milz_resource_2022} (in a slightly more general form) to define the notion of `completely admissible' transformations of process matrices.} $\att_1 \nrightarrow \att_4$ and $\att_3 \nrightarrow \att_2$, yielding the alternative projector extension
\begin{gather}
\begin{split}
    \map{P}_{\out\att}^{(\mathrm{NS})}[X]&:=X-{}_{\att_2}X  - {}_{\att_4}X + {}_{\att_2\att_4}X + {}_{\att_{3}\att_4}X- {}_{\att_2\att_3\att_4}X + {}_{\att_1\att_2}X - {}_{\att_1\att_2\att_4}X+{}_{\att_1\att_2\att_3\att_4}X,
\end{split}
\end{gather}
which is obtained from Eq.~\eqref{eqn::ProcMat} via the replacement $x \mapsto \att_x$ for $x\in \{1, 2, 3, 4\}$. Both of these choices satisfies $\gamma_{\out\att}/\gamma_{\inp\att} \cdot d_\inp = \gamma_{\out}/\gamma_{\inp} d_\inp = d_2 d_4$ as well as $\map{P}_{\out\att}^{(\mathrm{TP})} = \map{P}_{\out\att}^{(\mathrm{NS})} = \map{\id} = \map{P}_\out$ for $\H_{\att} = \mathbbm{C}$, and are thus, depending on the physical scenario one aims to investigate, valid extensions of $\map{P}_{\out}$. The resulting structural requirements on $T_{\inp \out}$ (besides $T_{\inp \out} \geq 0$ and $\tr[T_{\inp\out}] = d_2 d_4$) can, again, be computed using Lem.~\ref{lem::CompStrucPresTransf} and Obs.~\ref{obs::Complete}. 

For the case where the resulting map $\map{Z}^{(\att)}$ only needs to be trace-preserving, employing $\map{P}_{\inp\att}$ and $\map{P}_{\out\att}^{(\mathrm{NS})}$ in  Eq.~\eqref{eqn::CompStrucPres1} and comparing the linearly independent terms according to Obs.~\ref{obs::Complete} yields
\begin{gather}
\label{eqn::ProcMatCompProp}
\begin{split}
   &T = {}_4T + {}_2T - {}_{24}T, \quad {}_{34}T = {}_{234}T, \quad {}_{12}T = {}_{124}T \\
   \text{and} \quad &T = {}_2T  + {}_4T - {}_{24}T - {}_{34}T   + {}_{234}T - {}_{12}T + {}_{124}T.
\end{split}
\end{gather}
It is easy to see that the latter of these conditions implies the former three, and coincides with Eq.~\eqref{eqn::ProcMat}, which defined the projector onto the set of process matrices. Analogously, using Eq.~\eqref{eqn::CompStrucPres2} and Obs.~\ref{obs::Complete}, we see that the trace-rescaling property yields
\begin{gather}
\label{eqn::trace_rescal_proc_comp}
    T = {}_2T  + {}_4T - {}_{24}T - {}_{34}T   + {}_{234}T - {}_{12}T + {}_{124}T, 
\end{gather}
which coincides with the third condition in Eq.~\eqref{eqn::ProcMatCompProp}, and thus provides no new restrictions. Consequently, as also pointed out in Ref.~\cite{araujo16purification}, process matrices as derived in Ex.~\ref{ex::ProcMat} are also completely admissible with respect to the extensions $\widetilde P_{\inp\att}$ and $\map{P}_{\out\att}^{(\mathrm{TP})}$. 

However, this changes drastically when the extension of $\map{P}_\out$ is taken to be $\map{P}^{(\mathrm{NS})}_{\out\att}$. Now, employing Eq.~\eqref{eqn::CompStrucPres1} and Obs.~\ref{obs::Complete} yields the set of equations
\begin{gather}
\begin{split}
    &T = {}_2T, \quad T = {}_4 T, \quad T = {}_{24}T, \quad T = {}_4T-{}_{34}T + {}_{234}T, \quad T = {}_2T - {}_{12}T + {}_{124}T, \\
    \text{and} \quad &T = {}_2T + {}_4T - {}_{24}T - {}_{34}T + {}_{234}T - {}_{12}T + {}_{124}T. 
\end{split}
\end{gather}
Evidently, all of the above equations follow directly from the third one, $T = {}_{24}T$, implying that $T$ is of the form $T = \rho_{13} \otimes \id_{24}$, where $\rho_{13} \in \L(\H_1\otimes \H_3)$ is a quantum state; that is, $T$ corresponds to the situation where the two parties share a quantum state, and their respective outputs are discarded (see Fig.~\ref{fig::proc_mat_Comp2}). Since the requirements stemming from the trace-rescaling property only depend on $\map{P}_{\inp\att}$, here, they are the same as in Eq.~\eqref{eqn::trace_rescal_proc_comp} above, and thus yield no new restrictions, since Eq.~\eqref{eqn::trace_rescal_proc_comp} is already implied by $T = {}_{24}T$. Thus, process matrices $T$ that are completely admissible with respect to the extensions $\widetilde P_{\inp\att}$ and $\map{P}_{\out\att}^{(\mathrm{NS})}$ are of the form $\rho_{13} \otimes \id_{24}$, which is  a strict subset of all process matrices defined in Ex.~\ref{ex::ProcMat}. \hfill $\blacksquare$
\end{example} 
Lem.~\ref{lem::CompStrucPresTransf} and Obs.~\ref{obs::Complete} provide a systematic way to derive the requirements that follow from complete admissibility for many relevant cases. As we have seen, in some cases, complete admissibility does not impose any additional constraints, while in others it does. This raises the question, under what circumstances the extensions $\map{P}_{\inp\att}$ and $\map{P}_{\out\att}$ yield new restrictions on $T_{\inp\out}$.

\subsection{Sufficient conditions for complete admissibility}
In many cases considered in the previous Section, there is a simple relationship between the projectors $\map{P}_\inp$ and $\map{P}_\out$ and their respective extensions $\map{P}_{\inp\att}$ and $\map{P}_{\out\att}$ that allows for the direct computation of the properties of completely admissible maps $\map{T}_{\inp\out}$.  For example, in Ex.~\ref{ex::ProcMatComp}, $\map{P}_{\inp\att}$ is obtained from $\map{P}_{\inp}$ via the replacement ${}_x\sbt \mapsto {}_{\att_xx} \sbt$ for $x \in \{1,2,3,4\}$. The relationship between the projectors and their respective extensions allows one to decide whether the requirement of complete admissibility yields additional restrictions on the transformation $\map{T}_{\inp\out}$. Concretely, this question can be decided via the following two lemmata, where, for convenience, we consider the trace-rescaling property separately.\footnote{We emphasise that, as mentioned above, we always assume $\gamma_{\out\att}/\gamma_{\inp\att} = \gamma_{\out}/\gamma_{\inp}$ in what follows, since it is otherwise impossible for a transformation to be completely trace-rescaling.} 
\begin{lemma}[Sufficient condition for complete trace-rescaling]
\label{lem::CompTrRescal}
Let $\map{P}_\inp:\L(\H_\inp)\to\L(\H_\inp)$ be a linear, self-adjoint and unital projective map that commutes with the transposition, and let $\{\widetilde P_{\inp\att}\}_\att$ be its extension. If for all auxiliary spaces $\H_\att$ we have 
\begin{gather}
\label{eqn::CompTraceRescal}
    \map{P}_{\inp\att}[X_\inp \otimes \id_\att] = \map{P}_\inp[X_\inp] \otimes \id_\att \quad \forall X_\inp\in\H_\inp,
\end{gather}
then \textit{complete} trace-rescaling does not add any new restriction over `normal' trace-rescaling. 
\end{lemma}
\begin{proof}
Assuming Eq.~\eqref{eqn::CompTraceRescal} to hold, we obtain
\begin{gather}
(\map{P}_{\inp\att} \otimes \map{\id}_\out)[{}_{\out\att}T_{\inp\out}] = (\map{P}_\inp \otimes \map{\id}_{\out\att})[{}_{\inp\out\att} T_{\inp\out}].
\end{gather}
Now, if $T_{\inp\out}$ is trace-rescaling, we obtain from Eq.~\eqref{eqn::CompStrucPres2} in Lem.~\ref{lem::CompStrucPresTransf} (by setting $\H_\att \cong \mathbbm{C}$) that $\widetilde{P}_\inp[{}_\out T_{\inp\out}] = {}_{\inp\out}T_{\inp\out}$ holds. Inserting this into the above equation yields
\begin{gather}
(\map{P}_{\inp\att} \otimes \map{\id}_\out)[{}_{\out\att}T_{\inp\out}] = {}_{\inp\out\att} T_{\inp\out}, 
\end{gather}
which implies that $T_{\inp\out}$ is completely trace-rescaling.
\end{proof}
For example, the assumptions of this Lemma are satisfied for the case of superchannels (see Ex.~\ref{ex::supchann_comp}), where the extended input projector is of the form $\map{P}_{\inp\att}[X_{\inp\att}] = X_{\inp\att} - {}_{3\att_{2}}X_{\inp\att} + {}_{2\att_13\att_2}X_{\inp\att}$, which satisfies Eq.~\eqref{eqn::CompTraceRescal} for all $X_\inp \in \L(\H_2 \otimes \H_3)$. The same holds true for the projector extension $\map{P}_{\inp\att}$ of Ex.~\ref{ex::ProcMatComp}, where we considered complete admissibility for process matrices. We emphasise though, that this doe not have to hold in general. For example, when considering the set of channels that leave a fixed state $\eta_\inp \neq \id_\inp/d_\inp$ invariant (say, for instance, Gibbs-preserving channels~\cite{faist_gibbs-preserving_2015,lostaglio_introductory_2019}), the (non-unital) projector onto the input space is given by $\map{P}_\inp[X_\inp] = \tr[\eta'_\inp X_\inp]\eta'_\inp$, where $\eta_\inp' := \eta_\inp/\tr[\eta_\inp^2]$ (see also Ex.~\ref{Ex::Gibbs}). A possible extension of this projector would be $\map{P}_{\inp\att}[X_{\inp\att}] = \tr[(\eta'_\inp \otimes \xi'_\att)X_{\inp\out}](\eta'_\inp \otimes \xi_{\att}')$, where $\xi'_\att = \xi_\att/\tr[\xi_\att^2]$ and $\xi_\att \neq \id_\att/d_\att$ is some quantum state on $\H_\att$. Evidently, this projector extension does \textit{not} satisfy Eq.~\eqref{eqn::CompTraceRescal} and might thus add new restrictions to the trace re-scaling property.

With respect to the remaining restrictions that come due to the `completeness' requirement, we have the following Lemma:
\begin{lemma}[Sufficient conditions for complete admissibility]
\label{lem::SuffComp}
Let $\map{P}_\inp:\L(\H_\inp)\to\L(\H_\inp)$ and $\map{P}_\out:\L(\H_\out)\to\L(\H_\out)$ be linear, self-adjoint and unital projective maps that commute with the transposition, and let $\{\widetilde P_{\inp\att}\}_\att$ and $\{\widetilde P_{\out\att}\}_\att$ be their respective extensions. If for all auxiliary spaces $\H_\att$ we have  $P_{\inp\att} = \map{P}_\inp \otimes \map{P}_\att$ and $P_{\out\att} = \map{P}_\out \otimes \map{P}_\att$, where $\map{P}_\att$ is a projector that commutes with the transposition, then the requirement of complete admissibility does not add any new restrictions.
\end{lemma}
The proof of the Lemma follows directly from Lem.~\ref{lem::CompStrucPresTransf} and can be found in App.~\ref{app::proofSuffComp}. Here, we emphasise that the extensions $\map{P}_{\inp} \mapsto \map{P}_\inp \otimes \map{P}_\att$ and $\map{P}_{\out} \mapsto \map{P}_\out \otimes \map{P}_\att$ correspond exactly to the `no-signalling composition' considered in Ref.~\cite{hoffreumon_projective_2022}. If, additionally, the projectors $\map{P}_\att$ are unital (which we generally assume), then this extension also satisfies the assumptions of Lem.~\ref{lem::CompTrRescal}, explaining why such an extension will never lead to additional restrictions on $\widetilde T_{\inp\out}$ when complete admissibility is required. 

We emphasise that the above condition is only sufficient, but not necessary for $\widetilde T_{\inp\out}$ to be completely admissible with respect to the extensions $\{\map{P}_{\inp\att}\}_\att$ and $\{\map{P}_{\out\att}\}_\att$. For example, when considering the complete admissibility of process matrices in Ex.~\ref{ex::ProcMatComp}, the extensions $\map{P}_{\inp\att}$ and $\map{P}^{\mathrm{(TP)}}_{\out\att}$ satisfied neither of the conditions of the above Lemma, yet requiring complete admissibility did -- in total -- not add any new restrictions on the set of process matrices.

Together, the results of the present and the previous Section provide a simple framework to incorporate `completeness' into physical considerations, and to decide whether this addition leads to new sets of valid transformations. Importantly, complete admissibility is not a well-defined property per se, but is contingent on the respective projector extensions that are chosen, which, in turn, depend on the physical property one aims to preserve completely. Following the methods presented in Sec.~\ref{sec::GenAppr}, the above results on complete admissibility can readily be extended to more general projectors $\map{P}_\inp$ and $\map{P}_\out$ that are, for example, not self-adjoint, non-unital, or do not commute with the transposition. In particular the non-unital case, alluded to below Lem.~\ref{lem::CompTrRescal} can readily be treated by slightly changing Lem.~\ref{lem::CompTrRescal} alone. Here, since they are less frequently encountered in relevant physical setups, we will not consider complete admissibility for these more general scenarios explicitly. Rather, we now shift our attention to probabilistic quantum transformations.

\section{Probabilistic quantum transformations}
\label{sec::Probabilistic}

In the previous sections, we have -- except for short comments on the consequences of dropping trace rescaling conditions -- only addressed \textit{deterministic} quantum transformations, i.e., transformations that occur with unit probability. Specifically, these are transformations that are `built up' from quantum states (which can be prepared with unit probability); then, CPTP maps (transformations from states to states), superchannels (transformations from quantum channels to quantum channels), process matrices (transformations from channels to number $1$) are all deterministic, since they have a deterministic element as their 'base object'. More abstractly, here, we consider a transformation to be deterministic if it maps between affine quantum sets $\Scal_\inp$ and $\Scal_\out$ with $\gamma_\inp, \gamma_\out \neq 0$.

 However, quantum theory also admits probabilistic quantum transformations. For example, when considering quantum states, probabilistic transformations are described by quantum instruments~\cite{davis_quantum_1976, Lindblad1979}. Concretely, let $\rho\in\L(\H_\inp)$ be a quantum state, then a quantum instrument is a set of CP maps $\{\map{C}^{(i)}\}_i$ -- each of them corresponding to a possible measurement outcome -- with $\map{C}^{(i)}:\L(\H_\inp)\to \L(\H_\out) $ which add up to a quantum channel, that is, $\map{C}:=\sum_i\map{C}^{(i)}$ is CPTP. When the quantum instrument $\{\map{C}^{(i)}\}_i$ is applied on the state $\rho$, with probability $\tr\left[\map{C}^{(i)}[\rho]\right]$, the classical outcome $i$ is obtained and the state $\rho$ is transformed to 
\begin{align}
\rho':=\frac{\map{C}^{(i)}[\rho]}{\tr\left[\map{C}^{(i)}[\rho]\right]}.
\end{align}

In a similar vein, \textit{all} deterministic quantum transformations (in particular, all the quantum transformations we discussed above) have their probabilistic counterpart, given by sets of CP maps that add up to a deterministic quantum transformation.
\begin{definition}[Probabilistic Quantum Transformations]
Let $\map{P}_\inp:\L(\H_\inp)\to\L(\H_\inp)$ and  $\map{P}_\out:\L(\H_\out)\to\L(\H_\out)$ be linear projective maps and $\Scal_\inp \subseteq \L(\H_\inp) $ and $ \Scal_\out \subseteq \L(\H_\out)$ be sets of quantum objects defined by
\begin{align}
\setlength{\arrayrulewidth}{1pt}
\renewcommand{\arraystretch}{1.2}
\begin{tabular}{|c|c|c|}
\hhline{|-|~|-|}
$W\in \L(\H_\inp) $ belongs to $\Scal_\inp$ iff  & \hspace*{20mm}& $W'\in \L(\H_\out) $ belongs to $\Scal_\out$ iff  \\
$W\geq0$ &\Large{$\mathbf{\longrightarrow}$} & $W'\geq0$ \\
$\map{P}_\inp[W]=W$ && $\map{P}_\out[W']=W'$ \\
$\tr[W]=\gamma_\inp$ && $\tr[W']=\gamma_\out$  \\
\hhline{|-|~|-|}
\end{tabular}
\end{align} 
The set $\{\map{T}^{(i)}_{\inp\out}\}_i$, $\map{T}^{(i)}_{\inp\out}:\L(\H_\inp)\to\L(\H_\out)$ represents a probabilistic quantum transformation from $\Scal_\inp$ to $\Scal_\out$ when:
\begin{subequations}
\begin{empheq}[box=\fcolorbox{blu}{white}]{align}
	i:& \quad \map{T}^{(i)}_{\inp\out} \text{ is completely positive for every }i\\
	ii:& \quad \forall W\in\Scal_\inp, \text{ we have that } \sum_i \map{T}^{(i)}_{\inp\out}[W]\in\Scal_\out
\end{empheq}
\end{subequations}

	When a probabilistic quantum transformation $\{\map{T}^{(i)}_{\inp\out}\}_i$, $\map{T}^{(i)}_{\inp\out}:\L(\H_\inp)\to\L(\H_\out)$ is performed on a quantum object $W\in\Scal_\inp$, with probability $p(i)=\tr\left[\map{T}^{(i)}_{\inp\out}[W]\right]$, the classical outcome $i$ is obtained and $W$ is transformed to 
\begin{align}
W':=\frac{\map{T}^{(i)}_{\inp\out}[W]}{\tr\left[\map{T}^{(i)}_{\inp\out}[W]\right]}.
\end{align}
\end{definition}

Expressed in terms of Choi matrices, any $T_{\inp\out}^{(i)}$ for which there exists a deterministic transformation $T_{\inp\out}$  such that $T_{\inp\out}^{(i)} \leq T_{\inp\out}$ is a valid probabilistic transformation. In general, this latter requirement does not impose many restrictions on the structure of probabilistic elements. For example, if there exists a deterministic transformation $T_{\inp\out} \propto \id_{\inp\out}$ (which is the case for all examples we have considered so far), then for \textit{any} $T_{\inp\out}^{(i)} \geq 0$ there exists $\lambda \geq 0$, such that $\lambda T_{\inp\out}^{(i)} \leq T_{\inp \out}$. For more general cases, except for an irrelevant scaling factor $\lambda$, the only requirement for probabilistic transformations is that they entirely lie in the support of the deterministic transformations, i.e., $\map{P}_\texttt{S} [T_{\inp\out}^{(i)}] =  T_{\inp\out}^{(i)}$, where $\map{P}_\texttt{S}$ is the projector onto the space $\texttt{S} := \text{span}(\bigcup \{\text{supp}(T_{\inp\out})| T_{\inp\out} \ \text{deterministic transformation}\})$.

\section{Measuring quantum objects: dual affine sets, POVMs, and testers}
\label{sec::DualAff}

A particularly important set of probabilistic quantum transformations are measurements, i.e., probabilistic transformations with output space $\mathbbm{C}$. Considering measurements naturally leads to the concept of dual affine sets, which play a pivotal role in quantum mechanics (and beyond), and can also be characterised using the techniques we introduced in the previous sections.

Physically speaking, a quantum measurement is a process which allows one to extract classical outcomes from quantum objects. For instance, if $\rho\in\L(\H)$ is a quantum state, measurements on $\rho$ are described by means of a Positive Operator Valued Measure (POVM), which is a set $\left\{M^{(i)}\right\}_i$ of positive semidefinite operators $M^{(i)}\in\L(\H)$ that add up to the identity, \ie, $\sum_iM^{(i)}=\id$. When the POVM $\left\{M^{(i)}\right\}_i$ is performed on a quantum state $\rho$, the outcome $i$ is obtained with probability $\tr\left[ M^{(i)}\rho \right]$.
In similar spirit, one can also perform measurements on different quantum objects, for instance, one can perform measurements on quantum channels by means of the \textit{tester formalism} \cite{chiribella07circuit_architecture,chiribella09networks,bavaresco21}, also referred to as process POVMs~\cite{ziman08_process_POVM}.

Before going to more general scenarios, we present a brief discussion on how quantum testers can be used to measure a quantum channel. Let $C\in\L(\H_\inp\otimes\H_\out)$ be the Choi operator of a quantum channel and $\{T^{(i)}\}_i$, $T^{(i)}\in\L(\H_\inp\otimes\H_\out)$ be a set of operations such that the probability of measuring $i$ on the channel $C$ is given by $p(i)=\tr\Big[T^{(i)} C\Big]$. In order for $\{p(i)\}_i$ to be a positive number for every positive operator\footnote{Here, the operator $C$ is assumed to be an \textit{arbitrary} positive semidefinite operator, which may not satisfy the constraints of a quantum channel. This is because we require the quantity $\tr[T^{(i)} C]$ to be non-negative not only on channels, but also instrument elements or when acting non-trivially only on a part of a quantum channel (this is similar to a complete positivity argument for the case of quantum channels).} $C$, we need $T^{(i)}\geq0$ for all $i$, and in order to ensure normalisation for every quantum channel $C$, we need that $\sum_i \tr[T^{(i)} C]=1$, which is equivalent to imposing $\tr[T C]=1$ for every channel $C$, where $T:=\sum_i T^{(i)}$. 
The set of all operators $T$ respecting $\tr[T C]=1$ for every channel $C$ is the \textit{dual affine set} of the set of quantum channels, and a set of operators $\{T^{(i)}\}_i$ respecting 
\begin{align}
    T^{(i)}\geq&0 \\
    \sum_i \tr[T^{(i)} C]=&1, \quad \text{ for every channel $C$} 
\end{align}
is called a tester. Interestingly, all quantum testers may be realised within standard quantum circuits, that is, for any tester $\{T^{(i)}\}_i$, $T^{(i)}\in\L(\H_\inp\otimes\H_\out)$ there always exist a state $\rho\in\L(\H_\inp\otimes\H_\text{aux})$ and a POVM $\{M^{(i)}\}_i$, $M^{(i)}\in\L(\H_\text{aux}\otimes\H_\out)$ such that $\tr[T^{(i)} C]=\tr\left[M^{(i)}\; \left(\map{C}\otimes\map{\id}_\text{aux}[\rho]\right)\right]$. Although we might not always have a quantum circuit realisation for other quantum objects (such as process matrices), the concept of dual affine imposes the minimal normalisation constraint required by measuring general quantum objects and plays a fundamental role in general quantum measurements~\cite{ebler16,bavaresco21,hoffreumon_projective_2022} and general quantum assemblages~\cite{bavaresco19_SDI}.

\begin{definition}[Dual Affine set]
Let $\mathcal{S} \subseteq \L(\H)$. An operator $\overline{W}\in\L(\H)$ belongs to $\overline{\mathcal{S}}$, the \textbf{dual affine} set of $\mathcal{S}$ if\,\footnote{In this work we are mostly interested in self-adjoint operators, hence, when $W=W^\dagger$,  we have $\tr[\overline{W}^\dagger \,W]=\tr[\overline{W}\,W]$.}
\begin{equation}
	\tr[\overline{W}^\dagger \,W]=1, \quad \forall W\in\mathcal{S}.
\end{equation}
\end{definition}
Naturally, for \textit{any} set $\mathcal{S}$, its dual affine set is indeed affine, since $\sum_i \lambda_i \tr[
{\overline{W}^{(i)}}^\dagger W] = 1$  for all $\overline{W}^{(i)} \in \overline{\mathcal{S}}$ and  $W \in \mathcal{S}$ if $\sum_i\lambda_i = 1$. If the set $\mathcal{S}$ itself is affine, then we can derive the properties of elements in $\overline{\mathcal{S}}$ in a straightforward way.

We now present a Theorem -- also obtained in an independent way in Ref.~\cite{hoffreumon_projective_2022} -- that allows us to obtain a simple characterisation for dual affine sets of quantum objects. 
\begin{theorem}
\label{thm::DualAff}
Let $\map{P}:\L(\H)\to\L(\H)$ be a  linear projective map and $\mathcal{S}\subseteq \L(\H)$ be an affine set defined by
\begin{subequations}
\begin{empheq}[box=\fcolorbox{brow}{white}]{align}
		W \in \L(\H) \ \text{be} &\text{longs to}  \ \Scal \ \text{iff} \\
    W=\,&\map{P}[W] \\
	\tr[W]=\,&\gamma.
\end{empheq}
\end{subequations}
where $\widetilde P$ is self-adjoint, unital, and commutes with the transposition \ie,
\begin{subequations}
\begin{empheq}[box=\widefbox]{align}
		\map{P}=\,&\map{P}^\dagger \\
		\map{P}[\id]=\,&\id \\
		\map{P}[W^\tau]=\,&{\map{P}[W]}^\tau, \quad \forall W\in\L(\H)\, ,
\end{empheq}
\end{subequations}
and $\gamma \neq 0$. An operator $\overline{W}\in\L(\H)$ belongs to the dual affine set $\overline{\mathcal{S}}$ if  and only if
\begin{subequations}
\begin{empheq}[box=\fcolorbox{blu}{white}]{align}
\label{eqn::Dual1}
		\overline{W}=\, &\overline{W}-\map{P}[\overline{W}]+ \tr\left(\overline{W}\right) \frac{\id}{d} \\ \label{eqn::Dual2}
			\tr[\overline{W}]=\, &\frac{d}{\gamma} \, ,
\end{empheq}
\end{subequations}
where $d= \text{dim}(\H)$.
\end{theorem} 
\begin{proof}
	This Theorem can be shown in two separate ways. On the one hand, since $\map{P}$ satisfies the requirements of Thm.~\ref{thm::genFormProj}, we can directly use it to prove the above Theorem. Secondly, we can show it directly. Since this proof has merit in its own right, we start with this latter approach. To this end, we first note that, as we discuss in detail in Sec.~\eqref{sec::LinOp}, if a linear operator $\map{P}$ is self-adjoint and commutes with the transposition, then $\tr[A\map{P}[B]] = \tr[\map{P}[A]B]$ for all $A, B$. Thus, for any $\overline W$ that satisfies Eqs.~\eqref{eqn::Dual1} and~\eqref{eqn::Dual2}, we have  
\begin{align}
	\tr[\overline W W] = \tr\left[\left(\overline{W} - \map{P} [\overline{W}] + \frac{\mathbbm{1}}{\gamma}\right)W\right] =& \tr[(\overline{W}W) - \tr[\overline{W} \map{P}[W]] + \frac{1}{\gamma}\tr W = 1,
\end{align}
where we have used $\map{P}[W]=W$ and $\tr W = \gamma$ for all $W \in \mathcal{S}$. 

To prove the converse, first note that, since $\map{P}$ is a self-adjoint and unital projector, it is also trace-preserving, and we have $\map{P}[M] \in \mathcal{S}$ for all $M\in \L(\H)$ which satisfy $\tr[M] = \gamma$. The set of all such matrices $M$ spans $\L(\H)$. Now, for any $\overline W$ that satisfies $\tr[\overline W W] = 1$ for all $W \in \mathcal{S}$, we have 
\begin{gather}
\tr[\overline W \map{P}[M]] = \tr[\map{P}[\overline W] M] = 1 = \frac{1}{\gamma} \tr[M]\, ,
\end{gather}
where we have used $\tr[M] = \gamma$. Now, since the above equation holds for a full basis of $\L(\H)$, we have 
\begin{gather}
\map{P}[\overline W] = \frac{1}{\gamma} \mathbbm{1} \quad \Rightarrow \quad \tr[\overline W] = \frac{d}{\gamma}.
\end{gather}
Together, the two statements in the above equation yield Eqs.~\eqref{eqn::Dual1} and~\eqref{eqn::Dual2}, completing the proof. 

As mentioned, and as already implicitly done in Ex.~\ref{ex::ProcMat}, we can also prove this statement by directly employing Thm.~\ref{thm::genFormProj}. To do so, we note that in the considered case, the output space $\H_\out$ and output projector $\map{P}_\out$ are trivial, while we have the rescaling factor $\gamma_\out/\gamma_\inp = 1/\gamma$, such that Eqs.~\eqref{eqn::GenProj1} and~\eqref{eqn::GenProj2} of Thm.~\ref{thm::genFormProj} are directly equivalent to Eqs.~\eqref{eqn::Dual1} and~\eqref{eqn::Dual2} of the above Theorem.
\end{proof}
We emphasise that adding a positivity constraint to the objects in the set $\mathcal{S}$, as is often naturally the case in quantum mechanics, would yield the same linear constraints on the dual set $\overline{\mathcal{S}}$. As already outlined, the characterisation of the dual affine set is simply a special case of the overall characterisation of trace-rescaling linear maps between spaces that are defined by projectors $\map{P}_\inp$ and $\map{P}_\out$. In particular, denoting the corresponding map by $\map{T}[W] = 1$, we have $\overline W^\tau = T$, where $T$ is the Choi matrix of $\map{T}$ and the additional transposition $\sbt^\tau$ appears due to the convention we chose for the Choi formalism. Since dual sets play a prominent role in quantum mechanics, here we chose to discuss this important case explicitly. 

While the above Theorem only applies for self-adjoint, unital projectors that commute with the transposition, it can easily be phrased for more general situations (see Sec.~\ref{sec::GenAppr}).

\subsection{Quantum measurement and its relationship with probabilistic transformations}
From the above discussion, we can now consider quantum measurements on quantum objects in a more general way. We start by presenting their definition.
\begin{definition}
    Let $\Scal_\inp\subseteq\L(\H)$ be a set of quantum objects and  $\overline{\Scal}_\inp$ its dual affine. A general quantum measurement on $\Scal_\inp$ is given by a set of operators $\{M^{(i)}\}_i$, with $M^{(i)}\in\L(\H)$ respecting,
\begin{align}
    M^{(i)}\geq0 \\ 
    \sum_i M^{(i)} \in \overline{\Scal}_\inp,
\end{align}
and the probability of obtaining an outcome $i$ when measuring the object $W\in\Scal_\inp$ is $p(i)=\tr[M^{(i)} W]$.
\end{definition}
General measurements are the largest set of measurements which is in principle allowed by quantum theory, and may be used to perform measurements on general quantum objects such as process matrices, as in Ref.~\cite{lewandowska23discrimination} where the authors used general measurements to discriminate between process matrices with indefinite causal order. 
Similarly to other general transformations discussed in this manuscript, it may be the case that a general measurement may not be realised by quantum circuits (due to indefinite causality), or it might even be the case that one can never obtain a 'fair' physical implementation for some general measurements (due to some other physical principle, e.g., a reversibility preserving principle~\cite{araujo16purification,yokojima20} or the requirement of logically consistent processes~\cite{baumeler16_logically,vanrietvelde22_consistent_circuits}). However, any set greater than the one defined above is certainly forbidden by quantum theory.

We remark that the set of general quantum measurements is closely related to the set of probabilistic transformations. Similarly to quantum instruments, a probabilistic transformation may be viewed as a description of a quantum measurement and a post-measurement state. Hence, every probabilistic transformation corresponds to a quantum measurement.
More precisely, if $\{C^{(i)}\}_i$ is a probabilistic quantum transformation from $\Scal_\inp$ to $\Scal_\out$, its associated general measurement operators are given by $M^{(i)}:=\map{C}^{(i)\dagger}[\id_\out]$. Indeed,
\begin{align}
    \tr[M^{(i)} W]&= \tr[\map{C}^{(i)\dagger}[\id_\out] W] \\ 
    &=\tr[ \id_\inp \map{C}^{(i)}[ W]] \\
    &=\tr[ \map{C}^{(i)}[W]] \\
    &=p(i),
\end{align}
    which is precisely the probability of obtaining the outcome $i$.
Also, every quantum measurement may be viewed as a probabilistic transformation from a quantum object set $\Scal_\inp$ to the trivial set $\Scal_\out=\{1\}\subseteq\L(\mathbb{C}) \cong \mathbb{C}$, which contains only the unit scalar. More precisely, if $\{M^{(i)}\}_i$ is a general quantum measurement on $\Scal_\inp$, one can always define the probabilistic transformation $\map{C}^{(i)}[W]:=\tr[M^{(i)} W]$, where $\map{C}^{(i)}:\L(\H)\to\L(\mathbb{C})$. It is immediate to verify that the map $\map{C}^{(i)}[W]=\tr[M^{(i)} W]$ is completely positive and respects $\sum_i \map{C}^{(i)}[W] = 1$ for all $W\in \Scal_\inp$. 
\subsection{Non-signalling channels and multipartite process matrices}
\label{sec::NonSigChann}
We now use the concept of dual affines to return a final time to the case of process matrices as the dual affine set of the set of non-signalling channels. Here, for generality, we consider the $k$-party case (also considered in Ref.~\cite{araujo15witnessing}). To this end, let us first define non-signalling channels.
\begin{figure}[ht!]
    \centering
    \includegraphics[width = 0.4\linewidth]{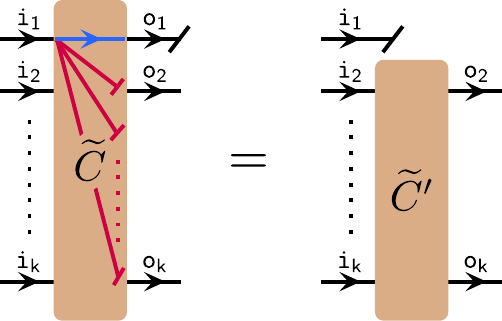}
    \caption{\textbf{Multipartite non-signalling channels.} Each input party $\inp_\ell$ can at most send signals to its corresponding output party $\out_\ell$. Here, this is depicted for the party $\inp_1$, where a blue arrow denotes the possibility to send a signal, while the red lines signify that no signal can be sent. Since the party $\inp_1$ can only send signals to $\out_1$, discarding said outcome then amounts to directly discarding the input $\inp_1$ (depicted in the figure). In terms of the Choi state $C$ of $\map{C}$, this corresponds to the requirement ${}_{\out_1}C = {}_{\inp_1\out_1}C$ of Eq.~\eqref{eqn::non_sig} and analogously for all other pairs $\{\inp_\ell, \out_\ell\}$.}
    \label{fig::non_sig}
\end{figure}
Informally, non-signalling channels are multipartite quantum channels which cannot be used for exchanging information between distinct parties. Let $k\in\mathbb{N}$ be an integer. The Hilbert spaces corresponding to the total input and output, respectively, of such a non-signalling map are given by
\begin{align}
\H_\inp:=&\H_{\inp_1}\otimes\H_{\inp_2}\otimes\ldots\otimes\H_{\inp_k}\\
\H_\out:=&\H_{\out_1}\otimes\H_{\out_2}\otimes\ldots\otimes\H_{\out_k}.
\end{align}
Then, a multipartite quantum channel  $\map{C}:\L(\H_\inp)\to\L(\H_\out)$ is non-signalling if its Choi state $C$ respects,
\begin{align}
\label{eqn::non_sig}
    _{\out_\ell}(C)=_{\inp_\ell\out_\ell}C, \quad \forall \ell\in\{1,2,\ldots,k\}.
\end{align}
Intuitively, the above property implies that discarding the output of party $\ell$ amounts to directly discarding its input, which implies that the only signalling of party $\ell$ happens from $\inp_\ell$ to $\out_\ell$, but not to any other output $\out_{\ell'}$ (see Fig.~\ref{fig::non_sig} for a graphical depiction). 

The requirements of Eq.~\eqref{eqn::non_sig} are equivalent to stating that the map $\map{C}:\L(\H_\inp)\to\L(\H_\out)$ can be written as an affine combination of independent channels, that is $\map{C}=\sum_\alpha \gamma^{(\alpha)} \map{C}^{{(\alpha)}_1}\otimes \map{C}^{(\alpha)}_2\otimes \ldots\otimes \map{C}^{(\alpha)}_k $, where $\gamma^{(\alpha)}\in\mathbb{R}$, $\sum_\alpha \gamma^{(\alpha)}=1$, and all maps $\map{C}^{(\alpha)}_\ell:\L(\H_{\inp_\ell})\to\L(\H_{\out_\ell}) $ are quantum channels~\cite{Gutoski09,chiribella09_switch}.

From this, we obtain a simple characterisation of non-signalling quantum channels. To this end, we define the projectors:
\begin{alignat}{2} \label{eq::NS1}
\map{P}^{(\ell)}&:\L(\H_{\inp_\ell}\otimes\H_{\out_\ell})\to\L(\H_{\inp_\ell}\otimes\H_{\out_\ell}), && \ell\in\{1,2,\ldots,k\}, \\\label{eq::NS2}
\map{P}^{(\ell)}[C]&:=C-_{\out_\ell}C+_{\inp_\ell\out_\ell}C, &&  \ell\in\{1,2,\ldots,k\}, \\\label{eq::NS3}
\map{P}_{NS}&:\L(\H_\inp\otimes\H_\out)\to\L(\H_\inp\otimes\H_\out), &\qquad & \\       \label{eq::NS4}
       \map{P}_{NS}&:=\map{P}^{(k)} \circ \ldots\circ\map{P}^{(2)}\circ \map{P}^{(1)}. &&
\end{alignat}
We emphasise that, here, the order in which the projectors $\map{P}^{(\ell)}$ are applied in Eq.~\eqref{eq::NS4} does not matter, since they all commute (making a construction of $\map{P}_{NS}$ via concatenation possible in the first place).
Hence, a linear operator $C\in\L(\H_\inp\otimes\H_\out)$ is a non-signalling quantum channel if and only if
\begin{align}
    &C\geq0 \\
    &\map{P}_{NS}[C]=C\\
    &\tr[C]=d_\out.
\end{align}
Since the multipartite process matrices lie in the dual affine set of non-signalling channels, Thm.~\ref{thm::DualAff} provides a simple characterisation of multipartite process matrices for an arbitrary number of parties.

\begin{example}[Multipartite process matrices]
Using the projectors for multipartite non-signalling channels defined in Eqs.~\eqref{eq::NS2} and \eqref{eq::NS4} and Thm.~\ref{thm::DualAff}, we obtain a simple characterisation for multipartite process matrices for any numbers of parties $k$.
\begin{align}
\hspace*{-12mm}
\setlength{\arrayrulewidth}{1pt}
\renewcommand{\arraystretch}{1.2}
\hspace{0.6cm}\begin{tabular}{|c|c|c|}
\hhline{|-|~|-|}
\bf{(Multipartite non-signalling channel)} & \hspace*{6mm}& \bf{(Complex number)}\\
$C_{}\in \L(\H_\inp\otimes\H_\out) $ belongs to $\Scal_\inp$ iff  & \hspace*{6mm}& $c\in \L(\mathbb{C}) $  belongs to $\Scal_\out$ iff  \\
$C\geq0$, && $c\geq0$, \\
$C=\map{P}_{NS}[C]$, &\hspace*{2mm} \Large{$\mathbf{\longrightarrow}$}  \hspace*{2mm}& $\map{P}_\out[c]:=c$, \\
$\tr[C_{}]=d_\inp$. && $\tr[c]=1$.\\
\hhline{|-|~|-|}
\end{tabular}
\end{align} 
\begin{subequations}\label{eqn::ProcMatMulti}
\begin{empheq}[box=\fcolorbox{blu}{white}]{align}
 \text{\bf{(Multipartite process}} & \text{  \bf{matrix)}}\nonumber \\
W\geq&0, \\
    W=W-\map{P}_{NS}[W]&+\tr[W]\frac{\id_{\inp\out}}{d_\inp d_\out },\\
 \tr [W] =& d_\out\,.
\end{empheq}
\end{subequations}
We emphasise that this characterisation of multi-partite process matrices has also been provided in equivalent form in App.~B3 of Ref.~\cite{araujo15witnessing}. Here, it follows straightforwardly from the (readily derived) properties of non-signalling channels, and the fact that process matrices form their dual affine set.\hfill $\blacksquare$
\end{example}

\section{Link product and key concepts}
\label{sec::Link}
The proofs of Thm.~\ref{thm::genFormProj} as well as its generalisations rely on only a handful of simple mathematical concepts, which we now discuss. Predominantly, we will rely on three main ingredients: The Choi-Jamio{\l}kowski isomorphism (CJI), which allows us to phrase all statements on maps in terms of matrices; the link product, which translates the concatenation of maps to the corresponding manipulation on the level of Choi matrices; and the fact that linear operators with particular properties can be moved around freely in the link product. 

\subsection{Link Product}
We start by discussing the link product $\star$, already informally introduced in Eq.~\eqref{eqn::LinkIntro}, which captures the action of maps in terms of the CJI~\cite{chiribella09networks}. Concretely, for any linear maps $\map{T}_{\xt\yt}: \L(\H_\xt) \rightarrow \L(\H_\yt)$, $\map{V}_{\yt\zt}: \L(\H_\yt) \rightarrow \L(\H_\zt)$ and arbitrary matrices $M_\xt\in \L(\H_\xt)$, we have 
\begin{gather}
    \text{Choi}[\map{T}_{\yt\zt} \circ \map{V}_{\xt\yt}] =: V_{\yt\zt}  \star T_{\xt\yt} \in \L(\H_\xt\otimes \H_\zt) \quad \text{and} \quad \map{T}_{\xt\yt}[M_\xt] = T_{\xt\yt} \star M_\xt \in \L(\H_\yt)\, ,
\end{gather}
where $\text{Choi}[\sbt\,]$ transforms a map to its corresponding Choi matrix. In particular, the link product of two arbitrary matrices $T_{\xt\yt}: \L(\H_\xt \otimes \H_\yt)$ and $V_{\yt\zt}: \L(\H_\yt \otimes \H_\zt)$ is given by a trace over the spaces they are both defined on and a partial transpose over the same space\footnote{The concrete form of the link product -- in particular the presence of partial transposes -- depends on the convention of the CJI one employs. The form of the link product we present here is in line with the convention chosen in Eq.~\eqref{eqn::Choi}.}, i.e., 
\begin{gather}
\label{eqn::Link}
    T_{\xt\yt} \star V_{\yt\zt} := \tr_\zt[(T_{\xt\yt} \otimes \mathbbm{1}_\zt)(\mathbbm{1}_\xt \otimes V_{\yt\zt}^{\tau_\yt})]\, ,
\end{gather}
where $\sbt^{\tau_\yt}$ denotes the partial transpose with respect to the computational basis of $\H_\yt$. As has been shown in Refs.~\cite{chiribella09networks}, the link product of positive semidefinite (Hermitian) matrices is again positive semidefinite (Hermitian), and it is both associative and -- for all cases we consider -- commutative (up to a re-ordering of tensor factors, which we always tacitly assume). Additionally, it is easy to see that the link product satisfies 
\begin{gather}
\label{eqn::Completeness}
    A \star B = A' \star B \quad \forall B \quad \Leftrightarrow \quad A = A' ,
\end{gather}
since $A\star B$ and $A'\star B$ are equal to $\map{A}[B]$ and $\map{A}'[B]$, respectively, and if two linear maps agree on all elements they act on, they coincide, i.e., $\map{A} = \map{A}'$ and thus $A = A'$ (the converse direction in Eq.~\eqref{eqn::Completeness} holds trivially). 

Importantly, the link product allows us to re-phrase the question of finding the properties of (trace-rescaling) mappings $\map{T}_{\inp\out}$ between sets $\Scal_\inp$ and $\Scal_\out$ defined by projectors $\map{P}_\inp$ and $\map{P}_\out$, respectively, in terms of Choi matrices. The requirements that $\map{T}_{\inp\out}[W_\inp] \in \Scal_\out$ and $\tr[\map{T}_{\inp\out}[W_\inp]] = \gamma_\out/\gamma_\inp \tr[W_\inp]$ for all $W_\inp \in \Scal_\inp$ can now be phrased as
\begin{gather}
 (\map{\id}_\inp \otimes \map{P}_\out)[T_{\inp\out} \star W_\inp] = T_{\inp\out} \star W_\inp \quad \text{and} \quad \tr[T_{\inp\out} \star W_\inp] = \frac{\gamma_\out}{\gamma_\inp} \tr[W_\inp]
\end{gather}
for all $W_\inp = \map{P}[W_\inp]$ and $\tr[W_\inp] = \gamma_\inp$. In order to deduce the structural properties these two equations engender for $T_{\inp\out}$, all constraints need to be `moved onto' $T_{\inp\out}$. Consequently, we now discuss how to `move around' linear maps in the link product.

\subsection{Linear operators in the link product}
\label{sec::LinOp}
The final property of the link product that we will make frequent use of is the fact that linear maps that act on one of the factors in the link product can be `moved around'{(this is akin to finding their adjoint action)}. In order to obtain simplifications for the special case of self-adjoint, unital maps that commute with the transposition -- the case most frequently encountered in quantum mechanics -- we first recall some (well-known) pertinent properties of such maps:

\begin{lemma}[Properties of linear maps]
\label{lem::PropMap}
Let $\map{P}: \L(\H) \rightarrow \L(\H)$ be a linear map. The following statements hold:
\begin{enumerate}
    \item If $\map{P}$ is self-adjoint, then it is Hermiticity preserving. 
    \item If $\map{P}$ is self-adjoint and unital, then it is trace-preserving.
    \item If $\map{P}$ is self-adjoint, then it commutes with the transposition iff it commutes with complex conjugation (with respect to the same basis).
    \item If $\map{P}$ is self-adjoint and commutes with the transposition (or complex conjugation), then $\tr[M' \widetilde P[M]] = \tr[\map{P}[M'] M]$ for all $M', M \in \L(\H)$.
    \item If $\map{P}$ is Hermiticity preserving, then $\tr[H' \map{P}[H]] = \tr[\map{P}[H'] H]$ for all Hermitian $H', H \in \L(\H)$.
\end{enumerate} 
\end{lemma}

All the proofs follow by direct insertion and are provided in App.~\ref{app::LinMapsProp} for completeness. With these properties of linear maps $\widetilde P$ in hand, we can now investigate how linear maps can be `moved around' in the link product. 
\begin{lemma}
\label{lem::LinkGen}
Let $A_{\textup{\inp}\textup{\out}} \in \L(\H_\textup{\inp} \otimes \H_\textup{\out})$ and $B_\textup{\inp} \in \L(\H_\textup{\inp})$, and let $\widetilde P_\textup{\inp}: \L(\H_\textup{\inp}) \rightarrow \L(\H_\textup{\inp})$ be a linear operator. Then 
\begin{gather}
\label{eqn::LinkGen}
 A_{\textup{\inp}\textup{\out}} \star \widetilde P_\textup{\inp}[B_\textup{\inp}] = \widetilde P^\dagger_\textup{\inp}[A^{*}_{\textup{\inp\out}}]^* \star B_\inp =: \widetilde P_\textup{\inp}^\tau [A_{\textup{\inp\out}}] \star B_\textup{\inp}.
\end{gather}
If $\widetilde P_\textup{\inp}$ is self-adjoint and commutes with the transposition (or complex conjugation), then 
\begin{gather}
\label{eqn::LinkSpec}
    A_{\textup{\inp}\textup{\out}} \star \widetilde P_\textup{\inp}[B_\textup{\inp}] = \widetilde P_\inp[A_{\textup{\inp\out}}] \star B_\textup{\inp}
\end{gather}
holds.
\end{lemma}
\begin{proof}
First, we note that, if Eq.~\eqref{eqn::LinkGen} holds, then, due to the properties provided in Lem.~\ref{lem::PropMap},  Eq.~\eqref{eqn::LinkSpec} follows directly when $\widetilde P$ is self-adjoint and commutes with the transposition (or complex conjugation). For the proof of Eq.~\eqref{eqn::LinkGen}, we first recall that the action of any linear operator $\widetilde P_\inp: \L(\H_\inp) \rightarrow \L(\H_\inp)$ can be written as $\map{P}_\inp[\sbt \,] = \sum_\alpha L^{(\alpha)}_\inp \sbt \,R^{(\alpha)\dagger}_\inp$ for some matrices $L^{(\alpha)}_\inp, R^{(\alpha)\dagger}_\inp \in \H_\inp$. With this, from the definition of the adjoint, it is easy to see that $\widetilde P_\inp^\dagger[\sbt \,] = \sum_\alpha L_\inp^{(\alpha)\dagger} \sbt \, R_\inp^{(\alpha)}$ holds. Now, using the definition of the link product, we obtain 
\begin{align}
    A_{\textup{\inp}\textup{\out}} \star \map{P}_\textup{\inp}[B_\textup{\inp}] &= \tr_{\textup{\inp}}[A_{\textup{\inp\out}} (\map{P}_\textup{\inp}[B_\textup{\inp}]^\tau \otimes \mathbbm{1}_\textup{\out})] = \sum_\alpha \tr_{\textup{\inp}}[A_{\textup{\inp\out}} (R_\inp^{(\alpha)*} B_\textup{\inp}^\tau L_\inp^{(\alpha)\tau} \otimes \mathbbm{1}_\textup{\out})] \\
    \label{eqn::FinGenLink}
    &= \sum_\alpha \tr_{\textup{\inp}} [(L_\inp^{(\alpha)\dagger} A_{\textup{\inp\out}}^{*} R_\inp^{(\alpha)})^{*} (B_\textup{\inp}^\tau \otimes \mathbbm{1}_\textup{\out})] 
    = \widetilde P^\dagger_\inp[A^{*}_{\inp\out}]^* \star B_\inp\, ,
\end{align}
\end{proof}
We note that it is easy to see that, if $\map{P}_\inp[\,\sbt\,] = \sum_\alpha L_\inp^{(\alpha)} \sbt \, R_\inp^{(\alpha)\dagger}$, then $\map{P}_\inp^\tau[\,\sbt \,] := P^\dagger_\inp[\,\sbt^{*}]^* =  L_\inp^{(\alpha)\tau} \sbt \, R_\inp^{(\alpha)*}$. Importantly for our purposes, Eqs.~\eqref{eqn::LinkGen} and~\eqref{eqn::LinkSpec} allow us to move linear operators around freely in the link product, which we will now exploit to deduce the properties of $T_{\inp\out}$.

\subsection{Proving statements using the link product}
\label{sec::ProofLink}
Now, using the link product, we can easily provide the proofs for the statements made in Secs.~\ref{sec::GenTrans} and~\ref{sec::DualAff}. 
As mentioned, for any linear map $\map{T}_{\inp\out}$ that maps $\Scal_\inp$ onto $\Scal_\out$, we have
\begin{gather}
\label{eqn::ThmLink}
 (\map{\id}_\inp \otimes \map{P}_\out)[T_{\inp\out} \star W_\inp] = T_{\inp\out} \star W_\inp \quad \text{and} \quad \tr[T_{\inp\out} \star W_\inp] = \frac{\gamma_\out}{\gamma_\inp} \tr[W_\inp] \quad \forall W_\inp \in \mathcal{S}_\inp
\end{gather}
Let us start by providing the structural properties of $T_{\inp\out}$ for the case of self-adjoint, unital projectors $\map{P}_\inp$ and $\map{P}_\out$ that commute with the transposition. To do so, we first make use of $(\map{\id}_\inp \otimes \map{P}_\out)[T_{\inp \out} \star W_\inp] = (\map{\id}_\inp \otimes \map{P}_\out)[T_{\inp \out}] \star W_\inp = T_{\inp\out} \star W_\inp$ for all $W_\inp \in \mathcal{S}_\inp$. Importantly, since $\text{span}(\mathcal{S}_\inp)$ does generally \textit{not} coincide with the full space $\L(\H_\inp)$, this equation does \textit{not} allow us to deduce that $(\map{\id}_\inp \otimes \widetilde P_\out) [T_{\inp\out}] = T_{\inp\out}$. However, it is easy to see (since $\widetilde P_\inp$ is a projection) that for every $M \in \L(\H_\inp)$ with $\tr[W] \neq 0$ (for the case $\gamma_\inp = 0$ see App.~\ref{app::gam0Choi}) we have $\widetilde P_\inp[M] \in \text{span}(\mathcal{S}_\inp)$. Consequently, we obtain 
\begin{gather}
    (\map{\id}_\inp \otimes \widetilde P_\out) [T_{\inp\out}] \star \widetilde{P}_\inp[M] = T_{\inp\out} \star \widetilde{P}_\inp[M]  \quad \forall M \in \L(\H_\inp)
\end{gather}
Now, we can use the second part of Lem.~\ref{lem::LinkGen} to move the projector $\widetilde P_\inp$ inside the link product, such that 
\begin{gather}
(\widetilde P_\inp \otimes \widetilde P_\out) [T_{\inp\out}] \star M = \widetilde{P}_\inp[T_{\inp\out}] \star M  \quad \forall M \in \L(\H_\inp)\,, 
\end{gather}
which, since it holds for all $M\in \L(\H_\inp)$, implies $(\widetilde P_\inp \otimes \widetilde P_\out) [T_{\inp\out}] = \widetilde{P}_\inp[T_{\inp\out}]$.
This, in turn, can be phrased in terms of a projector on $T_{\inp\out}$ as 
\begin{gather}
\label{eqn::Proj1}
   T_{\inp\out} = T_{\inp\out} - \widetilde P_\inp [T_{\inp\out}] + (\widetilde P_\inp \otimes \widetilde P_\out) [T_{\inp\out}]\, ,
\end{gather}
where the signs in the above definition are chosen such that $\widetilde{P}_\inp[T_{\inp\out}] = (\widetilde P_\inp \otimes \widetilde P_\out) [T_{\inp\out}]$ still holds (which can be seen by direct insertion into~\eqref{eqn::Proj1} and using that $\widetilde P_\inp$ is a projector).

In a similar vein, we can analyse the trace-rescaling property $\tr[T_{\inp \out} \star W_\inp] = \gamma_\out/\gamma_\inp \tr[W]$ for all $W \in \Scal_\inp$. Following the same argument (and using the fact that $\mathbbm{1}_\inp$ is the Choi state of $\tr_\inp$), we obtain
\begin{gather}
\label{eqn::trEq1}
        \widetilde P_\inp[\tr_\out T_{\inp \out}] \star M =  \frac{\gamma_\out}{\gamma_\inp} \mathbbm{1}_\inp \star M\,.
\end{gather}
Again, this equality holds for all $M \in \L(\H_\inp)$, and thus implies 
\begin{gather}
\label{eqn::traceEq}
    \widetilde P_\inp[\tr_\out T_{\inp\out}] =  \frac{\gamma_\out}{\gamma_\inp} \mathbbm{1}_\inp\, .
\end{gather} 
Since $\widetilde P_\inp$ is unital and self-adjoint, it is trace-preserving (see Lem.~\ref{lem::PropMap}), and we see that $\tr [T_{\inp \out}] = \gamma_\out/\gamma_\inp \cdot d_\inp$. With this, by taking the tensor product of Eq.~\eqref{eqn::traceEq} with $\mathbbm{1}_\out/d_\out$, we obtain $\widetilde P_\inp[{}_\out T_{\inp \out}] =  {}_{\inp \out}T_{\inp \out}$, such that Eqs~\eqref{eqn::traceEq} and~\eqref{eqn::trEq1} can equivalently be written as 
\begin{gather}
\label{eqn::genStatement1}
    T_{\inp \out} = T_{\inp \out} - \widetilde P_\inp[{}_\out T_{\inp \out}] + {}_{\inp \out}T_{\inp \out}\,  \quad \text{and} \quad \tr [T_{\inp\out}] = \frac{\gamma_\out}{\gamma_\inp} d_\inp\, .
\end{gather}    
Now, inserting this into Eq.~\eqref{eqn::Proj1}, we obtain
\begin{gather}
\label{eqn::genStatement2}
    T_{\inp\out} = T_{\inp \out} - \widetilde P_\inp [T_{\inp\out}] + (\widetilde P_\inp \otimes \widetilde P_\out) [T_{\inp \out}]  - \widetilde P_\inp[{}_\out T_{\inp \out}] +  {}_{\inp \out}T_{\inp \out} =: \widetilde P_{\inp\out}[T_{\inp \out}] \quad \text{and} \quad \tr[T_{\inp\out}] = \frac{\gamma_\out}{\gamma_\inp} d_\inp\, ,
\end{gather}
which coincides exactly with Eqs.~\eqref{eqn::GenProj1} and~\eqref{eqn::GenProj2} of Thm.~\ref{thm::genFormProj}. For the converse direction, we first note that a self-adjoint, unital projector $\widetilde P_\out$ is trace-preserving, such that $\widetilde P_\out[{}_\out M] = {}_\out \circ \widetilde P_\out[M] = {}_\out M$ holds. With this, by direct insertion, it is easy to see that Eq.~\eqref{eqn::genStatement2} implies $\widetilde P_\inp[{}_\out T_{\inp\out}] = {}_{\inp\out}T_{\inp\out}$ and thus Eqs.~\eqref{eqn::genStatement1} and~\eqref{eqn::Proj1}; together with Eq.~\eqref{eqn::genStatement2}, these latter two equations directly lead to Eq.~\eqref{eqn::ThmLink}, thus proving Thm.~\ref{thm::genFormProj}. We emphasise, that this converse direction crucially requires the properties of $\widetilde P_\out$ (i.e., self-adjointness and unitality), while the forward direction also holds without these assumptions on $\widetilde P_\out$.

Finally, if we dropped the trace-rescaling property on $\widetilde T_{\inp\out}$ (and thus $T_{\inp\out}$), such that we only demand $\widetilde P_\out[T_{\inp\out} \star W_\inp] = T_{\inp\out} \star W_\inp$ for all $W_\inp = \widetilde P_\inp[W_\inp]${(but \textit{not} $\tr[T_{\inp\out} \star W_\inp] = \frac{\gamma_\out}{\gamma_\inp} \tr[W_\inp]$)}, following the above derivation, we arrive at 

\begin{theorem}[Transformation between linear spaces: specialised Choi version]
\label{thm::LinSpaceProj}
Let $\map{P}_\inp:\L(\H_\inp)\to\L(\H_\inp)$ and  $\map{P}_\out:\L(\H_\out)\to\L(\H_\out)$ be linear projective, self-adjoint and unital maps that commute with the transposition (or conjugation) and $\Scal_\inp \subseteq \L(\H_\inp) $ and $ \Scal_\out \subseteq \L(\H_\out)$ be linear spaces defined by
\begin{align}
\setlength{\arrayrulewidth}{1pt}
\renewcommand{\arraystretch}{1.2}
\begin{tabular}{|c|c|c|}
\hhline{|-|~|-|}
$W\in \L(\H_\inp) $ belongs to $\Scal_\inp$ iff  &\Large{$\mathbf{\longrightarrow}$} & $W'\in \L(\H_\out) $ belongs to $\Scal_\out$ iff  \\
$\map{P}_\inp[W]=W$. & \hspace*{20mm}& $\map{P}_\out[W']=W'$. \\
\hhline{|-|~|-|}
\end{tabular}
\end{align} 
A linear map $\map{T}_{\textup{\inp \out}}:\L(\H_\textup{\inp}) \rightarrow \L(\H_\textup{\out})$ satisfies
$
        \map{T}_{\textup{\inp \out}}[W]\in \Scal_\textup{\out},$ for all $ W \in \Scal_\textup{\inp} 
$
if and only if
\begin{empheq}[box=\fcolorbox{blu}{white}]{align}
\label{eqn::LinProj1}
T_{\inp\out} = T_{\inp\out} - (\map{P}_\inp \otimes \map{\id}_\out) [T_{\inp\out}] + (\map{P}_\inp \otimes \map{P}_\out) [T_{\inp\out}] =: \map{P}^{(\textup{ntr})}_{\inp\out}[T_{\inp\out}]\, ,
\end{empheq}

holds for its Choi matrix $T_{\inp\out}$, and $\map{P}^{(\text{ntr})}_{\inp\out}:\L(\H_\inp\otimes\H_\out)\to\L(\H_\inp\otimes\H_\out)$ is a self-adjoint, unital projector that commutes with the transposition. 
\end{theorem}
\begin{proof}
   The proof proceeds along the same line as the previous one, minus the additional requirement of a trace-rescaling property, i.e., it stops at Eq.~\eqref{eqn::genStatement1}, which coincides with Eq.~\eqref{eqn::LinProj1} of the Theorem. Conversely, it is easy to see that the above equality implies $\map{P}_\out[T_{\inp\out} \star W_\inp] = T_{\inp\out} \star W_\inp$ for all $W_\inp = \map{P}_\inp[W_\inp]$, independent of the properties of $\map{P}_\out$ (besides being a linear projector), proving Thm.~\ref{thm::LinSpaceProj}.
\end{proof}

Importantly, the above Theorem in not simply a special case of Thm.~\ref{thm::genFormProj}, particularly, it does \textit{not} coincide with it up to the affine constraint, but the respective constraints on $T_{\inp\out}$ are structurally different. This is a generalisation of the structural differences of, e.g., CP and CPTP maps; the former are not just equal to the latter up to a trace condition, but CPTP maps have an additional \textit{structural} property, that is absent in CP maps (namely, that the trace over the output degrees of freedom yields the identity matrix on the input space.).

In addition, we note that Thm.~\ref{thm::LinSpaceProj} covers the case $\gamma_\inp = 0$ of Thm.~\ref{thm::genFormProj}. As detailed in App.~\ref{app::gam0Choi}, in this case, \emph{both} $\gamma_\inp$ and $\gamma_\out$ are equal to zero, such that both spaces $\Scal_\inp$ and $\Scal_\out$ are entirely defined by linear projectors onto a vector space of traceless operators, which is a special instance of the scenario discussed in the above theorem. 

With this, we have considered all pertinent scenarios including projectors that are self-adjoint, unital, and commute with the transposition. We conclude this paper with the general case, where we impose \textit{no} constraints on $\map{P}_\inp$ and $\map{P}_\out$, besides them being {linear} projectors. 

\section{General approach}
\label{sec::GenAppr}
While physically the most relevant case, it is not necessary that the projectors $\map{P}_\inp$ and $\map{P}_{\out}$ defining the sets $\Scal_\inp$ and $\Scal_\out$, respectively, are self-adjoint, unital, and commute with the transposition. As a guiding example for a case where these properties fail to hold, consider the case where $\map{P}_\inp$ is given by a projector on the off-diagonal element $\ketbra{m}{n}$ (with $m\neq n$), i.e., it acts as $\widetilde P_\inp[B] = \braket{m}{B|n} \ketbra{m}{n}$. Naturally, the thusly defined $\map{P}_\inp$ is a projector (since it satisfies $\map{P}_\inp^2 = \map{P}_\inp$), but it is neither self-adjoint, unital, nor does it commute with the transposition. 

To derive the properties of maps $\map{T}_{\inp\out}$ between sets defined by such general projectors $\map{P}_\xt$, we can directly employ Lem.~\ref{lem::LinkGen}, which informs us how to move general linear operators around in the link product. With this, we deduce the concrete form of transformations $T_{\inp\out}$ in the same vein as the derivation for Thm.~\ref{thm::genFormProj} provided in Sec.~\ref{sec::ProofLink}.

\begin{theorem}
\label{thm::FullyGenProj}
Let $\map{P}_\inp:\L(\H_\inp)\to\L(\H_\inp)$ and  $\map{P}_\out:\L(\H_\out)\to\L(\H_\out)$ be linear projections and $\mathcal{S}_\textup{\inp} \subseteq \L(\H_\textup{\inp})$ and $\Scal_\out$ be affine spaces of matrices defined by 

\begin{align}
\setlength{\arrayrulewidth}{1pt}
\renewcommand{\arraystretch}{1.2}
\begin{tabular}{|c|c|c|}
\hhline{|-|~|-|}
$W\in \L(\H_\inp) $ belongs to $\Scal_\inp$ iff  & \hspace*{20mm}& $W'\in \L(\H_\out) $ belongs to $\Scal_\out$ iff  \\
$\map{P}_\inp[W]=W$, & \Large{$\mathbf{\longrightarrow}$}& $\map{P}_\out[W']=W'$, \\
$\tr[W] = \gamma_\inp$. && $ \tr[W']=\gamma_\out$.\\
\hhline{|-|~|-|}
\end{tabular}
\end{align} 
For $\gamma_\inp \neq 0$, a linear map $\map{T}_{\inp\out}: \L(\H_\inp) \rightarrow \L(\H_\out)$ satisfies $\widetilde{T}_{\inp\out}[W] \in \Scal_\out$ for all $W \in \Scal_\inp$ if and only if
\begin{subequations}
\begin{empheq}[box=\fcolorbox{blu}{white}]{align}
\label{eqn::FullyGenProj1}
    &T_{\textup{\inp \out}} = T_{\textup{\inp \out}} - \map{P}_\textup{\inp}^\tau [T_{\textup{\inp\out}}] + (\map{P}_\textup{\inp}^{\tau} \otimes \map{P}_\textup{\out}) [T_{\textup{\inp\out}}] =: \map{P}_{\inp\out}[T_{\textup{\inp \out}}],\\
\label{eqn::FullyGenProj2}
    &\map{P}_\textup{\inp}^{\tau}[(\tr_\textup{\out}T_{\textup{\inp\out}})] = \frac{\gamma_\out}{\gamma_\inp} \map{P}^\tau_\textup{\inp}[\mathbbm{1}_\textup{\inp}]\, .
\end{empheq}
\end{subequations}
\end{theorem}
Before providing a proof, we emphasise that the fact that we allow for non-unital, non-self-adjoint projectors implies that the respective sets $\Scal_\inp$ and $\Scal_\out$ do \textit{not} have to contain an element that is proportional to the identity matrix. The membership of the identity matrix facilitates many considerations in the literature when dealing with transformations between quantum objects (see, for example, Ref.~\cite{hoffreumon_projective_2022}). The relative ease with which the link product can be manipulated allows us to go beyond this case without much added difficulty (see Ex.~\ref{Ex::Gibbs}). 

Furthermore, we stress that the above Theorem exactly coincides with Thm.~\ref{thm::genFormProj} for the case of self-adjoint, unital projectors $\widetilde P_\inp$ that commute with the transposition. In this case, it is easy to see that Eq.~\eqref{eqn::FullyGenProj1} amounts to $T_{\inp\out} =T_{\textup{\inp \out}} - \widetilde P_\textup{\inp} [T_{\textup{\inp\out}}] + (\widetilde P_\textup{\inp} \otimes \widetilde P_\textup{\out}) [T_{\textup{\inp\out}}]$, while Eq.~\eqref{eqn::FullyGenProj2} implies $\widetilde P_\inp[\tr_\out T_{\inp\out}] = \gamma_\out/\gamma_\inp \mathbbm{1}_\inp$, which are exactly the properties we used in the proof of Thm.~\ref{thm::genFormProj}. 
\begin{proof}
The proof of Thm.~\ref{thm::FullyGenProj} proceeds along the same lines as that of Thm.~\ref{thm::genFormProj} with the difference that now the assumptions on the involved projectors are weaker. First, we note that, since $\gamma_\inp \neq 0$, we have $\text{span}(\Scal_\inp) = \map{P}_\inp[\L(\H_\inp)]$. Then, from $\map{P}_\out[T_{\inp\out} \star W_\inp] = T_{\inp\out} \star W_\inp $ for all $W_\inp \in \mathcal{S}_\inp$ we obtain
\begin{gather}
    \map{P}_\out[T_{\inp\out} \star \map{P}_\inp[M]] = (\map{P}_\inp^\tau \otimes \map{P}_\out)[T_{\inp\out}] \star M = T_{\inp\out} \star \map{P}_\inp[M] = \map{P}_\inp^\tau [T_{\inp\out}] \star M\,
\end{gather}
for all $M \in \L(\H_\inp)$, where $\map{P}_\inp^\tau$ has been defined in Eq.~\eqref{eqn::LinkGen}. From this, we directly obtain Eq.~\eqref{eqn::FullyGenProj1}. From the fact that $\tr[T_{\inp\out} \star W_\inp] = \gamma_\out$ for all $W_\inp \in \Scal_\inp$, it then follows that $\tr[T_{\inp\out} \star \widetilde P_\inp[M]] = \gamma_\out/\gamma_\inp \tr[\map{P}_\inp[M]]$ for all $M\in \L(\H_\inp)$. Using the fact that $\id_\xt$ is the Choi matrix of $\tr_\xt$, this can be written as $\tr[T_{\inp\out} \star \map{P}_\inp[M]] = \id_\inp \star \map{P}_\inp[M]$. Employing Lem.~\ref{lem::LinkGen} and using the fact that this equality holds for all $M\in \L(\H_\inp)$ then directly yields Eq.~\eqref{eqn::FullyGenProj2}. The fact that the resulting linear operator $\map{P}_{\inp\out} = \map{\id}_{\inp\out} - \widetilde P_\inp ^\tau + \widetilde P_\inp ^\tau \otimes \widetilde P_\out$ is indeed a projector can be seen by direct insertion and using the fact that, $\widetilde P_\inp = \widetilde P_\inp^2$ implies $\widetilde P_\inp^\tau = (\widetilde P_\inp^\tau)^2$. 

In the converse direction, using $\widetilde P_\inp^\tau = (\widetilde P_\inp^\tau)^2$, by direct insertion, it is easy to see that Eq.~\eqref{eqn::FullyGenProj1} implies $\map{P}_\out[T_{\inp\out} \star W_\inp] = T_{\inp\out} \star W_\inp$ for all $W_\inp \in \Scal_\inp$. 
\end{proof}
As for the previous Theorems, the case $\gamma_\inp = 0$ needs to be discussed in slightly more detail and is provided in App.~\ref{app::gam0Choi}. 

\begin{example}[Gibbs-preserving maps]
\label{Ex::Gibbs}
To see a concrete application of Thm.~\ref{thm::FullyGenProj}, consider the case of transformations that have a given fix point, e.g., the Gibbs state $\eta_1 = \exp{(-\beta H)}/\tr[\exp{(-\beta H)}] \in \L(\H_1)$ for some given inverse temperature $\beta$ and some Hamiltonian $H$. In this case, we have $\map{T}[\eta_1] = \eta_2 \in \L(\H_2)$ for some Gibbs state $\eta_2$. The projector on the input space is given by $\map{P}_\inp[X] = \tr[X\eta_1']\eta_1'=:{}_{\eta_1}Y$, while the projector on the output space is given by $\map{P}_\out[Y] = \tr[Y\eta_2']\eta_2' =: {}_{\eta_2}Y$, where $\eta_\xtt' := \eta_\xtt/\sqrt{\tr[\eta_\xtt^2]}$ (for $\xtt \in \{1, 2\}$), and we have $\gamma_\inp = \gamma_\out = 1$. Assuming that the Choi isomorphism on $\L(\H_1)$ is performed with respect to the eigenbasis of $\eta_1$, we have $\map{P}_\inp^\tau = \map{P}_\inp$ and $\map{P}_\inp^\tau[\id_1] = \eta_1/\tr[\eta_1^2] \neq \id_1$ for $\eta_1 \neq \id_1$. Consequently, the projector on the input space is non-unital, and we must apply Thm.~\ref{thm::FullyGenProj} (instead of Thm.~\ref{thm::genFormProj}) to deduce the properties of $T_{\inp\out}$. Employing Eqs.~\eqref{eqn::FullyGenProj1} and~\eqref{eqn::FullyGenProj2}, we obtain
\begin{gather}
    T = T - {}_{\eta_1}T  + {}_{\eta_1\eta_2}T \quad \& \quad \tr[T(\eta_1'\otimes \id_2)] = 1/\sqrt{\tr[\eta_1^2]}.
\end{gather}
These two constraints (together with $T\geq0$) entirely characterise the set of transformations that leave the Gibbs state invariant. Unsurprisingly, they do not enforce trace preservation (they only guarantee trace preservation on the span of $\eta_1$). Adding this as an additional constraint trivialises the second term of the above equation and leads to the following set of \textit{trace-preserving} Gibbs-preserving transformations:
\begin{gather}
    T = T - {}_{\eta_1}T  + {}_{\eta_1\eta_2}T =: \map{P}^{(\mathrm{GP})}[T] \quad \& \quad T = T - {}_\out T + {}_{\inp\out}T =: \map{P}^{(\mathrm{TP})}[T].
\end{gather}
Since $\map{P}^{(\mathrm{GP})}$ and $\map{P}^{(\mathrm{TP})}$ do not commute, this cannot be further combined into a single projector of the form $\map{P}^{(\mathrm{GP})} \circ \map{P}^{(\mathrm{TP})}$ or $\map{P}^{(\mathrm{TP})} \circ \map{P}^{(\mathrm{GP})}$.\hfill $\blacksquare$
\end{example}

Similarly to Thm.~\ref{thm::LinSpaceProj}, one can also drop the trace-rescaling property (i.a., the trace constraints on the elements of $\Scal_\inp$ and $\Scal_\out$) for general projectors. In this case, one would simply have to drop Eq.~\eqref{eqn::FullyGenProj2} in the above theorem to obtain the properties of $T_{\inp\out}$: 

\begin{theorem}[Transformation between linear spaces: Choi version]
\label{thm::TransChoiGen}
Let $\map{P}_\inp:\L(\H_\inp)\to\L(\H_\inp)$ and  $\map{P}_\out:\L(\H_\out)\to\L(\H_\out)$ be linear projections and $\mathcal{S}_\textup{\inp} \subseteq \L(\H_\textup{\inp})$ and $\Scal_\out$ be linear spaces of matrices defined by 

\begin{align}
\setlength{\arrayrulewidth}{1pt}
\renewcommand{\arraystretch}{1.2}
\begin{tabular}{|c|c|c|}
\hhline{|-|~|-|}
$W\in \L(\H_\inp) $ belongs to $\Scal_\inp$ iff  & \Large{$\mathbf{\longrightarrow}$}& $W'\in \L(\H_\out) $ belongs to $\Scal_\out$ iff  \\
$\map{P}_\inp[W]=W$. &\hspace*{20mm}& $\map{P}_\out[W']=W'$. \\
\hhline{|-|~|-|}
\end{tabular}
\end{align} 
A linear map $\map{T}_{\inp\out}: \L(\H_\inp) \rightarrow \L(\H_\out)$ satisfies $\widetilde{T}_{\inp\out}[W] \in \Scal_\out$ for all $W \in \Scal_\inp$ if and only if
\begin{empheq}[box=\fcolorbox{blu}{white}]{align}
\label{eqn::VecFullyGenProj1}
    &T_{\textup{\inp \out}} = T_{\textup{\inp \out}} - (\map{P}_\textup{\inp}^\tau \otimes \map{\id}_\out) [T_{\textup{\inp\out}}] + (\map{P}_\textup{\inp}^{\tau} \otimes \map{P}_\textup{\out}) [T_{\textup{\inp\out}}] =: \map{P}_{\inp\out}[T_{\textup{\inp \out}}].
\end{empheq}
\end{theorem}
As was the case for Thm.~\ref{thm::LinSpaceProj}, we note again, that this theorem also covers the case $\gamma_\inp = 0$ for the case where additional trace constraints are required of $\Scal_\inp$ and $\Scal_\out${(see App.~\ref{app::gam0Choi})}. 

\begin{example}[Projectors on off-diagonal terms] To provide a second concrete example for the above Theorems for the case of projectors that do not preserve Hermiticity, let us return to the simple case mentioned at the beginning of this Section, where $\map{P}_\textup{\inp}$ is given by a projection on the off-diagonal term $\ketbra{m}{n}$ (where $m\neq n$), i.e., it acts as $\map{P}_\textup{\inp}[M] = \braket{m}{M|n} \ketbra{m}{n}$, and let $\map{P}_\textup{\out}$ be a projector on the off-diagonal term $\ketbra{\alpha}{\beta} \in \L(\H_\textup{\out})$, where $\alpha \neq \beta$. With this, the set $\mathcal{S}_\textup{\inp}$ consists of all matrices $W_\textup{\inp}$ that are proportional to $\ketbra{m}{n}$ and the output space $\mathcal{S}_\textup{\out}$ consists of all matrices $W_\textup{\out}$ that are proportional to $\ketbra{\alpha}{\beta}$ (by construction, all elements of $\Scal_\inp$ and $\Scal_\out$ are traceless automatically due to the properties of $\map{P}_\inp$ and $\map{P}_\out$). It is easy to see (assuming that $\{\ket{m}\}_m$ and $\{\ket{\alpha}\}_\alpha$ constitute the canonical computational basis of $\H_\inp$ and $\H_\out$, respectively) that the action of $\map{P}_\textup{\inp}^\tau$ is given by $\map{P}_\textup{\inp}^\tau[M] = \ketbra{m}{m} M \ketbra{n}{n}$, while $\map{P}_\textup{\out}[M'] = \ketbra{\alpha}{\alpha} M' \ketbra{\beta}{\beta}$. Then, the properties of the Choi matrix $T_{\textup{\inp} \textup{\out}}$ of a transformation $\map{T}_{\textup{\inp\out}}: \mathcal{S}_\textup{\inp} \rightarrow \mathcal{S}_\textup{\out}$ follow directly from Eq.~\eqref{eqn::VecFullyGenProj1} of Thm.~\ref{thm::TransChoiGen} as 
\begin{gather}
T_{\textup{\inp\out}} = T_{\textup{\inp\out}} - \ketbra{m}{m}T_{\textup{\inp\out}} \ketbra{n}{n} + \ketbra{m\alpha}{m\alpha}T_{\textup{\inp\out}} \ketbra{n\beta}{n\beta}\, .
\end{gather}
With this, for any $\lambda \ketbra{m}{n} \in \Scal_\textup{\inp}$ (where $\lambda \in \mathbbm{C}$), we have 
\begin{gather}
\begin{split}
 T_\textup{\inp\out} \star \lambda \ketbra{m}{n} &= \lambda \tr_\textup{\inp}[T_{\textup{\inp\out}} \ketbra{n}{m}] \\
 &= \lambda \tr_\textup{\inp}[(T_{\textup{\inp\out}} - \ketbra{m}{m}T_{\textup{\inp\out}} \ketbra{n}{n} + \ketbra{m\alpha}{m\alpha}T_{\textup{\inp\out}} \ketbra{n\beta}{n\beta})\ketbra{n}{m}] \\
 &= \braket{m\alpha}{T_\textup{\inp\out}|n\beta} \ketbra{\alpha}{\beta} \in \Scal_\textup{\out}\,,
\end{split}
\end{gather}
i.e., $\widetilde T_{\inp \out}$ maps any element of $\Scal_\textup{\inp}$ onto an element of $\Scal_\textup{\out}$. \hfill $\blacksquare$
\end{example}

To finish this section, we provide a characterisation of the dual set $\overline{\Scal}$ for the case of general projectors, i.e., the generalised version of Thm.~\ref{thm::DualAff}:
\begin{theorem}[Dual affine for arbitrary projectors]
\label{thm::DualAffGen}
	Let $\mathcal{S}$ be a set defined by
\begin{empheq}[box=\fcolorbox{brow}{white}]{align}
		W \in \L(\H) \ \text{be} &\text{longs to } \Scal \text{ iff} \\
    W=&\map{P}[W], \\
		\tr[W]=&\gamma \, ,
\end{empheq}
where $\widetilde P$ is a projector and $\gamma \neq 0$. An operator $\overline{W}\in\L(\H)$ belongs to the dual affine set $\overline{\mathcal{S}}$ if and only if it satisfies 
\begin{empheq}[box=\fcolorbox{blu}{white}]{align}
\label{eqn::DualnonGen}
	\widetilde P^\tau[\overline W^\tau] = \frac{1}{\alpha}\widetilde P^\tau[\mathbbm{1}]\, ,
\end{empheq}
where $\widetilde P^\tau [\overline W^\tau] := P^\dagger_\textup{\inp}[{\overline W}^{\dagger}]^*$.
\end{theorem} 
\begin{proof}
The proof follows directly from Thm.~\ref{thm::FullyGenProj} in Sec.~\ref{sec::GenAppr}, where we derive the property of trace-rescaling mappings $\widetilde{T}: \L(\H_\inp) \rightarrow \L(\H_\out)$ (with rescaling factor $\gamma_\out/\gamma_\out$) between spaces defined by general linear projector $\widetilde P_\inp$ and  $\widetilde P_\out$. For this case, we have (see Eqs.~\eqref{eqn::FullyGenProj1} and~\eqref{eqn::FullyGenProj2})
\begin{align}
\label{eqn::ProofAffGen1}
&T_{\textup{\inp \out}} = T_{\textup{\inp \out}} - \widetilde P_\textup{\inp}^\tau [T_{\textup{\inp\out}}] + (\widetilde P_\textup{\inp}^{\tau} \otimes \widetilde P_\textup{\out}) [T_{\textup{\inp\out}}] \\
\label{eqn::ProofAffGen2}
\text{and} \quad &\widetilde P_\textup{\inp}^{\tau}[(\tr_\textup{\out}T_{\textup{\inp\out}})] = \frac{\gamma_\out}{\gamma_\inp} \widetilde P^\tau_\textup{\inp}[\mathbbm{1}_\textup{\inp}]\, .
\end{align}
For the case considered in the above Theorem, we have $\gamma_\out/\gamma_\inp = 1/\gamma$, $\widetilde P_\inp = \widetilde P$, and $\H_\out = \mathbbm{C}$, such that Eq.~\eqref{eqn::ProofAffGen1} becomes the trivial statement $T_{\inp\out} = T_{\inp\out}$. As mentioned above, for the convention of the CJI that we choose, we have $\overline W = T^\tau$. With this, Eq.~\eqref{eqn::ProofAffGen2} coincides exactly with Eq.~\eqref{eqn::DualnonGen} of the Theorem.
\end{proof}
Naturally, Thm.~\ref{thm::DualAffGen} contains Thm.~\ref{thm::DualAff}, where the properties of dual matrices for the case of self-adjoint, unital projectors that commute with the transposition were presented as special cases. To see this, recall that $\widetilde P$ is self-adjoint, unital, and commutes with the transposition it is also trace-preserving, such that Eq.~\eqref{eqn::DualnonGen} of Thm.~\ref{thm::DualAffGen} implies $\tr[\overline W] = \frac{d}{\alpha}$ [i.e., Eq.~\eqref{eqn::Dual2}]. Together with Eq.~\eqref{eqn::DualnonGen}, this yields Eq.~\eqref{eqn::Dual1} and we thus recover Thm.~\ref{thm::DualAff}.

\section{Applications for numerical computation and code availability}
As discussed previously, the projective characterisation of quantum objects analysed in this manuscript is also useful for tackling several problems by means of semidefinite programming. This approach was first presented at Ref.~\cite{araujo15witnessing}, where the authors derive an SDP for witnessing and quantifying indefinite causality in quantum theory. Since then, such methods have been employed in various other works and contexts, ranging from detecting indefinite causal order \cite{branciard_simplest_2015,branciard16_case_studies}, analysing quantum causal relations~\cite{nery21CCDC} and transforming quantum operations\cite{quintino19PRA,quintino22deterministic,yoshida23isometry,yoshida22_det_exact_inversion}, to the quantification of causal connection~\cite{milz_resource_2022} and channel discrimination\cite{bavaresco21,bavaresco21b}.

We have implemented all projective maps discussed in this manuscript and various other useful related functions in Matlab, and all our code is publicly available in an online repository~\cite{MTQ_github_HigherOrderProjectors}. It can be directly used for SDP problems on higher-order quantum problems and other corresponding SDP problems involving transformations between linear and affine sets. 

\section{Comparison with prior works}
General quantum transformation from a single quantum channel to a single quantum channel were first discussed in~\cite{chiribella08quantum_supermaps} in terms of deterministic quantum supermaps (here, quantum superchannels). Ref.~\cite{chiribella08quantum_supermaps} proves that all (one-slot) superchannels admit a decomposition in terms of a fixed ordere quantum circuit. Later, this concept of general transformations was extended to multipartite quantum channels which may be implemented as a sequential circuit, objects also referred to as causal-channels and channels with memory, or sequential channels. These ideas are discussed and detailed in Refs.~\cite{chiribella07circuit_architecture,chiribella09networks}, where a formalisation of $k$-slot quantum combs, the most general fixed order quantum transformations, can be found.

Later, Ref.~\cite{chiribella09_switch} considered more general quantum transformations, and analysed maps from bipartite non-signalling channels to quantum channels (due to linearity, this is equivalent to requiring that pairs of independent quantum channels being mapped to quantum channels). Ref.~\cite{chiribella09_switch} presented the quantum switch, a quantum transformation which does not have a fixed causal order. In similar vein, Ref.~\cite{oreshkov11} considered the scenario of transforming a pair of independent channels to unit probability (due to linearity, this is equivalent to transforming independent quantum instruments into probability distributions). Ref.~\cite{oreshkov11} is also the first paper to define definite causality in terms of convex combination of quantum transformations with fixed ordered.

The first work which characterised quantum transformations by means of projective maps is Ref.~\cite{araujo15witnessing}. This work focuses on the case of transformations of $N$ independent quantum channels to unit probability (this is equivalent to transforming $N$-partite non-signalling channels to unit probability), and presents a general method to obtain a projective characterisation for these scenarios. Later, this method was also adapted in Ref.~\cite{araujo16purification} to consider transformations from non-signalling channels into general channels. Subsequently, this method was used in Ref.~\cite{castro-ruiz_dynamics_2018} to characterise general transformations between bipartite process matrices.

In order to derive a resource theory of causal connection, Ref.~\cite{milz_resource_2022} introduced the concept of adapters, transformations between general sets of quantum objects. More precisely, this work focuses on transformations of quantum objects which can be written as independent quantum channels (due to linearity, equivalent to non-signalling channels) and arbitrary process matrices. The set of admissible adapters (Def.~1 of~\cite{milz_resource_2022}) is a special case of the set of quantum transformations we characterise in Thm.~\ref{thm::genFormProj}. Similarly to our work, Ref.~\cite{milz_resource_2022} characterises these transformations by means of linear projectors, but the corresponding formalism and proofs do not directly apply to all quantum sets, in particular not to many of the scenarios covered in Sec.~\ref{sec::GenTrans} as well as those discussed in Sec.~\ref{sec::GenAppr}.

In a similar vein to our considerations, Ref.~\cite{hoffreumon_projective_2022} has considered the question of characterising `structure-preserving maps', i.e., transformations between quantum sets. The resulting characterisation of structure-preserving maps (Prop.~3 of Ref.~\cite{hoffreumon_projective_2022}) is equivalent to Thm.~\ref{thm::genFormProj} presented in this work. Our results go beyond the cases considered in Ref.~\cite{hoffreumon_projective_2022} and cover a larger set of relevant quantum objects, specifically those characterised by non-unital projectors, and which do not obey the restrictions discussed in the beginning of Sec.~\ref{sec::ChoiCharSpec} (i.e., self-adjointness and commutation with the transposition). Consequently, all results from Sec.~\ref{sec::GenAppr} are novel and not covered by the methods of~\cite{hoffreumon_projective_2022} or other previous works. Additionally, our work also includes various proofs and techniques using the link product operation, considerably simplifying all derivations, which we believe to be of independent interest beyond their initial motivation.

Although not explored in greater generality, the concept of completely admissible transformations appears in Ref.~\cite{araujo16purification}, where the requirement that linear map which transforms process matrices to process matrices should be completely trace preserving is imposed. In a similar vein, Refs.~\cite{gour_comparison_2019, gour_inevitable_2024} discussed the concepts of `completely uniformity preserving' and `completely unital-channel preserving' superchannels, while Ref.~\cite{burniston_necessary_2020} introduced the concept of completely trace non-increasing maps to show that quantum testers (see Sec.~\ref{sec::DualAff}) are indeed the most general method to measure quantum channels. In Ref.~\cite{milz_resource_2022}, completely admissible transformations are identified for a particular choice of projectors and projector extensions, and it is demonstrated that the additional `completeness' requirement indeed changes the set of admissible transformations. In regard to `completeness' of properties, our work provides two main novel contributions: Our work formalises the concept of `completeness' for admissible transformations and provides a systematic way to deduce their properties.  We also present novel sufficient conditions for a transformation to be completely admissible.

Finally, the question of how to transform general quantum sets was also studied from a type-system~\cite{perinotti16higher,bisio_theoretical_2019}, category theory~\cite{Kissinger2017Categorical}, and linear logic~\cite{simmons_higher-order_2022,hoffreumon_projective_2022} perspective. While these works consider very similar questions to the ones of this manuscript, their methodological approach is different from ours, which is exclusively based on standard linear algebra.

\section{Discussion}
In this work, we have provided a systematic way to derive the properties of transformations between quantum sets. While a priori an abstract endeavour, such characterisations play an important role for many questions in quantum information theory -- in particular the study of causal order -- and our results offer a handy tool to deal with such problems in a simple and streamlined manner. We have demonstrated the versatility of our approach by explicitly showing its usefulness for a wide array of concrete examples of higher-order quantum maps, as well as the derivations of the properties of dual sets and probabilistic quantum operations. 

In addition, going beyond the cases generally considered in the literature, we have employed our approach to derive the properties of general transformations when, in addition, `completeness' of their admissibility is required. Here, as for the derivation of our main result, using the properties of Choi states and the link product allows one to straightforwardly deduce the properties of completely admissible transformations in a systematic way. In particular, we provided a concrete strategy to this end, and showed its versatility and applicability by means of concrete examples, simultaneously highlighting the dependence on the chosen extension of input and output projectors.

Importantly, our results solely rely on the properties of the link product, and do not require the respective sets we transform between (and, in particular, the projectors that define them) to have \textit{any} particular properties. Owing to this simplicity, we not only recovered structural properties of objects frequently encountered in quantum mechanics, but our results can readily be applied to any situation where the properties of a linear transformation are to be deduced from those of its input and output space. One such more general example, where, for example, the maximally mixed state is not a member of the quantum set $\Scal_\xtt$ (thus making the corresponding projector non-unital) is the set of Gibbs-preserving maps, which can readily be characterised using our approach.

Inferring the structural properties of their Choi matrices is a generic task when dealing with higher-order maps and/or trying to optimise an objective function over them. As such, the Theorems we derived in this work are of direct use to a whole host of problems in this field and substantially simplify the associated considerations. Additionally, the manipulation of the link product we introduce in order to derive the dual action of a map is a fruitful technique in its own right and can readily be employed to obtain more intuitive insights into a problem via its dual version, whenever its primal is somewhat opaque. Together, our results thus provide a powerful toolbox that is of direct applicability in a wide array of fields. 

\begin{acknowledgments}
We would like to thank Jessica Bavaresco and Esteban Castro-Ruiz for insightful discussions, and Timoth{\'e}e Hoffreumon and Ognyan Oreshkov for helpful clarifications on their work~\cite{hoffreumon_projective_2022}. S.M. acknowledges funding from the Austrian Science Fund (FWF): ZK3 (Zukunftkolleg) and Y879-N27 (START project), and from the European Union’s Horizon Europe research and innovation programme under the Marie Sk{\l}odowska-Curie grant agreement No. 101068332. This project/research was supported by grant number FQXi-RFP-IPW-1910 from the Foundational Questions Institute and Fetzer Franklin Fund, a donor advised fund of Silicon Valley Community Foundation.
\end{acknowledgments}

\nocite{apsrev42Control} 
\bibliographystyle{0_MTQ_apsrev4-2_corrected}

\bibliography{0_MTQ_bib.bib}

\begin{thebibliography}{59}%
\makeatletter
\providecommand \@ifxundefined [1]{%
 \@ifx{#1\undefined}
}%
\providecommand \@ifnum [1]{%
 \ifnum #1\expandafter \@firstoftwo
 \else \expandafter \@secondoftwo
 \fi
}%
\providecommand \@ifx [1]{%
 \ifx #1\expandafter \@firstoftwo
 \else \expandafter \@secondoftwo
 \fi
}%
\providecommand \natexlab [1]{#1}%
\providecommand \enquote  [1]{``#1''}%
\providecommand \bibnamefont  [1]{#1}%
\providecommand \bibfnamefont [1]{#1}%
\providecommand \citenamefont [1]{#1}%
\providecommand \href@noop [0]{\@secondoftwo}%
\providecommand \href [0]{\begingroup \@sanitize@url \@href}%
\providecommand \@href[1]{\@@startlink{#1}\@@href}%
\providecommand \@@href[1]{\endgroup#1\@@endlink}%
\providecommand \@sanitize@url [0]{\catcode `\\12\catcode `\$12\catcode `\&12\catcode `\#12\catcode `\^12\catcode `\_12\catcode `\%12\relax}%
\providecommand \@@startlink[1]{}%
\providecommand \@@endlink[0]{}%
\providecommand \url  [0]{\begingroup\@sanitize@url \@url }%
\providecommand \@url [1]{\endgroup\@href {#1}{\urlprefix }}%
\providecommand \urlprefix  [0]{URL }%
\providecommand \Eprint [0]{\href }%
\providecommand \doibase [0]{https://doi.org/}%
\providecommand \selectlanguage [0]{\@gobble}%
\providecommand \bibinfo  [0]{\@secondoftwo}%
\providecommand \bibfield  [0]{\@secondoftwo}%
\providecommand \translation [1]{[#1]}%
\providecommand \BibitemOpen [0]{}%
\providecommand \bibitemStop [0]{}%
\providecommand \bibitemNoStop [0]{.\EOS\space}%
\providecommand \EOS [0]{\spacefactor3000\relax}%
\providecommand \BibitemShut  [1]{\csname bibitem#1\endcsname}%
\let\auto@bib@innerbib\@empty
\bibitem [{\citenamefont {{Chiribella}}\ \emph {et~al.}(2008{\natexlab{a}})\citenamefont {{Chiribella}}, \citenamefont {{D'Ariano}},\ and\ \citenamefont {{Perinotti}}}]{chiribella08quantum_supermaps}%
  \BibitemOpen
  \bibfield  {author} {\bibinfo {author} {\bibfnamefont {G.}~\bibnamefont {{Chiribella}}}, \bibinfo {author} {\bibfnamefont {G.~M.}\ \bibnamefont {{D'Ariano}}},\ and\ \bibinfo {author} {\bibfnamefont {P.}~\bibnamefont {{Perinotti}}},\ }\bibfield  {title} {\bibinfo {title} {{Transforming quantum operations: Quantum supermaps}},\ }\href {https://doi.org/10.1209/0295-5075/83/30004} {\bibfield  {journal} {\bibinfo  {journal} {EPL}\ }\textbf {\bibinfo {volume} {83}},\ \bibinfo {pages} {30004} (\bibinfo {year} {2008}{\natexlab{a}})},\ \Eprint {https://arxiv.org/abs/0804.0180}{arXiv:0804.0180}\BibitemShut {NoStop}%
\bibitem [{\citenamefont {{Chiribella}}\ \emph {et~al.}(2008{\natexlab{b}})\citenamefont {{Chiribella}}, \citenamefont {{D'Ariano}},\ and\ \citenamefont {{Perinotti}}}]{chiribella07circuit_architecture}%
  \BibitemOpen
  \bibfield  {author} {\bibinfo {author} {\bibfnamefont {G.}~\bibnamefont {{Chiribella}}}, \bibinfo {author} {\bibfnamefont {G.~M.}\ \bibnamefont {{D'Ariano}}},\ and\ \bibinfo {author} {\bibfnamefont {P.}~\bibnamefont {{Perinotti}}},\ }\bibfield  {title} {\bibinfo {title} {{Quantum Circuit Architecture}},\ }\href {https://doi.org/10.1103/PhysRevLett.101.060401} {\bibfield  {journal} {\bibinfo  {journal} {Phys. Rev. Lett.}\ }\textbf {\bibinfo {volume} {101}},\ \bibinfo {eid} {060401} (\bibinfo {year} {2008}{\natexlab{b}})},\ \Eprint {https://arxiv.org/abs/0712.1325}{arXiv:0712.1325}\BibitemShut {NoStop}%
\bibitem [{\citenamefont {Gutoski}\ and\ \citenamefont {Watrous}(2007)}]{gutoski07quantum_games}%
  \BibitemOpen
  \bibfield  {author} {\bibinfo {author} {\bibfnamefont {G.}~\bibnamefont {Gutoski}}\ and\ \bibinfo {author} {\bibfnamefont {J.}~\bibnamefont {Watrous}},\ }\bibfield  {title} {\bibinfo {title} {Toward a general theory of quantum games},\ }in\ \href {https://doi.org/10.1145/1250790.1250873} {\emph {\bibinfo {booktitle} {Proceedings of the thirty-ninth annual ACM symposium on Theory of computing}}}\ (\bibinfo {year} {2007})\ pp.\ \bibinfo {pages} {565--574},\ \Eprint {https://arxiv.org/abs/quant-ph/0611234}{arXiv:quant-ph/0611234}\BibitemShut {NoStop}%
\bibitem [{\citenamefont {Pollock}\ \emph {et~al.}(2018)\citenamefont {Pollock}, \citenamefont {Rodr\'{\i}guez-Rosario}, \citenamefont {Frauenheim}, \citenamefont {Paternostro},\ and\ \citenamefont {Modi}}]{pollock15markovian}%
  \BibitemOpen
  \bibfield  {author} {\bibinfo {author} {\bibfnamefont {F.~A.}\ \bibnamefont {Pollock}}, \bibinfo {author} {\bibfnamefont {C.}~\bibnamefont {Rodr\'{\i}guez-Rosario}}, \bibinfo {author} {\bibfnamefont {T.}~\bibnamefont {Frauenheim}}, \bibinfo {author} {\bibfnamefont {M.}~\bibnamefont {Paternostro}},\ and\ \bibinfo {author} {\bibfnamefont {K.}~\bibnamefont {Modi}},\ }\bibfield  {title} {\bibinfo {title} {Non-markovian quantum processes: Complete framework and efficient characterization},\ }\href {https://doi.org/10.1103/PhysRevA.97.012127} {\bibfield  {journal} {\bibinfo  {journal} {Phys. Rev. A}\ }\textbf {\bibinfo {volume} {97}},\ \bibinfo {pages} {012127} (\bibinfo {year} {2018})},\ \Eprint {https://arxiv.org/abs/1512.00589}{arXiv:1512.00589}\BibitemShut {NoStop}%
\bibitem [{\citenamefont {{Chiribella}}\ \emph {et~al.}(2009)\citenamefont {{Chiribella}}, \citenamefont {{D'Ariano}},\ and\ \citenamefont {{Perinotti}}}]{chiribella09networks}%
  \BibitemOpen
  \bibfield  {author} {\bibinfo {author} {\bibfnamefont {G.}~\bibnamefont {{Chiribella}}}, \bibinfo {author} {\bibfnamefont {G.~M.}\ \bibnamefont {{D'Ariano}}},\ and\ \bibinfo {author} {\bibfnamefont {P.}~\bibnamefont {{Perinotti}}},\ }\bibfield  {title} {\bibinfo {title} {{Theoretical framework for quantum networks}},\ }\href {https://doi.org/10.1103/PhysRevA.80.022339} {\bibfield  {journal} {\bibinfo  {journal} {Phys. Rev.~A}\ }\textbf {\bibinfo {volume} {80}},\ \bibinfo {eid} {022339} (\bibinfo {year} {2009})},\ \Eprint {https://arxiv.org/abs/0904.4483}{arXiv:0904.4483}\BibitemShut {NoStop}%
\bibitem [{\citenamefont {{Ziman}}(2008)}]{ziman08_process_POVM}%
  \BibitemOpen
  \bibfield  {author} {\bibinfo {author} {\bibfnamefont {M.}~\bibnamefont {{Ziman}}},\ }\bibfield  {title} {\bibinfo {title} {{Process positive-operator-valued measure: A mathematical framework for the description of process tomography experiments}},\ }\href {https://doi.org/10.1103/PhysRevA.77.062112} {\bibfield  {journal} {\bibinfo  {journal} {Phys. Rev. A}\ }\textbf {\bibinfo {volume} {77}},\ \bibinfo {eid} {062112} (\bibinfo {year} {2008})},\ \Eprint {https://arxiv.org/abs/0802.3862}{arXiv:0802.3862}\BibitemShut {NoStop}%
\bibitem [{\citenamefont {{Bavaresco}}\ \emph {et~al.}(2021)\citenamefont {{Bavaresco}}, \citenamefont {{Murao}},\ and\ \citenamefont {{Quintino}}}]{bavaresco21}%
  \BibitemOpen
  \bibfield  {author} {\bibinfo {author} {\bibfnamefont {J.}~\bibnamefont {{Bavaresco}}}, \bibinfo {author} {\bibfnamefont {M.}~\bibnamefont {{Murao}}},\ and\ \bibinfo {author} {\bibfnamefont {M.~T.}\ \bibnamefont {{Quintino}}},\ }\bibfield  {title} {\bibinfo {title} {{Strict Hierarchy between Parallel, Sequential, and Indefinite-Causal-Order Strategies for Channel Discrimination}},\ }\href {https://doi.org/10.1103/PhysRevLett.127.200504} {\bibfield  {journal} {\bibinfo  {journal} {Phys. Rev. Lett.}\ }\textbf {\bibinfo {volume} {127}},\ \bibinfo {eid} {200504} (\bibinfo {year} {2021})},\ \Eprint {https://arxiv.org/abs/2011.08300}{arXiv:2011.08300}\BibitemShut {NoStop}%
\bibitem [{\citenamefont {{Oreshkov}}\ \emph {et~al.}(2012)\citenamefont {{Oreshkov}}, \citenamefont {{Costa}},\ and\ \citenamefont {{Brukner}}}]{oreshkov11}%
  \BibitemOpen
  \bibfield  {author} {\bibinfo {author} {\bibfnamefont {O.}~\bibnamefont {{Oreshkov}}}, \bibinfo {author} {\bibfnamefont {F.}~\bibnamefont {{Costa}}},\ and\ \bibinfo {author} {\bibfnamefont {{\v C}.}~\bibnamefont {{Brukner}}},\ }\bibfield  {title} {\bibinfo {title} {{Quantum correlations with no causal order}},\ }\href {https://doi.org/10.1038/ncomms2076} {\bibfield  {journal} {\bibinfo  {journal} {Nat. Commun.}\ }\textbf {\bibinfo {volume} {3}},\ \bibinfo {eid} {1092} (\bibinfo {year} {2012})},\ \Eprint {https://arxiv.org/abs/1105.4464}{arXiv:1105.4464}\BibitemShut {NoStop}%
\bibitem [{\citenamefont {{Kretschmann}}\ and\ \citenamefont {{Werner}}(2005)}]{kretschmann05channels_memory}%
  \BibitemOpen
  \bibfield  {author} {\bibinfo {author} {\bibfnamefont {D.}~\bibnamefont {{Kretschmann}}}\ and\ \bibinfo {author} {\bibfnamefont {R.~F.}\ \bibnamefont {{Werner}}},\ }\bibfield  {title} {\bibinfo {title} {{Quantum channels with memory}},\ }\href {https://doi.org/10.1103/PhysRevA.72.062323} {\bibfield  {journal} {\bibinfo  {journal} {Phys. Rev.~A}\ }\textbf {\bibinfo {volume} {72}},\ \bibinfo {eid} {062323} (\bibinfo {year} {2005})},\ \Eprint {https://arxiv.org/abs/quant-ph/0502106}{arXiv:quant-ph/0502106}\BibitemShut {NoStop}%
\bibitem [{\citenamefont {Castro-Ruiz}\ \emph {et~al.}(2018)\citenamefont {Castro-Ruiz}, \citenamefont {Giacomini},\ and\ \citenamefont {Brukner}}]{castro-ruiz_dynamics_2018}%
  \BibitemOpen
  \bibfield  {author} {\bibinfo {author} {\bibfnamefont {E.}~\bibnamefont {Castro-Ruiz}}, \bibinfo {author} {\bibfnamefont {F.}~\bibnamefont {Giacomini}},\ and\ \bibinfo {author} {\bibfnamefont {{\v C}.}~\bibnamefont {Brukner}},\ }\bibfield  {title} {\bibinfo {title} {Dynamics of {Quantum} {Causal} {Structures}},\ }\href {https://doi.org/10.1103/PhysRevX.8.011047} {\bibfield  {journal} {\bibinfo  {journal} {Phys. Rev. X}\ }\textbf {\bibinfo {volume} {8}},\ \bibinfo {pages} {011047} (\bibinfo {year} {2018})},\ \Eprint {https://arxiv.org/abs/1710.03139}{arXiv:1710.03139}\BibitemShut {NoStop}%
\bibitem [{\citenamefont {Perinotti}(2017)}]{perinotti16higher}%
  \BibitemOpen
  \bibfield  {author} {\bibinfo {author} {\bibfnamefont {P.}~\bibnamefont {Perinotti}},\ }\bibfield  {title} {\bibinfo {title} {{Causal structures and the classification of higher order quantum computations}},\ }\href {https://doi.org/10.1007/978-3-319-68655-4_7} {\bibfield  {journal} {\bibinfo  {journal} {Tutorials, Schools, and Workshops in the Mathematical Sciences\!\!\!}\ ,\ \bibinfo {pages} {103–127}} (\bibinfo {year} {2017})},\ \Eprint {https://arxiv.org/abs/1612.05099}{arXiv:1612.05099}\BibitemShut {NoStop}%
\bibitem [{\citenamefont {Bisio}\ and\ \citenamefont {Perinotti}(2019)}]{bisio_theoretical_2019}%
  \BibitemOpen
  \bibfield  {author} {\bibinfo {author} {\bibfnamefont {A.}~\bibnamefont {Bisio}}\ and\ \bibinfo {author} {\bibfnamefont {P.}~\bibnamefont {Perinotti}},\ }\bibfield  {title} {\bibinfo {title} {Theoretical framework for higher-order quantum theory},\ }\href {https://doi.org/10.1098/rspa.2018.0706} {\bibfield  {journal} {\bibinfo  {journal} {Proc. R. Soc. A}\ }\textbf {\bibinfo {volume} {475}},\ \bibinfo {pages} {20180706} (\bibinfo {year} {2019})},\ \Eprint {https://arxiv.org/abs/1806.09554}{arXiv:1806.09554}\BibitemShut {NoStop}%
\bibitem [{\citenamefont {Simmons}\ and\ \citenamefont {Kissinger}(2022)}]{simmons_higher-order_2022}%
  \BibitemOpen
  \bibfield  {author} {\bibinfo {author} {\bibfnamefont {W.}~\bibnamefont {Simmons}}\ and\ \bibinfo {author} {\bibfnamefont {A.}~\bibnamefont {Kissinger}},\ }\bibfield  {title} {\bibinfo {title} {Higher-order causal theories are models of {BV}-logic},\ }\href {https://www.arXiv.org/abs/2205.11219} {\bibfield  {journal} {\bibinfo  {journal} {arXiv:2205.11219}\ } (\bibinfo {year} {2022})}\BibitemShut {NoStop}%
\bibitem [{\citenamefont {Hoffreumon}\ and\ \citenamefont {Oreshkov}(2022)}]{hoffreumon_projective_2022}%
  \BibitemOpen
  \bibfield  {author} {\bibinfo {author} {\bibfnamefont {T.}~\bibnamefont {Hoffreumon}}\ and\ \bibinfo {author} {\bibfnamefont {O.}~\bibnamefont {Oreshkov}},\ }\bibfield  {title} {\bibinfo {title} {Projective characterization of higher-order quantum transformations},\ }\href@noop {} {\bibfield  {journal} {\bibinfo  {journal} {arXiv e-prints}\ } (\bibinfo {year} {2022})}\BibitemShut {NoStop}%
\bibitem [{\citenamefont {{Ara{\'u}jo}}\ \emph {et~al.}(2015)\citenamefont {{Ara{\'u}jo}}, \citenamefont {{Branciard}}, \citenamefont {{Costa}}, \citenamefont {{Feix}}, \citenamefont {{Giarmatzi}},\ and\ \citenamefont {{Brukner}}}]{araujo15witnessing}%
  \BibitemOpen
  \bibfield  {author} {\bibinfo {author} {\bibfnamefont {M.}~\bibnamefont {{Ara{\'u}jo}}}, \bibinfo {author} {\bibfnamefont {C.}~\bibnamefont {{Branciard}}}, \bibinfo {author} {\bibfnamefont {F.}~\bibnamefont {{Costa}}}, \bibinfo {author} {\bibfnamefont {A.}~\bibnamefont {{Feix}}}, \bibinfo {author} {\bibfnamefont {C.}~\bibnamefont {{Giarmatzi}}},\ and\ \bibinfo {author} {\bibfnamefont {{\v C}.}~\bibnamefont {{Brukner}}},\ }\bibfield  {title} {\bibinfo {title} {{Witnessing causal nonseparability}},\ }\href {https://doi.org/10.1088/1367-2630/17/10/102001} {\bibfield  {journal} {\bibinfo  {journal} {New J. Phys.}\ }\textbf {\bibinfo {volume} {17}},\ \bibinfo {eid} {102001} (\bibinfo {year} {2015})},\ \Eprint {https://arxiv.org/abs/1506.03776}{arXiv:1506.03776}\BibitemShut {NoStop}%
\bibitem [{\citenamefont {Milz}\ \emph {et~al.}(2022)\citenamefont {Milz}, \citenamefont {Bavaresco},\ and\ \citenamefont {Chiribella}}]{milz_resource_2022}%
  \BibitemOpen
  \bibfield  {author} {\bibinfo {author} {\bibfnamefont {S.}~\bibnamefont {Milz}}, \bibinfo {author} {\bibfnamefont {J.}~\bibnamefont {Bavaresco}},\ and\ \bibinfo {author} {\bibfnamefont {G.}~\bibnamefont {Chiribella}},\ }\bibfield  {title} {\bibinfo {title} {Resource theory of causal connection},\ }\href {https://doi.org/10.22331/q-2022-08-25-788} {\bibfield  {journal} {\bibinfo  {journal} {Quantum}\ }\textbf {\bibinfo {volume} {6}},\ \bibinfo {pages} {788} (\bibinfo {year} {2022})},\ \Eprint {https://arxiv.org/abs/2110.03233}{arXiv:2110.03233}\BibitemShut {NoStop}%
\bibitem [{\citenamefont {Apadula}\ \emph {et~al.}(2024)\citenamefont {Apadula}, \citenamefont {Bisio},\ and\ \citenamefont {Perinotti}}]{apadula22nosignalling}%
  \BibitemOpen
  \bibfield  {author} {\bibinfo {author} {\bibfnamefont {L.}~\bibnamefont {Apadula}}, \bibinfo {author} {\bibfnamefont {A.}~\bibnamefont {Bisio}},\ and\ \bibinfo {author} {\bibfnamefont {P.}~\bibnamefont {Perinotti}},\ }\bibfield  {title} {\bibinfo {title} {No-signalling constrains quantum computation with indefinite causal structure},\ }\href {https://doi.org/10.22331/q-2024-02-05-1241} {\bibfield  {journal} {\bibinfo  {journal} {Quantum}\ }\textbf {\bibinfo {volume} {8}},\ \bibinfo {pages} {1241} (\bibinfo {year} {2024})},\ \Eprint {https://arxiv.org/abs/2202.10214}{arXiv:2202.10214}\BibitemShut {NoStop}%
\bibitem [{\citenamefont {Ara{\'{u}}jo}\ \emph {et~al.}(2017)\citenamefont {Ara{\'{u}}jo}, \citenamefont {Feix}, \citenamefont {Navascu{\'{e}}s},\ and\ \citenamefont {Brukner}}]{araujo16purification}%
  \BibitemOpen
  \bibfield  {author} {\bibinfo {author} {\bibfnamefont {M.}~\bibnamefont {Ara{\'{u}}jo}}, \bibinfo {author} {\bibfnamefont {A.}~\bibnamefont {Feix}}, \bibinfo {author} {\bibfnamefont {M.}~\bibnamefont {Navascu{\'{e}}s}},\ and\ \bibinfo {author} {\bibfnamefont {{\v{C}}.}~\bibnamefont {Brukner}},\ }\bibfield  {title} {\bibinfo {title} {A purification postulate for quantum mechanics with indefinite causal order},\ }\href {https://doi.org/10.22331/q-2017-04-26-10} {\bibfield  {journal} {\bibinfo  {journal} {{Quantum}}\ }\textbf {\bibinfo {volume} {1}},\ \bibinfo {pages} {10} (\bibinfo {year} {2017})},\ \Eprint {https://arxiv.org/abs/1611.08535}{arXiv:1611.08535}\BibitemShut {NoStop}%
\bibitem [{\citenamefont {Burniston}\ \emph {et~al.}(2020)\citenamefont {Burniston}, \citenamefont {Grabowecky}, \citenamefont {Scandolo}, \citenamefont {Chiribella},\ and\ \citenamefont {Gour}}]{burniston_necessary_2020}%
  \BibitemOpen
  \bibfield  {author} {\bibinfo {author} {\bibfnamefont {J.}~\bibnamefont {Burniston}}, \bibinfo {author} {\bibfnamefont {M.}~\bibnamefont {Grabowecky}}, \bibinfo {author} {\bibfnamefont {C.~M.}\ \bibnamefont {Scandolo}}, \bibinfo {author} {\bibfnamefont {G.}~\bibnamefont {Chiribella}},\ and\ \bibinfo {author} {\bibfnamefont {G.}~\bibnamefont {Gour}},\ }\bibfield  {title} {\bibinfo {title} {Necessary and sufficient conditions on measurements of quantum channels},\ }\href {https://doi.org/10.1098/rspa.2019.0832} {\bibfield  {journal} {\bibinfo  {journal} {Proc. R. Soc. A}\ }\textbf {\bibinfo {volume} {476}},\ \bibinfo {pages} {20190832} (\bibinfo {year} {2020})},\ \Eprint {https://arxiv.org/abs/1904.09161}{arXiv:1904.09161}\BibitemShut {NoStop}%
\bibitem [{\citenamefont {{Wilson}}\ \emph {et~al.}(2023)\citenamefont {{Wilson}}, \citenamefont {{Chiribella}},\ and\ \citenamefont {{Kissinger}}}]{Wilson22Locality}%
  \BibitemOpen
  \bibfield  {author} {\bibinfo {author} {\bibfnamefont {M.}~\bibnamefont {{Wilson}}}, \bibinfo {author} {\bibfnamefont {G.}~\bibnamefont {{Chiribella}}},\ and\ \bibinfo {author} {\bibfnamefont {A.}~\bibnamefont {{Kissinger}}},\ }\bibfield  {title} {\bibinfo {title} {{Quantum Supermaps are Characterized by Locality}},\ }\href@noop {} {\bibfield  {journal} {\bibinfo  {journal} {arXiv e-prints}\ } (\bibinfo {year} {2023})},\ \Eprint {https://arxiv.org/abs/2205.09844}{arXiv:2205.09844 [quant-ph]}\BibitemShut {NoStop}%
\bibitem [{\citenamefont {{Wilson}}\ and\ \citenamefont {{Ormrod}}(2023)}]{Wilson23linearity}%
  \BibitemOpen
  \bibfield  {author} {\bibinfo {author} {\bibfnamefont {M.}~\bibnamefont {{Wilson}}}\ and\ \bibinfo {author} {\bibfnamefont {N.}~\bibnamefont {{Ormrod}}},\ }\bibfield  {title} {\bibinfo {title} {{On the Origin of Linearity and Unitarity in Quantum Theory}},\ }\href@noop {} {\bibfield  {journal} {\bibinfo  {journal} {arXiv e-prints}\ } (\bibinfo {year} {2023})},\ \Eprint {https://arxiv.org/abs/2305.20063}{arXiv:2305.20063 [quant-ph]}\BibitemShut {NoStop}%
\bibitem [{\citenamefont {de~Pillis}(1967)}]{pills67}%
  \BibitemOpen
  \bibfield  {author} {\bibinfo {author} {\bibfnamefont {J.}~\bibnamefont {de~Pillis}},\ }\bibfield  {title} {\bibinfo {title} {Linear transformations which preserve hermitian and positive semidefinite operators},\ }\href {https://doi.org/10.2140/pjm.1967.23.129} {\bibfield  {journal} {\bibinfo  {journal} {Pac. J. Math.}\ }\textbf {\bibinfo {volume} {23}},\ \bibinfo {pages} {129--137} (\bibinfo {year} {1967})}\BibitemShut {NoStop}%
\bibitem [{\citenamefont {Jamio{\l}kowski}(1972)}]{jamiolkowski72}%
  \BibitemOpen
  \bibfield  {author} {\bibinfo {author} {\bibfnamefont {A.}~\bibnamefont {Jamio{\l}kowski}},\ }\bibfield  {title} {\bibinfo {title} {Linear transformations which preserve trace and positive semidefiniteness of operators},\ }\href {https://doi.org/https://doi.org/10.1016/0034-4877(72)90011-0} {\bibfield  {journal} {\bibinfo  {journal} {Rep. Math. Phys.}\ }\textbf {\bibinfo {volume} {3}},\ \bibinfo {pages} {275--278} (\bibinfo {year} {1972})}\BibitemShut {NoStop}%
\bibitem [{\citenamefont {Choi}(1975)}]{choi75}%
  \BibitemOpen
  \bibfield  {author} {\bibinfo {author} {\bibfnamefont {M.-D.}\ \bibnamefont {Choi}},\ }\bibfield  {title} {\bibinfo {title} {Completely positive linear maps on complex matrices},\ }\href {https://doi.org/https://doi.org/10.1016/0024-3795(75)90075-0} {\bibfield  {journal} {\bibinfo  {journal} {Linear Algebra Appl.}\ }\textbf {\bibinfo {volume} {10}},\ \bibinfo {pages} {285 -- 290} (\bibinfo {year} {1975})}\BibitemShut {NoStop}%
\bibitem [{\citenamefont {{Milz}}\ and\ \citenamefont {{Modi}}(2021)}]{milz21markov}%
  \BibitemOpen
  \bibfield  {author} {\bibinfo {author} {\bibfnamefont {S.}~\bibnamefont {{Milz}}}\ and\ \bibinfo {author} {\bibfnamefont {K.}~\bibnamefont {{Modi}}},\ }\bibfield  {title} {\bibinfo {title} {{Quantum Stochastic Processes and Quantum non-Markovian Phenomena}},\ }\href {https://doi.org/10.1103/PRXQuantum.2.030201} {\bibfield  {journal} {\bibinfo  {journal} {PRX Quantum}\ }\textbf {\bibinfo {volume} {2}},\ \bibinfo {eid} {030201} (\bibinfo {year} {2021})},\ \Eprint {https://arxiv.org/abs/2012.01894}{arXiv:2012.01894}\BibitemShut {NoStop}%
\bibitem [{\citenamefont {{Beckman}}\ \emph {et~al.}(2001)\citenamefont {{Beckman}}, \citenamefont {{Gottesman}}, \citenamefont {{Nielsen}},\ and\ \citenamefont {{Preskill}}}]{backman01causal_operations}%
  \BibitemOpen
  \bibfield  {author} {\bibinfo {author} {\bibfnamefont {D.}~\bibnamefont {{Beckman}}}, \bibinfo {author} {\bibfnamefont {D.}~\bibnamefont {{Gottesman}}}, \bibinfo {author} {\bibfnamefont {M.~A.}\ \bibnamefont {{Nielsen}}},\ and\ \bibinfo {author} {\bibfnamefont {J.}~\bibnamefont {{Preskill}}},\ }\bibfield  {title} {\bibinfo {title} {{Causal and localizable quantum operations}},\ }\href {https://doi.org/10.1103/PhysRevA.64.052309} {\bibfield  {journal} {\bibinfo  {journal} {Phys. Rev.~A}\ }\textbf {\bibinfo {volume} {64}},\ \bibinfo {pages} {052309} (\bibinfo {year} {2001})},\ \Eprint {https://arxiv.org/abs/quant-ph/0102043}{arXiv:quant-ph/0102043}\BibitemShut {NoStop}%
\bibitem [{\citenamefont {{Piani}}\ \emph {et~al.}(2006)\citenamefont {{Piani}}, \citenamefont {{Horodecki}}, \citenamefont {{Horodecki}},\ and\ \citenamefont {{Horodecki}}}]{piani06nonsignalling}%
  \BibitemOpen
  \bibfield  {author} {\bibinfo {author} {\bibfnamefont {M.}~\bibnamefont {{Piani}}}, \bibinfo {author} {\bibfnamefont {M.}~\bibnamefont {{Horodecki}}}, \bibinfo {author} {\bibfnamefont {P.}~\bibnamefont {{Horodecki}}},\ and\ \bibinfo {author} {\bibfnamefont {R.}~\bibnamefont {{Horodecki}}},\ }\bibfield  {title} {\bibinfo {title} {{Properties of quantum nonsignaling boxes}},\ }\href {https://doi.org/10.1103/PhysRevA.74.012305} {\bibfield  {journal} {\bibinfo  {journal} {Phys. Rev.~A}\ }\textbf {\bibinfo {volume} {74}},\ \bibinfo {eid} {012305} (\bibinfo {year} {2006})},\ \Eprint {https://arxiv.org/abs/quant-ph/0505110}{arXiv:quant-ph/0505110}\BibitemShut {NoStop}%
\bibitem [{\citenamefont {Modi}(2012)}]{modi_operational_2012}%
  \BibitemOpen
  \bibfield  {author} {\bibinfo {author} {\bibfnamefont {K.}~\bibnamefont {Modi}},\ }\bibfield  {title} {\bibinfo {title} {Operational approach to open dynamics and quantifying initial correlations},\ }\href {http://www.nature.com/articles/srep00581} {\bibfield  {journal} {\bibinfo  {journal} {Sci. Rep.}\ }\textbf {\bibinfo {volume} {2}},\ \bibinfo {pages} {581} (\bibinfo {year} {2012})},\ \Eprint {https://arxiv.org/abs/1011.6138}{arXiv:1011.6138}\BibitemShut {NoStop}%
\bibitem [{\citenamefont {Chiribella}\ \emph {et~al.}(2009)\citenamefont {Chiribella}, \citenamefont {D’Ariano},\ and\ \citenamefont {Perinotti}}]{chiribella2009theoretical}%
  \BibitemOpen
  \bibfield  {author} {\bibinfo {author} {\bibfnamefont {G.}~\bibnamefont {Chiribella}}, \bibinfo {author} {\bibfnamefont {G.~M.}\ \bibnamefont {D’Ariano}},\ and\ \bibinfo {author} {\bibfnamefont {P.}~\bibnamefont {Perinotti}},\ }\bibfield  {title} {\bibinfo {title} {Theoretical framework for quantum networks},\ }\href {https://doi.org/https://doi.org/10.1103/PhysRevA.80.022339} {\bibfield  {journal} {\bibinfo  {journal} {Phys. Rev.~A}\ }\textbf {\bibinfo {volume} {80}},\ \bibinfo {pages} {022339} (\bibinfo {year} {2009})},\ \Eprint {https://arxiv.org/abs/0904.4483}{arXiv:0904.4483}\BibitemShut {NoStop}%
\bibitem [{\citenamefont {Quintino}\ and\ \citenamefont {Ebler}(2022)}]{quintino22deterministic}%
  \BibitemOpen
  \bibfield  {author} {\bibinfo {author} {\bibfnamefont {M.~T.}\ \bibnamefont {Quintino}}\ and\ \bibinfo {author} {\bibfnamefont {D.}~\bibnamefont {Ebler}},\ }\bibfield  {title} {\bibinfo {title} {Deterministic transformations between unitary operations: {E}xponential advantage with adaptive quantum circuits and the power of indefinite causality},\ }\href {https://doi.org/10.22331/q-2022-03-31-679} {\bibfield  {journal} {\bibinfo  {journal} {{Quantum}}\ }\textbf {\bibinfo {volume} {6}},\ \bibinfo {pages} {679} (\bibinfo {year} {2022})},\ \Eprint {https://arxiv.org/abs/2109.08202}{arXiv:2109.08202}\BibitemShut {NoStop}%
\bibitem [{\citenamefont {{Oreshkov}}\ and\ \citenamefont {{Giarmatzi}}(2016)}]{oreschkov2016definition}%
  \BibitemOpen
  \bibfield  {author} {\bibinfo {author} {\bibfnamefont {O.}~\bibnamefont {{Oreshkov}}}\ and\ \bibinfo {author} {\bibfnamefont {C.}~\bibnamefont {{Giarmatzi}}},\ }\bibfield  {title} {\bibinfo {title} {{Causal and causally separable processes}},\ }\href {https://doi.org/10.1088/1367-2630/18/9/093020} {\bibfield  {journal} {\bibinfo  {journal} {New J. Phys.}\ }\textbf {\bibinfo {volume} {18}},\ \bibinfo {eid} {093020} (\bibinfo {year} {2016})},\ \Eprint {https://arxiv.org/abs/1506.05449}{arXiv:1506.05449}\BibitemShut {NoStop}%
\bibitem [{\citenamefont {{Wechs}}\ \emph {et~al.}(2019)\citenamefont {{Wechs}}, \citenamefont {{Abbott}},\ and\ \citenamefont {{Branciard}}}]{Wechs2019MultipartiteCausality}%
  \BibitemOpen
  \bibfield  {author} {\bibinfo {author} {\bibfnamefont {J.}~\bibnamefont {{Wechs}}}, \bibinfo {author} {\bibfnamefont {A.~A.}\ \bibnamefont {{Abbott}}},\ and\ \bibinfo {author} {\bibfnamefont {C.}~\bibnamefont {{Branciard}}},\ }\bibfield  {title} {\bibinfo {title} {{On the definition and characterisation of multipartite causal (non)separability}},\ }\href {https://doi.org/10.1088/1367-2630/aaf352} {\bibfield  {journal} {\bibinfo  {journal} {New J. Phys.}\ }\textbf {\bibinfo {volume} {21}},\ \bibinfo {eid} {013027} (\bibinfo {year} {2019})},\ \Eprint {https://arxiv.org/abs/1807.10557}{arXiv:1807.10557}\BibitemShut {NoStop}%
\bibitem [{\citenamefont {{Wechs}}\ \emph {et~al.}(2021)\citenamefont {{Wechs}}, \citenamefont {{Dourdent}}, \citenamefont {{Abbott}},\ and\ \citenamefont {{Branciard}}}]{wechs21control}%
  \BibitemOpen
  \bibfield  {author} {\bibinfo {author} {\bibfnamefont {J.}~\bibnamefont {{Wechs}}}, \bibinfo {author} {\bibfnamefont {H.}~\bibnamefont {{Dourdent}}}, \bibinfo {author} {\bibfnamefont {A.~A.}\ \bibnamefont {{Abbott}}},\ and\ \bibinfo {author} {\bibfnamefont {C.}~\bibnamefont {{Branciard}}},\ }\bibfield  {title} {\bibinfo {title} {{Quantum Circuits with Classical Versus Quantum Control of Causal Order}},\ }\href {https://doi.org/10.1103/PRXQuantum.2.030335} {\bibfield  {journal} {\bibinfo  {journal} {PRX Quantum}\ }\textbf {\bibinfo {volume} {2}},\ \bibinfo {eid} {030335} (\bibinfo {year} {2021})},\ \Eprint {https://arxiv.org/abs/2101.08796}{arXiv:2101.08796}\BibitemShut {NoStop}%
\bibitem [{\citenamefont {Gutoski}(2009)}]{Gutoski09}%
  \BibitemOpen
  \bibfield  {author} {\bibinfo {author} {\bibfnamefont {G.}~\bibnamefont {Gutoski}},\ }\bibfield  {title} {\bibinfo {title} {Properties of local quantum operations with shared entanglement},\ }\href {https://dl.acm.org/doi/10.5555/2011804.2011806} {\bibfield  {journal} {\bibinfo  {journal} {Quantum Inf. Comput.}\ }\textbf {\bibinfo {volume} {9}},\ \bibinfo {pages} {739--764} (\bibinfo {year} {2009})},\ \Eprint {https://arxiv.org/abs/0805.2209}{arXiv:0805.2209}\BibitemShut {NoStop}%
\bibitem [{\citenamefont {{Chiribella}}\ \emph {et~al.}(2013)\citenamefont {{Chiribella}}, \citenamefont {{D'Ariano}}, \citenamefont {{Perinotti}},\ and\ \citenamefont {{Valiron}}}]{chiribella09_switch}%
  \BibitemOpen
  \bibfield  {author} {\bibinfo {author} {\bibfnamefont {G.}~\bibnamefont {{Chiribella}}}, \bibinfo {author} {\bibfnamefont {G.~M.}\ \bibnamefont {{D'Ariano}}}, \bibinfo {author} {\bibfnamefont {P.}~\bibnamefont {{Perinotti}}},\ and\ \bibinfo {author} {\bibfnamefont {B.}~\bibnamefont {{Valiron}}},\ }\bibfield  {title} {\bibinfo {title} {{Quantum computations without definite causal structure}},\ }\href {https://doi.org/10.1103/PhysRevA.88.022318} {\bibfield  {journal} {\bibinfo  {journal} {Phys. Rev.~A}\ }\textbf {\bibinfo {volume} {88}},\ \bibinfo {eid} {022318} (\bibinfo {year} {2013})},\ \Eprint {https://arxiv.org/abs/0912.0195}{arXiv:0912.0195}\BibitemShut {NoStop}%
\bibitem [{\citenamefont {Branciard}\ \emph {et~al.}(2015)\citenamefont {Branciard}, \citenamefont {Ara{\'u}jo}, \citenamefont {Feix}, \citenamefont {Costa},\ and\ \citenamefont {Brukner}}]{branciard_simplest_2015}%
  \BibitemOpen
  \bibfield  {author} {\bibinfo {author} {\bibfnamefont {C.}~\bibnamefont {Branciard}}, \bibinfo {author} {\bibfnamefont {M.}~\bibnamefont {Ara{\'u}jo}}, \bibinfo {author} {\bibfnamefont {A.}~\bibnamefont {Feix}}, \bibinfo {author} {\bibfnamefont {F.}~\bibnamefont {Costa}},\ and\ \bibinfo {author} {\bibfnamefont {{\v C}.}~\bibnamefont {Brukner}},\ }\bibfield  {title} {\bibinfo {title} {The simplest causal inequalities and their violation},\ }\href {https://doi.org/10.1088/1367-2630/18/1/013008} {\bibfield  {journal} {\bibinfo  {journal} {New J. Phys.}\ }\textbf {\bibinfo {volume} {18}},\ \bibinfo {pages} {013008} (\bibinfo {year} {2015})},\ \Eprint {https://arxiv.org/abs/1508.01704}{arXiv:1508.01704}\BibitemShut {NoStop}%
\bibitem [{\citenamefont {{Chiribella}}\ and\ \citenamefont {{Liu}}(2022)}]{chiribella20timeflip}%
  \BibitemOpen
  \bibfield  {author} {\bibinfo {author} {\bibfnamefont {G.}~\bibnamefont {{Chiribella}}}\ and\ \bibinfo {author} {\bibfnamefont {Z.}~\bibnamefont {{Liu}}},\ }\bibfield  {title} {\bibinfo {title} {{Quantum operations with indefinite time direction}},\ }\href {https://doi.org/10.1038/s42005-022-00967-3} {\bibfield  {journal} {\bibinfo  {journal} {Commun. Phys.}\ }\textbf {\bibinfo {volume} {5}},\ \bibinfo {eid} {190} (\bibinfo {year} {2022})},\ \Eprint {https://arxiv.org/abs/2012.03859}{arXiv:2012.03859}\BibitemShut {NoStop}%
\bibitem [{\citenamefont {{Str{\"o}mberg}}\ \emph {et~al.}(2024)\citenamefont {{Str{\"o}mberg}}, \citenamefont {{Schiansky}}, \citenamefont {{Quintino}}, \citenamefont {{Antesberger}}, \citenamefont {{Rozema}}, \citenamefont {{Agresti}}, \citenamefont {{Brukner}},\ and\ \citenamefont {{Walther}}}]{Stromberg2024timeflip}%
  \BibitemOpen
  \bibfield  {author} {\bibinfo {author} {\bibfnamefont {T.}~\bibnamefont {{Str{\"o}mberg}}}, \bibinfo {author} {\bibfnamefont {P.}~\bibnamefont {{Schiansky}}}, \bibinfo {author} {\bibfnamefont {M.~T.}\ \bibnamefont {{Quintino}}}, \bibinfo {author} {\bibfnamefont {M.}~\bibnamefont {{Antesberger}}}, \bibinfo {author} {\bibfnamefont {L.~A.}\ \bibnamefont {{Rozema}}}, \bibinfo {author} {\bibfnamefont {I.}~\bibnamefont {{Agresti}}}, \bibinfo {author} {\bibfnamefont {{\v{C}}.}~\bibnamefont {{Brukner}}},\ and\ \bibinfo {author} {\bibfnamefont {P.}~\bibnamefont {{Walther}}},\ }\bibfield  {title} {\bibinfo {title} {{Experimental superposition of a quantum evolution with its time reverse}},\ }\href {https://doi.org/10.1103/PhysRevResearch.6.023071} {\bibfield  {journal} {\bibinfo  {journal} {Phys. Rev. Research}\ }\textbf {\bibinfo {volume} {6}},\ \bibinfo {eid} {023071} (\bibinfo {year} {2024})},\ \Eprint {https://arxiv.org/abs/2211.01283}{arXiv:2211.01283}\BibitemShut {NoStop}%
\bibitem [{\citenamefont {{Guo}}\ \emph {et~al.}(2024)\citenamefont {{Guo}}, \citenamefont {{Liu}}, \citenamefont {{Tang}}, \citenamefont {{Hu}}, \citenamefont {{Liu}}, \citenamefont {{Huang}}, \citenamefont {{Li}}, \citenamefont {{Guo}},\ and\ \citenamefont {{Chiribella}}}]{Guo2024timeflip}%
  \BibitemOpen
  \bibfield  {author} {\bibinfo {author} {\bibfnamefont {Y.}~\bibnamefont {{Guo}}}, \bibinfo {author} {\bibfnamefont {Z.}~\bibnamefont {{Liu}}}, \bibinfo {author} {\bibfnamefont {H.}~\bibnamefont {{Tang}}}, \bibinfo {author} {\bibfnamefont {X.-M.}\ \bibnamefont {{Hu}}}, \bibinfo {author} {\bibfnamefont {B.-H.}\ \bibnamefont {{Liu}}}, \bibinfo {author} {\bibfnamefont {Y.-F.}\ \bibnamefont {{Huang}}}, \bibinfo {author} {\bibfnamefont {C.-F.}\ \bibnamefont {{Li}}}, \bibinfo {author} {\bibfnamefont {G.-C.}\ \bibnamefont {{Guo}}},\ and\ \bibinfo {author} {\bibfnamefont {G.}~\bibnamefont {{Chiribella}}},\ }\bibfield  {title} {\bibinfo {title} {{Experimental Demonstration of Input-Output Indefiniteness in a Single Quantum Device}},\ }\href {https://doi.org/10.1103/PhysRevLett.132.160201} {\bibfield  {journal} {\bibinfo  {journal} {Phys. Rev. Lett.}\ }\textbf {\bibinfo {volume} {132}},\ \bibinfo {eid} {160201} (\bibinfo {year} {2024})},\ \Eprint {https://arxiv.org/abs/2210.17046}{arXiv:2210.17046}\BibitemShut {NoStop}%
\bibitem [{\citenamefont {Faist}\ \emph {et~al.}(2015)\citenamefont {Faist}, \citenamefont {Oppenheim},\ and\ \citenamefont {Renner}}]{faist_gibbs-preserving_2015}%
  \BibitemOpen
  \bibfield  {author} {\bibinfo {author} {\bibfnamefont {P.}~\bibnamefont {Faist}}, \bibinfo {author} {\bibfnamefont {J.}~\bibnamefont {Oppenheim}},\ and\ \bibinfo {author} {\bibfnamefont {R.}~\bibnamefont {Renner}},\ }\bibfield  {title} {\bibinfo {title} {Gibbs-preserving maps outperform thermal operations in the quantum regime},\ }\href {https://doi.org/10.1088/1367-2630/17/4/043003} {\bibfield  {journal} {\bibinfo  {journal} {New J. Phys.}\ }\textbf {\bibinfo {volume} {17}},\ \bibinfo {pages} {043003} (\bibinfo {year} {2015})},\ \Eprint {https://arxiv.org/abs/1406.3618}{arXiv:1406.3618}\BibitemShut {NoStop}%
\bibitem [{\citenamefont {Lostaglio}(2019)}]{lostaglio_introductory_2019}%
  \BibitemOpen
  \bibfield  {author} {\bibinfo {author} {\bibfnamefont {M.}~\bibnamefont {Lostaglio}},\ }\bibfield  {title} {\bibinfo {title} {An introductory review of the resource theory approach to thermodynamics},\ }\href {https://doi.org/10.1088/1361-6633/ab46e5} {\bibfield  {journal} {\bibinfo  {journal} {Rep. Prog. Phys.}\ }\textbf {\bibinfo {volume} {82}},\ \bibinfo {pages} {114001} (\bibinfo {year} {2019})},\ \Eprint {https://arxiv.org/abs/1807.11549}{arXiv:1807.11549}\BibitemShut {NoStop}%
\bibitem [{\citenamefont {Davis}(1976)}]{davis_quantum_1976}%
  \BibitemOpen
  \bibfield  {author} {\bibinfo {author} {\bibfnamefont {E.~B.}\ \bibnamefont {Davis}},\ }\href@noop {} {{\selectlanguage {English}\emph {\bibinfo {title} {Quantum {Theory} of {Open} {Systems}}}}}\ (\bibinfo  {publisher} {Academic Press Inc},\ \bibinfo {address} {London; New York},\ \bibinfo {year} {1976})\BibitemShut {NoStop}%
\bibitem [{\citenamefont {Lindblad}(1979)}]{Lindblad1979}%
  \BibitemOpen
  \bibfield  {author} {\bibinfo {author} {\bibfnamefont {G.}~\bibnamefont {Lindblad}},\ }\bibfield  {title} {\bibinfo {title} {{Non-Markovian quantum stochastic processes and their entropy}},\ }\href {http://dx.doi.org/10.1007/BF01197883} {\bibfield  {journal} {\bibinfo  {journal} {Comm. Math. Phys.}\ }\textbf {\bibinfo {volume} {65}},\ \bibinfo {pages} {281--294} (\bibinfo {year} {1979})}\BibitemShut {NoStop}%
\bibitem [{\citenamefont {{Chiribella}}\ and\ \citenamefont {{Ebler}}(2016)}]{ebler16}%
  \BibitemOpen
  \bibfield  {author} {\bibinfo {author} {\bibfnamefont {G.}~\bibnamefont {{Chiribella}}}\ and\ \bibinfo {author} {\bibfnamefont {D.}~\bibnamefont {{Ebler}}},\ }\bibfield  {title} {\bibinfo {title} {{Optimal quantum networks and one-shot entropies}},\ }\href {https://doi.org/10.1088/1367-2630/18/9/093053} {\bibfield  {journal} {\bibinfo  {journal} {New J. Phys.}\ }\textbf {\bibinfo {volume} {18}},\ \bibinfo {eid} {093053} (\bibinfo {year} {2016})},\ \Eprint {https://arxiv.org/abs/1606.02394}{arXiv:1606.02394}\BibitemShut {NoStop}%
\bibitem [{\citenamefont {{Bavaresco}}\ \emph {et~al.}(2019)\citenamefont {{Bavaresco}}, \citenamefont {{Ara{\'u}jo}}, \citenamefont {{Brukner}},\ and\ \citenamefont {{Quintino}}}]{bavaresco19_SDI}%
  \BibitemOpen
  \bibfield  {author} {\bibinfo {author} {\bibfnamefont {J.}~\bibnamefont {{Bavaresco}}}, \bibinfo {author} {\bibfnamefont {M.}~\bibnamefont {{Ara{\'u}jo}}}, \bibinfo {author} {\bibfnamefont {{\v{C}}.}~\bibnamefont {{Brukner}}},\ and\ \bibinfo {author} {\bibfnamefont {M.~T.}\ \bibnamefont {{Quintino}}},\ }\bibfield  {title} {\bibinfo {title} {{Semi-device-independent certification of indefinite causal order}},\ }\href {https://doi.org/10.22331/q-2019-08-19-176} {\bibfield  {journal} {\bibinfo  {journal} {Quantum}\ }\textbf {\bibinfo {volume} {3}},\ \bibinfo {pages} {176} (\bibinfo {year} {2019})},\ \Eprint {https://arxiv.org/abs/1903.10526}{arXiv:1903.10526}\BibitemShut {NoStop}%
\bibitem [{\citenamefont {{Lewandowska}}\ \emph {et~al.}(2023)\citenamefont {{Lewandowska}}, \citenamefont {{Pawela}},\ and\ \citenamefont {{Pucha{\l}a}}}]{lewandowska23discrimination}%
  \BibitemOpen
  \bibfield  {author} {\bibinfo {author} {\bibfnamefont {P.}~\bibnamefont {{Lewandowska}}}, \bibinfo {author} {\bibfnamefont {L.}~\bibnamefont {{Pawela}}},\ and\ \bibinfo {author} {\bibfnamefont {Z.}~\bibnamefont {{Pucha{\l}a}}},\ }\bibfield  {title} {\bibinfo {title} {{Strategies for single-shot discrimination of process matrices}},\ }\href {https://doi.org/10.1038/s41598-023-30191-0} {\bibfield  {journal} {\bibinfo  {journal} {Sci. Rep.}\ }\textbf {\bibinfo {volume} {13}},\ \bibinfo {eid} {3046} (\bibinfo {year} {2023})},\ \Eprint {https://arxiv.org/abs/2210.14575}{arXiv:2210.14575}\BibitemShut {NoStop}%
\bibitem [{\citenamefont {Yokojima}\ \emph {et~al.}(2021)\citenamefont {Yokojima}, \citenamefont {Quintino}, \citenamefont {Soeda},\ and\ \citenamefont {Murao}}]{yokojima20}%
  \BibitemOpen
  \bibfield  {author} {\bibinfo {author} {\bibfnamefont {W.}~\bibnamefont {Yokojima}}, \bibinfo {author} {\bibfnamefont {M.~T.}\ \bibnamefont {Quintino}}, \bibinfo {author} {\bibfnamefont {A.}~\bibnamefont {Soeda}},\ and\ \bibinfo {author} {\bibfnamefont {M.}~\bibnamefont {Murao}},\ }\bibfield  {title} {\bibinfo {title} {Consequences of preserving reversibility in quantum superchannels},\ }\href {https://doi.org/10.22331/q-2021-04-26-441} {\bibfield  {journal} {\bibinfo  {journal} {Quantum}\ }\textbf {\bibinfo {volume} {5}},\ \bibinfo {pages} {441} (\bibinfo {year} {2021})},\ \Eprint {https://arxiv.org/abs/2003.05682}{arXiv:2003.05682}\BibitemShut {NoStop}%
\bibitem [{\citenamefont {{Baumeler}}\ and\ \citenamefont {{Wolf}}(2016)}]{baumeler16_logically}%
  \BibitemOpen
  \bibfield  {author} {\bibinfo {author} {\bibfnamefont {{\"A}.}~\bibnamefont {{Baumeler}}}\ and\ \bibinfo {author} {\bibfnamefont {S.}~\bibnamefont {{Wolf}}},\ }\bibfield  {title} {\bibinfo {title} {{The space of logically consistent classical processes without causal order}},\ }\href {https://doi.org/10.1088/1367-2630/18/1/013036} {\bibfield  {journal} {\bibinfo  {journal} {New J. of Phys.}\ }\textbf {\bibinfo {volume} {18}},\ \bibinfo {eid} {013036} (\bibinfo {year} {2016})},\ \Eprint {https://arxiv.org/abs/1507.01714}{arXiv:1507.01714}\BibitemShut {NoStop}%
\bibitem [{\citenamefont {{Vanrietvelde}}\ \emph {et~al.}(2022)\citenamefont {{Vanrietvelde}}, \citenamefont {{Ormrod}}, \citenamefont {{Kristj{\'a}nsson}},\ and\ \citenamefont {{Barrett}}}]{vanrietvelde22_consistent_circuits}%
  \BibitemOpen
  \bibfield  {author} {\bibinfo {author} {\bibfnamefont {A.}~\bibnamefont {{Vanrietvelde}}}, \bibinfo {author} {\bibfnamefont {N.}~\bibnamefont {{Ormrod}}}, \bibinfo {author} {\bibfnamefont {H.}~\bibnamefont {{Kristj{\'a}nsson}}},\ and\ \bibinfo {author} {\bibfnamefont {J.}~\bibnamefont {{Barrett}}},\ }\bibfield  {title} {\bibinfo {title} {{Consistent circuits for indefinite causal order}},\ }\href {https://www.arXiv.org/abs/2206.10042} {\bibfield  {journal} {\bibinfo  {journal} {arXiv:2206.10042}\ } (\bibinfo {year} {2022})}\BibitemShut {NoStop}%
\bibitem [{\citenamefont {{Branciard}}(2016)}]{branciard16_case_studies}%
  \BibitemOpen
  \bibfield  {author} {\bibinfo {author} {\bibfnamefont {C.}~\bibnamefont {{Branciard}}},\ }\bibfield  {title} {\bibinfo {title} {{Witnesses of causal nonseparability: an introduction and a few case studies}},\ }\href {https://doi.org/10.1038/srep26018} {\bibfield  {journal} {\bibinfo  {journal} {Sci. Rep.}\ }\textbf {\bibinfo {volume} {6}},\ \bibinfo {eid} {26018} (\bibinfo {year} {2016})},\ \Eprint {https://arxiv.org/abs/1603.00043}{arXiv:1603.00043}\BibitemShut {NoStop}%
\bibitem [{\citenamefont {{Nery}}\ \emph {et~al.}(2021)\citenamefont {{Nery}}, \citenamefont {{T. Quintino}}, \citenamefont {{A. Gu{\'e}rin}}, \citenamefont {{Maciel}},\ and\ \citenamefont {{Vianna}}}]{nery21CCDC}%
  \BibitemOpen
  \bibfield  {author} {\bibinfo {author} {\bibfnamefont {M.}~\bibnamefont {{Nery}}}, \bibinfo {author} {\bibfnamefont {M.}~\bibnamefont {{T. Quintino}}}, \bibinfo {author} {\bibfnamefont {P.}~\bibnamefont {{A. Gu{\'e}rin}}}, \bibinfo {author} {\bibfnamefont {T.~O.}\ \bibnamefont {{Maciel}}},\ and\ \bibinfo {author} {\bibfnamefont {R.~O.}\ \bibnamefont {{Vianna}}},\ }\bibfield  {title} {\bibinfo {title} {Simple and maximally robust processes with no classical common-cause or direct-cause explanation},\ }\href {https://doi.org/10.22331/q-2021-09-09-538} {\bibfield  {journal} {\bibinfo  {journal} {Quantum}\ }\textbf {\bibinfo {volume} {5}},\ \bibinfo {pages} {538} (\bibinfo {year} {2021})},\ \Eprint {https://arxiv.org/abs/2101.11630}{arXiv:2101.11630}\BibitemShut {NoStop}%
\bibitem [{\citenamefont {{Quintino}}\ \emph {et~al.}(2019)\citenamefont {{Quintino}}, \citenamefont {{Dong}}, \citenamefont {{Shimbo}}, \citenamefont {{Soeda}},\ and\ \citenamefont {{Murao}}}]{quintino19PRA}%
  \BibitemOpen
  \bibfield  {author} {\bibinfo {author} {\bibfnamefont {M.~T.}\ \bibnamefont {{Quintino}}}, \bibinfo {author} {\bibfnamefont {Q.}~\bibnamefont {{Dong}}}, \bibinfo {author} {\bibfnamefont {A.}~\bibnamefont {{Shimbo}}}, \bibinfo {author} {\bibfnamefont {A.}~\bibnamefont {{Soeda}}},\ and\ \bibinfo {author} {\bibfnamefont {M.}~\bibnamefont {{Murao}}},\ }\bibfield  {title} {\bibinfo {title} {{Probabilistic exact universal quantum circuits for transforming unitary operations}},\ }\href {https://doi.org/10.1103/PhysRevA.100.062339} {\bibfield  {journal} {\bibinfo  {journal} {Phys. Rev.~A}\ }\textbf {\bibinfo {volume} {100}},\ \bibinfo {eid} {062339} (\bibinfo {year} {2019})},\ \Eprint {https://arxiv.org/abs/1909.01366}{arXiv:1909.01366}\BibitemShut {NoStop}%
\bibitem [{\citenamefont {Yoshida}\ \emph {et~al.}(2023{\natexlab{a}})\citenamefont {Yoshida}, \citenamefont {Soeda},\ and\ \citenamefont {Murao}}]{yoshida23isometry}%
  \BibitemOpen
  \bibfield  {author} {\bibinfo {author} {\bibfnamefont {S.}~\bibnamefont {Yoshida}}, \bibinfo {author} {\bibfnamefont {A.}~\bibnamefont {Soeda}},\ and\ \bibinfo {author} {\bibfnamefont {M.}~\bibnamefont {Murao}},\ }\bibfield  {title} {\bibinfo {title} {Universal construction of decoders from encoding black boxes},\ }\href {https://doi.org/10.22331/q-2023-03-20-957} {\bibfield  {journal} {\bibinfo  {journal} {{Quantum}}\ }\textbf {\bibinfo {volume} {7}},\ \bibinfo {pages} {957} (\bibinfo {year} {2023}{\natexlab{a}})},\ \Eprint {https://arxiv.org/abs/2110.00258}{arXiv:2110.00258}\BibitemShut {NoStop}%
\bibitem [{\citenamefont {Yoshida}\ \emph {et~al.}(2023{\natexlab{b}})\citenamefont {Yoshida}, \citenamefont {Soeda},\ and\ \citenamefont {Murao}}]{yoshida22_det_exact_inversion}%
  \BibitemOpen
  \bibfield  {author} {\bibinfo {author} {\bibfnamefont {S.}~\bibnamefont {Yoshida}}, \bibinfo {author} {\bibfnamefont {A.}~\bibnamefont {Soeda}},\ and\ \bibinfo {author} {\bibfnamefont {M.}~\bibnamefont {Murao}},\ }\bibfield  {title} {\bibinfo {title} {Reversing {Unknown} {Qubit}-{Unitary} {Operation}, {Deterministically} and {Exactly}},\ }\href {https://doi.org/10.1103/PhysRevLett.131.120602} {\bibfield  {journal} {\bibinfo  {journal} {Phys. Rev. Lett.}\ }\textbf {\bibinfo {volume} {131}},\ \bibinfo {pages} {120602} (\bibinfo {year} {2023}{\natexlab{b}})},\ \Eprint {https://arxiv.org/abs/2209.02907}{arXiv:2209.02907}\BibitemShut {NoStop}%
\bibitem [{\citenamefont {{Bavaresco}}\ \emph {et~al.}(2022)\citenamefont {{Bavaresco}}, \citenamefont {{Murao}},\ and\ \citenamefont {{Quintino}}}]{bavaresco21b}%
  \BibitemOpen
  \bibfield  {author} {\bibinfo {author} {\bibfnamefont {J.}~\bibnamefont {{Bavaresco}}}, \bibinfo {author} {\bibfnamefont {M.}~\bibnamefont {{Murao}}},\ and\ \bibinfo {author} {\bibfnamefont {M.~T.}\ \bibnamefont {{Quintino}}},\ }\bibfield  {title} {\bibinfo {title} {{Unitary channel discrimination beyond group structures: Advantages of sequential and indefinite-causal-order strategies}},\ }\href {https://doi.org/10.1063/5.0075919} {\bibfield  {journal} {\bibinfo  {journal} {J. Math. Phys.}\ }\textbf {\bibinfo {volume} {63}},\ \bibinfo {eid} {042203} (\bibinfo {year} {2022})},\ \Eprint {https://arxiv.org/abs/2105.13369}{arXiv:2105.13369}\BibitemShut {NoStop}%
\bibitem [{\citenamefont {Quintino}(2023)}]{MTQ_github_HigherOrderProjectors}%
  \BibitemOpen
  \bibfield  {author} {\bibinfo {author} {\bibfnamefont {M.~T.}\ \bibnamefont {Quintino}},\ }\href {https://github.com/mtcq/HigherOrderProjectors} {\bibinfo {title} {https://github.com/mtcq/higherorderprojectors}} (\bibinfo {year} {2023})\BibitemShut {NoStop}%
\bibitem [{\citenamefont {Gour}(2019)}]{gour_comparison_2019}%
  \BibitemOpen
  \bibfield  {author} {\bibinfo {author} {\bibfnamefont {G.}~\bibnamefont {Gour}},\ }\bibfield  {title} {\bibinfo {title} {Comparison of {Quantum} {Channels} by {Superchannels}},\ }\href {https://doi.org/10.1109/TIT.2019.2907989} {\bibfield  {journal} {\bibinfo  {journal} {IEEE Trans. Inf. Theory}\ }\textbf {\bibinfo {volume} {65}},\ \bibinfo {pages} {5880--5904} (\bibinfo {year} {2019})},\ \Eprint {https://arxiv.org/abs/1808.02607}{arXiv:1808.02607}\BibitemShut {NoStop}%
\bibitem [{\citenamefont {Gour}\ \emph {et~al.}(2024)\citenamefont {Gour}, \citenamefont {Kim}, \citenamefont {Nateeboon}, \citenamefont {Shemesh},\ and\ \citenamefont {Yoeli}}]{gour_inevitable_2024}%
  \BibitemOpen
  \bibfield  {author} {\bibinfo {author} {\bibfnamefont {G.}~\bibnamefont {Gour}}, \bibinfo {author} {\bibfnamefont {D.}~\bibnamefont {Kim}}, \bibinfo {author} {\bibfnamefont {T.}~\bibnamefont {Nateeboon}}, \bibinfo {author} {\bibfnamefont {G.}~\bibnamefont {Shemesh}},\ and\ \bibinfo {author} {\bibfnamefont {G.}~\bibnamefont {Yoeli}},\ }\bibfield  {title} {\bibinfo {title} {Inevitable {Negativity}: {Additivity} {Commands} {Negative} {Quantum} {Channel} {Entropy}},\ }\href@noop {} {\bibfield  {journal} {\bibinfo  {journal} {arXiv e-prints}\ } (\bibinfo {year} {2024})},\ \Eprint {https://arxiv.org/abs/2406.13823}{arXiv:2406.13823 [quant-ph]}\BibitemShut {NoStop}%
\bibitem [{\citenamefont {Kissinger}\ and\ \citenamefont {Uijlen}(2019)}]{Kissinger2017Categorical}%
  \BibitemOpen
  \bibfield  {author} {\bibinfo {author} {\bibfnamefont {A.}~\bibnamefont {Kissinger}}\ and\ \bibinfo {author} {\bibfnamefont {S.}~\bibnamefont {Uijlen}},\ }\bibfield  {title} {\bibinfo {title} {{A categorical semantics for causal structure}},\ }\href {https://doi.org/10.23638/LMCS-15(3:15)2019} {\bibfield  {journal} {\bibinfo  {journal} {{Log. Methods Comput. Sci.}}\ }\textbf {\bibinfo {volume} {15}},\ \bibinfo {pages} {3} (\bibinfo {year} {2019})},\ \Eprint {https://arxiv.org/abs/1701.04732}{arXiv:1701.04732}\BibitemShut {NoStop}%
\end{thebibliography}%
\appendix
\section{Thm. \ref{thm::genFormProjOp} for \texorpdfstring{$\gamma_\inp = 0$}{}}
\label{app::gam0}
For the case $\gamma_\inp= 0$, the affine constraint $\tr[W] = \gamma_\inp = 0$ on $\Scal_\inp$ becomes a linear one, making $\Scal_\inp$ a vector space (in the absence of the positivity constraint on the elements on $\Scal_\inp$, that is). Since the mapping $\map{T}_{\inp\out}$ is linear, and we demand that $\map{T}_{\inp\out}[W] \in \Scal_\out$ for all $W \in \Scal_\inp$, we see that the only possibility for $\gamma_\out = \tr[\map{T}_{\inp\out}[W]]$ is $\gamma_\out = 0$, making $\Scal_\out$ a vector space as well. For any other choice of $\gamma_\out$, the desired mapping $\map{T}_{\inp\out}$ does not exist when $\gamma_\inp = 0$.  

Now, to derive the properties of $\map{T}_{\inp\out}$ (for the case when it exists), we first define a projector onto $\Scal_\inp$. Since $\Scal_\inp$ is a vector space, such a projector $\map{P}_\inp'$ always exists, and $W \in \Scal_\inp$ is equivalent to $\map{P}_\inp'[W] = W$. More concretely, we have $\Scal_\inp = \text{span}(\{\map{P}_\inp[X]|\tr[\map{P}_\inp[X]] = 0\})$. There exists a Hermitian orthogonal basis $\{\sigma_\alpha\}$ of $\Scal_\inp$ with $\tr[\sigma_\alpha \sigma_\beta] = \delta_{\alpha\beta}$, and the projector on $\Scal_\inp$ is given by 
\begin{gather}
    \map{P}'_\inp[X] = \sum_\alpha \tr[\sigma_\alpha X] \sigma_\alpha.
\end{gather}
We emphasise that, in general, $\map{P}'_\inp[X] \neq \map{P}_\inp \circ \map{N}_\inp$ and $\map{P}'_\inp[X] \neq  \map{N}_\inp \circ \map{P}_\inp$ -- where  $\map{N}_\inp$ is the projector onto the vector space of traceless matrices -- since $\map{P}_\inp$ and $\map{N}_\inp$ do not necessarily commute. 

With these preliminary definitions out of the way, following the same argument that led to the proof of Thm.~\ref{thm::genFormProjOp} it is easy to see that 
\begin{gather}
    \map{T}_{\inp\out}[W] \in \Scal_\out \ \forall W \in \Scal_\inp \quad \Leftrightarrow \quad \map{P}_\out \circ \map{T}_{\inp\out} \circ \map{P}'_\inp = \map{T}_{\inp\out} \circ \map{P}_\inp' \quad \text{and} \quad \tr\circ \, \map{T}_{\inp\out} \circ \map{P}_\inp' = 0.
\end{gather}
Finally, let us comment on the additional positivity constraint one would generally impose on the elements of the sets $\Scal_\inp$ and $\Scal_\out$. Whenever $\gamma_\inp, \gamma_\out \neq 0$, positivity has no impact on the respective spans of $\Scal_\inp$ and $\Scal_\out$ and thus no influence on the properties of $\map{T}_{\inp\out}$ (beyond the complete positivity constraint, which we impose either way whenever dealing with actual quantum objects). In contrast, for the case that the elements of $\Scal_\inp$ and $\Scal_\out$ are traceless, positivity yields an actual simplification, since the only traceless positive semidefinite matrix is the zero matrix $\mathbf{0}$. For this case then, we would have $\Scal_i = \{\mathbf{0}\}$ and $\Scal_\out = \{\mathbf{0}\}$, which would make \textit{any} linear map (that maps between the correct spaces $\L(\H_\inp)$ and $\L(\H_\out)$) admissible, since all linear maps map the zero element to the zero element.

\section{Thms.~\ref{thm::genFormProj} and~\ref{thm::FullyGenProj} for \texorpdfstring{$\gamma_\inp = 0$}{}}
\label{app::gam0Choi}
As was the case in App.~\ref{app::gam0}, for the case $\gamma_\inp = 0$, the affine constraint on $\Scal_\inp$ becomes a linear one, and, due to the linearity of $\map{T}_{\inp\out}$, the only case we have to consider is $\gamma_\out = 0$. In this case, following the same logic as in App.~\ref{app::gam0}, we can define a linear projector $\map{P}_\inp'$ onto a vector space of traceless objects, such that $W \in \Scal_\inp$ iff $\map{P}_\inp'[W] = W$. A priori, this projector does not have to be self-adjoint and unital, or commute with the transposition, even if the original projector $\map{P}_\inp$ did. Consequently, using the same steps that led to the proof of Thm.~\ref{thm::FullyGenProj}, we see that $\map{T}_{\inp\out}[W] \in \Scal_\out \ \forall W \in \Scal_\inp $ iff  
\begin{gather}
\label{eqn::AppGam0Choi}
   T_{\inp\out} = T_{\inp\out} - \map{P}_\inp^{\prime\tau} \otimes \map{\id}_\out[T_{\inp\out}] + \map{P}_\inp^{\prime\tau} \otimes \map{P}_\out [T_{\inp\out}]' \quad \text{and} \quad \map{P}_\textup{\inp}^{\prime\tau}[(\tr_\textup{\out}T_{\textup{\inp\out}})] = 0.
\end{gather}
Analogously, one could also define a linear projector $\map{P}_\out'$ onto a space of traceless matrices, such that $W' \in \Scal_\out$ if and only if $W' = \map{P}_\out'[W'] = W'$ (i.e., the projector directly incorporates the requirement $\gamma_\out = 0$), and re-write Eq.~\eqref{eqn::AppGam0Choi} equivalently as 
\begin{gather}
\map{T}_{\inp\out}[W] \in \Scal_\out \ \forall W \in \Scal_\inp \quad \Leftrightarrow \quad T_{\inp\out} = T_{\inp\out} - (\map{P}_\inp^{\prime\tau} \otimes \map{\id}_\out)[T_{\inp\out}] + (\map{P}_\inp^{\prime\tau} \otimes \map{P}_\out') [T_{\inp\out}]'\, ,
\end{gather}
which is in line with the results of Thm.~\ref{thm::TransChoiGen}, where transformations between linear spaces defined by general projectors were characterised.

\section{Properties of linear maps}
\label{app::LinMapsProp}
 To prove the first statement of Lem.~\ref{lem::PropMap}, that self-adjoint maps $\map{P}$ are Hermiticity preserving, note that for any Hermitian matrix $H \in \L(\H)$ we have 
\begin{gather}
    \braket{i}{\widetilde P[H]|j} = \tr[\ketbra{j}{i} \widetilde P[H]] = \tr[\widetilde P[\ketbra{i}{j}]^\dagger H] = \tr[\widetilde P[\ketbra{i}{j}] H]^* = \braket{j}{ \widetilde{P}[H]^\dagger |i}^*\, ,
\end{gather}
i.e., $\widetilde P[H]$ is Hermitian whenever $H$ is Hermitian. Similarly, if, in addition, $\widetilde P$ is unital, then we have for arbitrary matrices $M\in \L(\H)$
\begin{gather}
\tr[\widetilde P[M]] = \tr[\widetilde P[\mathbbm{1}]M) = \tr[M]\, ,
\end{gather}
i.e., $\widetilde P$ is trace preserving. To prove the third statement, that for self-adjoint maps, commutation with the transposition is equivalent to commutation with complex conjugation (in the same basis), let us first show that this holds for Hermitian matrices  $H\in \L(\H)$. In this case we have 
\begin{gather}
 \widetilde P[H]^* = \widetilde P[H^\tau]^\dagger = \widetilde P[H^\tau] = \widetilde P[H^*]\, ,
\end{gather}
where we used commutation with the transposition for the first equality, and Hermiticity preservation for the second one. Now, any matrix $M$ can be written as $M = H + iH'$, where both $H$ and $H'$ are Hermitian. Consequently, using the linearity of $\widetilde P$, we obtain
\begin{gather}
 \widetilde P[M]^* = (\widetilde P[H] + i\widetilde P[H'])^* =  (\widetilde P[H^*] - i\widetilde P[(H')^*]) = \widetilde P[H^* - i (H')^*] = \widetilde P[M^*]\, .
\end{gather}
The proof of the converse direction follows along the same lines. For the fourth statement, if $\widetilde P$ is self-adjoint and commutes with the transposition, then for all matrices $M', M \in \L(\H)$ we have 
\begin{gather}
 \tr[M'\widetilde P[M]] = \tr[\widetilde P[M^{\prime\dagger}]^\dagger M] = \tr[\widetilde P[M'] M]\, ,
\end{gather}
where we have used the self-adjointness of $\widetilde P$ and the fact that linear maps that commute with the transposition also commute with complex conjugation (and thus with the Hermitian conjugate). Finally, if $\widetilde P$ is Hermiticity preserving, then we have 
\begin{gather}
\tr[H' \widetilde P[H]] = \tr[\widetilde P[H^{\prime \dagger}]^\dagger H] =  \tr[\widetilde P[H^{\prime}] H]\, ,
\end{gather}
for all Hermitian matrices $H', H \in \L(\H)$ (in fact, Hermiticity of $H'$ would have sufficed).\footnote{Note that, since most maps in quantum mechanics preserve Hermiticity and the matrices that are considered are Hermitian, one can often find the definition of an adjoint map as $\tr[M'\widetilde P[M]] = \tr[\widetilde P^\dagger[M']M]$ and self-adjointness via $\tr[M'\widetilde P[M]] = \tr[\widetilde P[M']M]$ in the quantum information literature. }

\section{Proof of Lem.~\ref{lem::CompStrucPresTransf}}
\label{app::Complete}
Here, using the link product approach discussed in Sec.~\ref{sec::ProofLink}, we prove Lem.~\ref{lem::CompStrucPresTransf}, i.e., we derive the structural properties on transformations $\map{T}_{\inp\out}$ that are completely admissible with respect to the extensions $\map{P}_{\inp\att}$ and $\map{P}_{\out\att}$. First, let the sets $\Scal_{\inp\att}$ and $\Scal_{\out\att}$ be defined via
\begin{align*}
    &W_{\inp\att} \in \Scal_{\inp\att} \quad \Leftrightarrow \quad \map{P}_{\inp\att}[W_{\inp\att}] = W_{\inp\att} \quad \textup{\&} \quad \tr[W_{\inp\att}] = \gamma_{\inp\att} \\
    \text{and} \quad  &W_{\out\att} \in \Scal_{\out\att} \quad \Leftrightarrow \quad \map{P}_{\out\att}[W_{\out\att}] = W_{\out\att} \quad \textup{\&} \quad \tr[W_{\out\att}] = \gamma_{\out\att}
\end{align*}
In order for $\map{T}_{\inp\out}$ to be completely admissible with respect to $\map{P}_{\inp\att}$ and $\map{P}_{\out\att}$, it has to satisfy 
\begin{gather}
\label{eqn::appComp}
\map{P}_{\out\att}[T_{\inp\out} \star W_{\inp\att}] = T_{\inp\out} \star W_{\inp\att} \quad \text{and} \quad \tr[T_{\inp\out} \star W_{\inp\att}] = \gamma_{\out\att} \quad \forall W_{\inp\att} \in \Scal_{\inp\att}. 
\end{gather}
By introducing the identity channel $\map{\id}_{\att\to\att'}$ with corresponding Choi state $\Phi_{\att\att'}^+$, where $\H_\att \cong \H_{\att'}$, the first of these two conditions can equivalently be phrased as 
\begin{gather}
    \map{P}_{\out\att'}[(T_{\inp\out} \otimes \Phi^+_{\att\att'}) \star \map{P}_{\inp\att}[M_{\inp\att}]] = (T_{\inp\out} \otimes \Phi_{\att\att'}^+) \star \map{P}_{\inp\att}[M_{\inp\att}] \quad \forall M_{\inp\att} \in \L(\H_{\inp\att}), 
\end{gather}
where $\map{P}_{\out\att'}: \L(\H_\out \otimes \H_{\att'}) \to  \L(\H_\out \otimes \H_{\att'})$ and $\map{P}_{\out\att'} \cong \map{P}_{\out\att}$. Now, the operator $T_{\inp\out} \otimes \Phi^+_{\att\att'}$ acts non-trivially on all of $W_{\inp\att}$, thus allowing one to deduce its properties in the same way as was done for the proof of Thm.~\ref{thm::FullyGenProj} (see Sec.~\ref{sec::ProofLink}): Since the projector $\map{P}_{\inp\att}$ is assumed to be self-adjoint, and commutes with the transposition, it can be `moved around' in the link product, such that the above equation amounts to
\begin{gather}
     (\map{P}_{\inp\att} \otimes \map{P}_{\out\att'})[T_{\inp\out} \otimes \Phi^+_{\att\att'}] = (\map{P}_{\inp\att} \otimes \map{\id}_{\out\att'})[T_{\inp\out} \otimes \Phi_{\att\att'}^+],
\end{gather} 
which is exactly Eq.~\eqref{eqn::CompStrucPres1} from Lem.~\ref{lem::CompStrucPresTransf}. A priori, there is a freedom in the sense that the explicit form of the Choi state $\Phi^+_{\att\att'}$ of the identity channel $\map{\id}_{\att\to\att'}$ depends on the choice of basis with respect to which the Choi isomorphism is carried out. While the above equation holds independent of the respective choice, its concrete form (i.e., what basis to choose for the Choi isomorphism) will always be clear from context/the concrete physical setup that is considered.

Since $\map{P}_{\inp\att}$ is self-adjoint and unital, it is also trace preserving, such that, in a similar vein to the previous consideration, we see that the second part of Eq.~\eqref{eqn::appComp} is equivalent to
\begin{gather}
    (\map{P}_{\inp\att} \otimes \map{\id_{\out}})[\tr_\out[T_{\inp\out}] \otimes \id_{\att}] \star W_{\inp\att} = \gamma_{\out\att} \quad \forall W_{\inp\att} \in \L(\H_\inp \otimes \H_\out), \ \tr[W_{\inp\att}] = \gamma_{\inp\att}.
\end{gather}
This implies $(\map{P}_{\inp\att} \otimes \map{\id_{\out}})[\tr_\out[T_{\inp\out}] \otimes \id_{\att}] = \gamma_{\out\att}/\gamma_{\inp\att}\cdot\id_{\inp\att}$. Multiplying both sides by $\id_{\out}/d_{\out}d_\att$, we see that the above equation coincides with  
\begin{gather}
(\map{P}_{\inp\att} \otimes \map{\id}_{\out}) [{}_{\out\att}T_{\inp\out}] = {}_{\inp\out\att}T_{\inp\out}\quad \text{\&} \quad \tr[T_{\inp\out}] = \frac{\gamma_{\out\att}}{\gamma_{\inp\att}},    
\end{gather}
which is exactly Eq.~\eqref{eqn::CompStrucPres2} from Lem.~\ref{lem::CompStrucPresTransf}.

\section{Proof of Lem.~\ref{lem::SuffComp}}
\label{app::proofSuffComp}
To prove Lem.~\ref{lem::SuffComp}, let us assume that $\map{P}_{\inp\att} = \map{P}_{\inp} \otimes \map{P}_\att$ and $\map{P}_{\out\att} = \map{P}_{\out} \otimes \map{P}_\att$ for all $\H_\att$ and projectors $\map{P}_\att$ that commute with the transposition. Additionally, let $\widetilde T_{\inp\out}$ be admissible for $\H_\att = \mathbbm{C}$, i.e., 
\begin{gather}
    (\map{P}_{\inp} \otimes \map{\id}_{\out})[T_{\inp\out}] = (\map{P}_{\inp} \otimes \map{P}_{\out})[T_{\inp\out}]. 
\end{gather}
Now, tensoring both sides of this equation with $\Phi_{\att\att'}^+$ and applying $\map{P}_{\att} \otimes \map{P}_{\att'}$ (where $\map{P}_{\att'} \cong \map{P}_{\att}$) yields
\begin{gather}
 (\map{P}_\inp \otimes \map{P}_\att \otimes \map{\id}_{\out} \otimes \map{P}_{\att'})[T_{\inp\out} \otimes \Phi^{+}_{\att\att'}] = (\underbrace{\map{P}_\inp \otimes \map{P}_{\att}}_{=\map{P_{\inp\att}}} \otimes \underbrace{\map{P}_\out \otimes \map{P}_{\att'}}_{=\map{P}_{\out\att'}})[T_{\inp\out} \otimes \Phi^{+}_{\att\att'}]
\end{gather}
Now, it is easy to see that $(\map{\id}_\att \otimes \map{P}_{\att'})[\Phi^{+}_{\att\att'}] = (\map{P}_{\att} \otimes \map{\id}_\att)[\Phi^{+}_{\att\att'}]$ for projectors $\map{P}_{\att'}$ that commute with the transposition. With this, using that $\map{P}_\att^2 = \map{P}_\att$, the above Equation yields 
\begin{gather}
 (\underbrace{\map{P}_\inp \otimes \map{P}_\att}_{=\map{P_{\inp\att}}} \otimes \map{\id}_{\out\att'})[T_{\inp\out} \otimes \Phi^{+}_{\att\att'}] = (\map{P}_{\inp\att} \otimes \map{P}_{\out\att'})[T_{\inp\out} \otimes \Phi^{+}_{\att\att'}], 
\end{gather}
implying (according to Lem.~\ref{lem::CompStrucPresTransf}) that $T_{\inp\out}$ is completely admissible according to the extensions $\map{P}_{\inp\att} = \map{P}_{\inp} \otimes \map{P}_\att$ and $\map{P}_{\out\att} = \map{P}_{\out} \otimes \map{P}_\att$. We emphasise that the additional restriction that $\map{P}_\att$ commutes with the transposition can be dropped, by using the results of Sec.~\ref{sec::GenAppr} to derive a generalisation of Lem.~\ref{lem::CompStrucPresTransf} that also holds for projectors that do not commute with the transposition. 
\end{document}